\documentclass[12pt,a4paper]{amsart}

\usepackage{tikz}

\usepackage{amsmath,amsfonts,amssymb}

\usepackage[hmargin=2cm,vmargin=2cm]{geometry}

\usepackage{hyperref}
\usepackage{nameref,zref-xr}

\newcommand{\mbZ}{\mathbb Z}
\newcommand{\mbC}{\mathbb C}
\newcommand{\cP}{\mathcal P}
\newcommand{\oM}{\overline{\mathcal M}}
\newcommand{\tg}{\widetilde g}
\newcommand{\tu}{{\widetilde u}}
\newcommand{\og}{\overline g}
\newcommand{\oh}{\overline h}
\newcommand{\hLambda}{\widehat\Lambda}

\def\cM{{\mathcal{M}}}
\def\oM{{\overline{\mathcal{M}}}}

\def\mbQ{{\mathbb Q}}
\def\d{{\partial}}

\newcommand{\<}{\left<}
\renewcommand{\>}{\right>}
\newcommand{\eps}{\varepsilon}

\newcommand{\str}{\mathrm{str}}

\newcommand{\cA}{\mathcal A}
\newcommand{\hcA}{\widehat{\mathcal A}}
\newcommand{\DR}{\mathrm{DR}}
\newcommand{\DZ}{\mathrm{DZ}}

\newcommand{\even}{\mathrm{even}}
\newcommand{\ct}{\mathrm{ct}}

\renewcommand{\th}{\widetilde h}

\newcommand{\Coef}{\mathrm{Coef}}

\newcommand{\Desc}{\mathrm{Desc}}

\newcommand{\ST}{\mathrm{ST}}
\newcommand{\tv}{\widetilde v}
\renewcommand{\top}{\mathrm{top}}
\newcommand{\red}{\mathrm{red}}

\newcommand{\cQ}{\mathcal{Q}}

\newcommand{\Ch}{\mathrm{Ch}}
\newcommand{\dil}{\mathrm{dil}}

\newcommand{\MST}{\mathrm{MST}}
\newcommand{\AMST}{\mathrm{AMST}}
\newcommand{\gl}{\mathrm{gl}}

\newcommand{\oG}{{\overline G}}
\newcommand{\oH}{{\overline H}}
\newcommand{\tC}{\widetilde C}
\newcommand{\Li}{\mathrm{Li}}
\newcommand{\cD}{{\mathcal{D}}}
\newcommand{\wk}{\mathrm{wk}}
\newcommand{\rt}{\mathrm{rt}}
\renewcommand{\t}{\mathrm{t}}
\newcommand{\pol}{\mathrm{pol}}

\newcommand{\exc}{\mathrm{exc}}
\newcommand{\reg}{\mathrm{reg}}
\newcommand{\AST}{\mathrm{AST}}

\newcommand{\Spl}{\mathrm{Spl}}
\newcommand{\Con}{\mathrm{Con}}
\newcommand{\tR}{\widetilde{R}}

\newtheorem{theorem}{Theorem}[section]
\newtheorem{proposition}[theorem]{Proposition}
\newtheorem{lemma}[theorem]{Lemma}
\newtheorem{corollary}[theorem]{Corollary}
\newtheorem{conjecture}[theorem]{Conjecture}

\theoremstyle{definition}

\newtheorem{definition}[theorem]{Definition}

\newtheorem{remark}[theorem]{Remark}

\usepackage{color}

\def\&{\vspace{-5pt}&}

\numberwithin{equation}{section}

\title{Integrable systems of double ramification type}

\author{Alexandr Buryak}
\address{A.~Buryak:\newline School of Mathematics, University of Leeds, Leeds, LS2 9JT, United Kingdom}
\email{a.buryak@leeds.ac.uk}

\author{Boris Dubrovin}
\address{B.~Dubrovin:\newline SISSA, via Bonomea 265, Trieste 34136, Italy}
\email{dubrovin@sissa.it}

\author{J\'er\'emy Gu\'er\'e}
\address{J.~Gu\'er\'e:\newline Institut Fourier, Universit\'e de Grenoble Alpes,\newline
100 rue des Math\'ematiques, 38610 Gi\`eres, France}
\email{jeremy.guere@univ-grenoble-alpes.fr}

\author{Paolo Rossi}
\address{P.~Rossi:\newline Dipartimento di Matematica ``Tullio Levi-Civita'', Universit\`a degli Studi di Padova,\newline
Via Trieste 63, 35121 Padova, Italy}
\email{paolo.rossi@math.unipd.it}

\begin{document}

\begin{abstract}
In this paper we study various aspects of the double ramification (DR) hierarchy, introduced by the first author, and its quantization. We extend the notion of tau-symmetry to quantum integrable hierarchies and prove that the quantum DR hierarchy enjoys this property. We determine explicitly the genus $1$ quantum correction and, as an application, compute completely the quantization of the $3$- and $4$-KdV hierarchies (the DR hierarchies for Witten's $3$- and $4$-spin theories). We then focus on the recursion relation satisfied by the DR Hamiltonian densities and, abstracting from its geometric origin, we use it to characterize and construct a new family of quantum and classical integrable systems which we call of double ramification type, as they satisfy all of the main properties of the DR hierarchy. In the second part, we obtain new insight towards the Miura equivalence conjecture between the DR and Dubrovin-Zhang hierarchies, via a geometric interpretation of the correlators forming the double ramification tau-function. We then show that the candidate Miura transformation between the DR and DZ hierarchies (which we uniquely identified in our previous paper) indeed turns the Dubrovin-Zhang Poisson structure into the standard form. Eventually, we focus on integrable hierarchies associated with rank-$1$ cohomological field theories and their deformations, and we prove the DR/DZ equivalence conjecture up to genus $5$ in this context. 

\end{abstract}

\maketitle

\tableofcontents

\markboth{A. Buryak, B. Dubrovin, J. Gu\'er\'e, P. Rossi}{Integrable systems of double ramification type}

\section{Introduction}
The Dubrovin-Zhang (DZ) hierarchy \cite{DZ05} is an integrable system of Hamiltonian PDEs associated to any given semisimple cohomological field theory (CohFT). As an important property, it is tau-symmetric and we can then define its partition function as the tau-function of its topological solution. The DZ hierarchy plays a central role in generalizing to any semisimple CohFT the notion, underlying the Witten-Kontsevich theorem \cite{Wit91,Kon92}, which states that the partition function of the CohFT should correspond to the topological tau-function of some integrable Hamiltonian tau-symmetric hierarchy of evolutionary PDEs.\\

The double ramification (DR) hierarchy has been introduced in \cite{Bur15} by the first author and is another integrable system of Hamiltonian PDEs, associated to any given cohomological field theory (CohFT). It does not require any semisimplicity condition and it is also defined for partial CohFTs, satisfying  weaker axioms, see \cite{BDGR16}.
At the heart of its construction lies the double ramification cycle $\DR_g(a_1,\ldots,a_n)$, which is the push-forward to the moduli space of stable curves $\oM_{g,n}$ of the virtual fundamental cycle of the moduli space of rubber stable maps to $\mathbb{P}^1$ relative to $0$ and $\infty$, with ramification profile (orders of poles and zeros) given by $(a_1,\ldots,a_n)\in \mathbb{Z}^n$.\\

We prove in \cite{BDGR16} that the DR hierarchy is also tau-symmetric and we define its partition function as the tau-function of its string solution.
The DR/DZ equivalence conjecture \cite{BDGR16} predicts the existence (and unicity) of a normal Miura transformation under which the partition function of a given CohFT equals the associated DR partition function.
As a consequence, we recover in the semisimple case the original conjecture from \cite{Bur15} that the DR and DZ hierarchies are Miura equivalent.\\

One application of the DR/DZ equivalence conjecture, when proved true, is to give a quantization of any Dubrovin-Zhang hierarchy.
Indeed, the DR hierarchy has a natural quantization, constructed in \cite{BR15} and recalled in Section \ref{section:quantumDRH}.
In this paper, we prove that the quantum DR hierarchy is also tau-symmetric and we define a quantum tau-function.
In the limit when the quantum parameter $\hbar$ tends to zero, we recover results from \cite{BDGR16}.
We also study the first quantum correction in genus $1$ and, as an application, we completely determine the quantum DR hierarchies associated to the Witten's $3$- and $4$-spin theories.\\

One of the most striking property of the quantum DR hierarchy is that it can be recovered recursively from the knowledge of one Hamiltonian, usually denoted $\overline{G}_{1,1}$, via the recursion equations of Theorem \ref{theorem:recursion}, proved in \cite{BR15}.
Conversely, any Hamiltonian $\overline{H}$ compatible with these recursion equations in the sense of Theorem \ref{theorem:recursion->integrability} produces a unique quantum integrable tau-symmetric hierarchy.
An integrable hierarchy obtained in this way is said to be of double ramification type.
As an example, we study the dispersionless quantum deformations of DR type of the Riemann hierarchy and suggest they are in one-to-one correspondence with the DR hierarchies associated with CohFTs of rank $1$.\\

Starting from Section \ref{section:geometric formula}, we go back to the classical DR hierarchy and to the DR/DZ equivalence conjecture.
In Theorem \ref{theorem:main geometric formula} we give a very explicit and geometric formula for the coefficients of the DR partition function, called the DR correlators.
This formula is used in Section \ref{section:Miura for DZ} towards the DR/DZ equivalence conjecture.
More precisely, we prove in Theorem \ref{theorem:Miura for DZ} that the candidate Miura transformation between the two, which we uniquely identified in \cite{BDGR16}, indeed transforms the Hamiltonian operator $K^{\mathrm{DZ}}$ of the DZ hierarchy to the standard operator~$\eta \partial_x$ used in the DR hierarchy, giving a new evidence for the conjecture.\\

To conclude, we give various results about the DR and DZ hierarchies associated to CohFTs of rank $1$.
In particular, we show that the DR hierarchy is a standard deformation of the Riemann hierarchy in the sense of~\cite{DLYZ16} and we prove the existence of a normal Miura transformation that reduces the Dubrovin-Zhang hierarchy to its unique standard form, proving one of the conjectures from~\cite{DLYZ16} about tau-symmetric deformations of the Riemann hierarchy.
Lastly, we prove that the DR/DZ equivalence conjecture holds for rank-$1$ CohFTs at the approximation up to genus $5$.

\subsection{Acknowledgements}

We would like to thank Andrea Brini, Guido Carlet, Rahul Pandharipande, Sergey Shadrin and Dimitri Zvonkine for useful discussions. A. B. has received funding from the European Union's Horizon 2020 research and innovation programme under the Marie Sk\l odowska-Curie grant agreement No. 797635 and was also supported by Grant ERC-2012-AdG-320368-MCSK and Grant RFFI-16-01-00409. J. G. was supported by the Einstein foundation. P.~R.~was partially supported by a Chaire CNRS/Enseignement superieur 2012-2017 grant.

%%%%%%%%%%%%%%%%%%%%%%%%%%%%%%%%%%%%%%%%%%%%%%%%%%%%%%%%%%%%%%%%%%%%%
%%%%%%%%%%%%%%%%%%%%%%%%%%%%%%%%%%%%%%%%%%%%%%%%%%%%%%%%%%%%%%%%%%%%%

\section{Double ramification hierarchy}\label{section:DR hierarchy}

In this section we recall the main definitions and results from~\cite{Bur15,BR14,BR15}. The classical double ramification (DR) hierarchy is a system of commuting Hamiltonians on an infinite dimensional phase space that can be heuristically thought of as the loop space of a fixed vector space. The entry datum for this construction is a cohomological field theory (CohFT) in the sense of Kontsevich and Manin~\cite{KM94} or, more in general, a partial CohFT in the sense of \cite{LRZ} (the definition of a partial CohFT is the same as the one for a CohFT apart from the loop axiom, which is not required in the first). For actual CohFTs (not just partial), in \cite{BR15} a quantization was constructed for the classical double ramification hierarchy, dubbed quantum double ramification (qDR) hierarchy.

\subsection{Formal loop space} 

Let $V$ be an $N$-dimensional vector space and $\eta$ a symmetric bilinear form on it. The loop space of $V$ will be defined somewhat formally by describing its ring of functions. Following~\cite{DZ05} (see also~\cite{Ros10}), let us consider formal variables~$u^\alpha_i$, $\alpha=1,\ldots,N$, $i=0,1,\ldots$, associated to a basis $e_1,\ldots,e_N$ of $V$. Always just at a heuristic level, the variable $u^\alpha:=u^\alpha_0$ can be thought of as the component $u^\alpha(x)$ along $e_\alpha$ of a formal loop $u\colon S^1\to V$, where~$x$ is the coordinate on $S^1$, and the variables $u^\alpha_{x}:=u^\alpha_1, u^\alpha_{xx}:=u^\alpha_2,\ldots$ as its $x$-derivatives. We then define the ring $\cA_N$ of {\it differential polynomials} as the ring of polynomials $f(u^*;u^*_x,u^*_{xx},\ldots)$ in the variables~$u^\alpha_i, i>0$, with coefficients in the ring of formal power series in the variables $u^\alpha=u^\alpha_0$ (when it does not give rise to confusion, we will use the symbol $*$ to indicate any value, in the appropriate range, of the sub or superscript). We can differentiate a differential polynomial with respect to $x$ by applying the operator $\partial_x := \sum_{i\geq 0} u^\alpha_{i+1} \frac{\partial}{\partial u^\alpha_i}$ (in general, we use the convention of sum over repeated Greek indices, but not over repeated Latin indices). Finally, we consider the quotient~$\Lambda_N$ of the ring of differential polynomials first by constants and then by the image of~$\partial_x$, and we call its elements {\it local functionals}. A local functional, that is the equivalence class of a differential polynomial~$f=f(u^*;u^*_x,u^*_{xx},\ldots)$, will be denoted by $\overline{f}=\int f dx$. Let us introduce a grading $\deg u^\alpha_i = i$ and define $\cA^{[k]}_N$ and $\Lambda^{[k]}_N$ as the subspaces of degree~$k$ of $\cA_N$ and of~$\Lambda_N$ respectively.

Differential polynomials and local functionals can also be described using another set of formal variables, corresponding heuristically to the Fourier components $p^\alpha_k$, $k\in\mbZ$, of the functions $u^\alpha=u^\alpha(x)$. Let us, hence, define a change of variables
\begin{gather}\label{eq:u-p change}
u^\alpha_j = \sum_{k\in\mbZ} (i k)^j p^\alpha_k e^{i k x},
\end{gather}
which allows us to express a differential polynomial $f(u;u_x,u_{xx},\ldots)$ as a formal Fourier series in $x$ where the coefficient of $e^{i k x}$ is a power series in the variables $p^\alpha_j$ (where the sum of the subscripts in each monomial in $p^\alpha_j$ equals $k$). Moreover, the local functional~$\overline{f}$ corresponds to the constant term of the Fourier series of $f$.

Let us describe a natural class of Poisson brackets on the space of local functionals. Given an $N\times N$ matrix~$K=(K^{\mu\nu})$ of differential operators of the form $K^{\mu\nu} = \sum_{j\geq 0} K^{\mu\nu}_j \partial_x^j$, where the coefficients $K^{\mu\nu}_j$ are differential polynomials and the sum is finite, we define
$$
\{\overline{f},\overline{g}\}_{K}:=\int\left(\frac{\delta \overline{f}}{\delta u^\mu}K^{\mu \nu}\frac{\delta \overline{g}}{\delta u^\nu}\right)dx,
$$
where we have used the variational derivative $\frac{\delta \overline{f}}{\delta u^\mu}:=\sum_{i\geq 0} (-\partial_x)^i \frac{\partial f}{\partial u^\mu_i}$. Imposing that such bracket satisfies the anti-symmetry and the Jacobi identity will translate, of course, into conditions for the coefficients~$K^{\mu \nu}_j$. An operator that satisfies such conditions will be called Hamiltonian. A standard example of a Hamiltonian operator is given by $\eta \partial_x$. The corresponding Poisson bracket $\{\cdot,\cdot\}_{\eta\d_x}$ will sometimes be denoted just by $\{\cdot,\cdot\}$ when no confusion arises. Such Poisson bracket also has a nice expression in terms of the variables $p^\alpha_k$:
\begin{gather}\label{eq:bracket of p's}
\{p^\alpha_k, p^\beta_j\}_{\eta \partial_x} = i k \eta^{\alpha \beta} \delta_{k+j,0}.
\end{gather}

Finally, we will need to consider extensions $\hcA_N$ and $\hLambda_N$ of the spaces of differential polynomials and local functionals. Introduce a new variable~$\eps$ with $\deg\eps = -1$. Then $\hcA^{[k]}_N$ and $\hLambda^{[k]}_N$ are defined, respectively, as the subspaces of degree~$k$ of $\hcA_N:=\cA_N[[\eps]]$ and of~$\hLambda_N:=\Lambda_N[[\eps]]$. Their elements will still be called differential polynomials and local functionals. We can also define Poisson brackets as above, starting from a Hamiltonian operator $K=(K^{\mu\nu})$, $K^{\mu\nu} = \sum_{i,j\geq 0} (K^{[i]}_j)^{\mu\nu} \eps^i \partial_x^j$, where $(K^{[i]}_j)^{\mu\nu}\in\cA_N$ and $\deg (K^{[i]}_j)^{\mu\nu}=i-j+1$. The corresponding Poisson bracket will then have degree $1$. In the sequel only such Hamiltonian operators will be considered.

A Hamiltonian hierarchy of PDEs is a family of systems of the form
\begin{gather}\label{eq:Hamiltonian system}
\frac{\partial u^\alpha}{\partial \tau_i} = K^{\alpha\mu} \frac{\delta\overline{h}_i}{\delta u^\mu}, \ \alpha=1,\ldots,N ,\ i=1,2,\ldots,
\end{gather}
where $\oh_i\in\hLambda^{[0]}_N$ are local functionals with the compatibility condition $\{\oh_i,\oh_j\}_K=0$, for $i,j\geq 1$. The local functionals~$\oh_i$ are called the {\it Hamiltonians} of the systems~\eqref{eq:Hamiltonian system}.

%%%%%%%%%%%%%%%%%%%%%%%%%%%%%%%%%%%%%%%%%%%%%%%%%%%%%%%%%%%%%%

\subsection{Classical double ramification hierarchy} 

Let $c_{g,n}\colon V^{\otimes n} \to H^{\even}(\oM_{g,n},\mbC)$ be the system of linear maps defining a (possibly partial, in the sense of \cite{LRZ}) cohomological field theory, $V$ its underlying $N$-dimensional vector space, $\eta$ its metric tensor and $e_1\in V$ the unit vector. Let $\psi_i$ be the first Chern class of the line bundle over~$\oM_{g,n}$ formed by the cotangent lines at the $i$-th marked point. Denote by~$\mathbb E$ the rank~$g$ Hodge vector bundle over~$\oM_{g,n}$ whose fibers are the spaces of holomorphic one-forms. Let $\lambda_j:=c_j(\mathbb E)\in H^{2j}(\oM_{g,n},\mbQ)$. The Hamiltonians of the double ramification hierarchy are defined as follows:
\begin{gather}\label{DR Hamiltonians}
\og_{\alpha,d}:=\sum_{\substack{g\ge 0\\n\ge 2}}\frac{(-\eps^2)^g}{n!}\sum_{\substack{a_1,\ldots,a_n\in\mbZ\\\sum a_i=0}}\left(\int_{\oM_{g,n+1}}\DR_g(0,a_1,\ldots,a_n)\lambda_g\psi_1^d c_{g,n+1}(e_\alpha\otimes \otimes_{i=1}^n e_{\alpha_i})\right)\prod_{i=1}^n p^{\alpha_i}_{a_i},
\end{gather}
for $\alpha=1,\ldots,N$ and $d=0,1,2,\ldots$. Here $\DR_g(a_1,\ldots,a_n) \in H^{2g}(\oM_{g,n},\mbQ)$ is the double ramification cycle. If not all of $a_i$'s are equal to zero, then the restriction~$\left.\DR_g(a_1,\ldots,a_n)\right|_{\cM_{g,n}}$ can be defined as the Poincar\'e dual to the locus of pointed smooth curves~$[C,p_1,\ldots,p_n]$ satisfying $\mathcal O_C\left(\sum_{i=1}^n a_ip_i\right)\cong\mathcal O_C$, and we refer the reader, for example, to~\cite{BSSZ15} for the definition of the double ramification cycle on the whole moduli space~$\oM_{g,n}$. We will often consider the Poincar\'e dual to the double ramification cycle~$\DR_g(a_1,\ldots,a_n)$. It is an element of $H_{2(2g-3+n)}(\oM_{g,n},\mbQ)$ and, abusing our notations a little bit, it will also be denoted by $\DR_g(a_1,\ldots,a_n)$. In particular, the integral in~\eqref{DR Hamiltonians} will often be written in the following way:
\begin{gather}\label{DR integral}
\int_{\DR_g(0,a_1,\ldots,a_n)}\lambda_g\psi_1^d c_{g,n+1}(e_\alpha\otimes\otimes_{i=1}^n e_{\alpha_i}).
\end{gather}

The expression on the right-hand side of~\eqref{DR Hamiltonians} can be uniquely written as a local functional from $\hLambda_N^{[0]}$ using the change of variables~\eqref{eq:u-p change}. Concretely it can be done in the following way. The integral~\eqref{DR integral} is a polynomial in $a_1,\ldots,a_n$ homogeneous of degree~$2g$. It follows from Hain's formula~\cite{Hai11}, the results of~\cite{MW13} and the fact that $\lambda_g$ vanishes on $\oM_{g,n}\setminus\cM_{g,n}^{\ct}$, where $\cM_{g,n}^{\ct}$ is the moduli space of stable curves of compact type \cite{Mum83,Fab99}. Thus, the integral~\eqref{DR integral} can be written as a polynomial
\begin{gather*}
P_{\alpha,d,g;\alpha_1,\ldots,\alpha_n}(a_1,\ldots,a_n)=\sum_{\substack{b_1,\ldots,b_n\ge 0\\b_1+\ldots+b_n=2g}}P_{\alpha,d,g;\alpha_1,\ldots,\alpha_n}^{b_1,\ldots,b_n}a_1^{b_1}\ldots a_n^{b_n}.
\end{gather*}
Then we have
$$
\og_{\alpha,d}=\int\sum_{\substack{g\ge 0\\n\ge 2}}\frac{\eps^{2g}}{n!}\sum_{\substack{b_1,\ldots,b_n\ge 0\\b_1+\ldots+b_n=2g}}P_{\alpha,d,g;\alpha_1,\ldots,\alpha_n}^{b_1,\ldots,b_n} u^{\alpha_1}_{b_1}\ldots u^{\alpha_n}_{b_n}dx.
$$
Note that the integral~\eqref{DR integral} is defined only when $a_1+\ldots+a_n=0$. Therefore the polynomial~$P_{\alpha,d,g;\alpha_1,\ldots,\alpha_n}$ is actually not unique. However, the resulting local functional $\og_{\alpha,d}\in\hLambda_N^{[0]}$ doesn't depend on this ambiguity (see~\cite{Bur15}). In fact, in \cite{BR14}, a special choice of differential polynomial densities $g_{\alpha,d} \in \hcA^{[0]}_N$ for $\og_{\alpha,d} = \int g_{\alpha,d} \ dx$ is selected. They are defined in terms of $p$-variables as
$$
g_{\alpha,d}:=\sum_{\substack{g\ge 0,\,n\ge 1\\2g-1+n>0}}\frac{(-\eps^2)^g}{n!}\sum_{\substack{a_0,\ldots,a_n\in\mbZ\\\sum a_i=0}}\left(\int_{\DR_g(a_0,a_1,\ldots,a_n)}\lambda_g\psi_1^d c_{g,n+1}(e_\alpha\otimes \otimes_{i=1}^n e_{\alpha_i})\right)\prod_{i=1}^n p^{\alpha_i}_{a_i} e^{-i a_0 x},
$$
and converted univocally to differential polynomials using again the change of variables (\ref{eq:u-p change}).\\

The fact that the local functionals~$\og_{\alpha,d}$ mutually commute with respect to the standard bracket~$\eta\d_x$ was proved in~\cite{Bur15} for CohFTs and in~\cite{BDGR16} for partial CohFTs. The system of local functionals $\og_{\alpha,d}$, for $\alpha=1,\ldots,N$, $d=0,1,2,\ldots$, and the corresponding system of Hamiltonian PDEs with respect to the standard Poisson bracket~$\{\cdot,\cdot\}_{\eta\partial_x}$,
$$
\frac{\d u^\alpha}{\d t^\beta_q}=\eta^{\alpha\mu}\d_x\frac{\delta\og_{\beta,q}}{\delta u^\mu},
$$
is called the \emph{double ramification hierarchy}.

%%%%%%%%%%%%%%%%%%%%%%%%%%%%%%%%%%%%%%%%%%%%%%

\subsection{Quantum Hamiltonian systems}\label{section:quantum hamiltonian systems}

We will need, first, to extend the space of differential polynomials to allow for dependence on the quantization formal parameter $\hbar$. A \emph{quantum differential polynomial} $f=f(u^*,u^*_x,u^*_{xx,}\ldots;\eps,\hbar)$ is a formal power series in $\hbar$ and $\epsilon$ whose coefficients are polynomials in $u^\alpha_k$, for $k>0$, and power series in $u^\alpha_0$, where $\alpha=1,\ldots,N$. The quantization parameter has degree $\deg\hbar=-2$ and all other formal variables retain the same degree as in the classical case. The space of quantum differential polynomials will be denoted by $\hcA^\hbar_N$. The space of \emph{quantum local functionals} $\hLambda^\hbar_N$ is given, as in the classical case, by taking the quotient of $\hcA_N^\hbar$ with respect to formal power series in $\eps$ and $\hbar$ and the image of the $\partial_x$-operator.\\

As in the classical case, the change of variables
$$
u^\alpha_j=\sum_{k\in\mbZ}(ik)^jp^\alpha_k e^{ikx},
$$
allows to express any quantum differential polynomial $f=f(u^*_*;\eps,\hbar)$ as a formal Fourier series in~$x$ with coefficients that are (power series in~$\eps$ with coefficients) in the Weyl algebra
$\mbC[p^1_{k>0},\ldots,p^N_{k>0}][[p^1_{k\leq 0},\ldots,p^N_{k\leq 0},\hbar]]$ endowed with the ``normal ordering'' $\star$-product
$$
f \star g =f \left(  e^{\sum_{k>0} i \hbar k \eta^{\alpha \beta} \overleftarrow{\frac{\partial }{\partial p^\alpha_{k}}} \overrightarrow{\frac{\partial }{\partial p^\beta_{-k}}}}\right) g.
$$
and the commutator $[f,g]:=f\star g - g \star f$.\\

These structures can then be translated to the language of differential polynomials and local functionals. In \cite{BR15} it was proved that, for any two differential polynomials $f(x)=f(u^*,u^*_x,u^*_{xx},\ldots;\eps,\hbar)$ and $g(y)=g(u^*,u^*_y,u^*_{yy},\ldots;\eps,\hbar)$, we have
\begin{equation*}
f(x)\star g(y) =\sum_{\substack{n\geq 0\\ r_1,\ldots,r_n\geq 0\\ s_1,\ldots , s_n\geq 0}} \frac{\hbar^{n}}{n!} \frac{\partial^n f}{\partial u^{\alpha_1}_{s_1}\ldots \partial u^{\alpha_n}_{s_n}}(x)\left( \prod_{k=1}^n (-1)^{r_k}  \eta^{\alpha_k\beta_k} \delta_+^{(r_k + s_k +1)}(x-y) \right)  \frac{\partial^n g}{\partial u^{\beta_1}_{r_1}\ldots \partial u^{\beta_n}_{r_n}}(y), 
\end{equation*}
where $\delta_+^{(s)}(x-y):= \sum_{k\geq 0} (ik)^s e^{i k (x-y)}$, $s\geq 0$, is the positive frequency part of the $s$-th derivative of the Dirac delta distribution $\delta(x-y)= \sum_{k\in \mbZ} e^{i k (x-y)}$
and
\begin{equation}
\begin{split}
[f(x),g(y)]=\sum_{\substack{n\geq 1\\ r_1,\ldots ,r_n\geq 0\\ s_1,\ldots,s_n\geq 0}}& \frac{(-i)^{n-1} \hbar^{n}}{n!}  \frac{\partial^n f}{\partial u^{\alpha_1}_{s_1}\ldots \partial u^{\alpha_n}_{s_n}}(x)  (-1)^{\sum_{k=1}^n r_k}  \left( \prod_{k=1}^n \eta^{\alpha_k \beta_k}\right) \times\\
& \times \sum_{j=1}^{2n-1+\sum_{k=1}^n (s_k+r_k)} C_j^{s_1+r_1+1,\ldots,s_n+r_n+1}\delta^{(j)}(x-y)\frac{\partial^n g}{\partial u^{\beta_1}_{r_1}\ldots \partial u^{\beta_n}_{r_n}}(y).
\end{split}
\end{equation}
where
\begin{gather}\label{eq:relation of coefficients}
C_j^{a_1,\ldots,a_n}=
\begin{cases}
(-1)^{\frac{n-1+\sum a_i-j}{2}}\tC_j^{a_1,\ldots,a_n},&\text{if $j=n-1+\sum_{i=1}^n a_i\ (\mathrm{mod}\ 2)$},\\
0,&\text{otherwise}.
\end{cases}
\end{gather}
and
\begin{gather}\label{eq:decomposition}
\prod_{i=1}^k\Li_{-d_i}(z)=\sum_{j=1}^{k-1+\sum d_i}\tC^{d_1,\ldots,d_k}_j\Li_{-j}(z), \qquad \Li_{-d}(z):=\sum_{k\ge 0}k^d z^k.
\end{gather}
In particular, for $f\in \hcA_N^\hbar$ and $\og \in \hLambda_N^\hbar$, we get
\begin{equation}
\begin{split}
[f,\og]=\sum_{\substack{n\geq 1\\ r_1,\ldots ,r_n\geq 0\\ s_1,\ldots,s_n\geq 0}} \frac{(-i)^{n-1} \hbar^{n}}{n!}  &\frac{\partial^n f}{\partial u^{\alpha_1}_{s_1}\ldots \partial u^{\alpha_n}_{s_n}}  (-1)^{\sum_{k=1}^n r_k}  \left( \prod_{k=1}^n \eta^{\alpha_k \beta_k}\right) \times\\
& \times \sum_{j=1}^{2n-1+\sum_{k=1}^n (s_k+r_k)} C_j^{s_1+r_1+1,\ldots,s_n+r_n+1}  \partial_x^j  \frac{\partial^n g}{\partial u^{\beta_1}_{r_1}\ldots \partial u^{\beta_n}_{r_n}}.
\end{split}
\end{equation}
If $f$ and $\og$ are homogeneous, $[f,\og]$ is a non homogeneous element of $\hcA_N^\hbar$ of top degree equal to $\deg f + \deg \og - 1$. Taking the classical limit of this expression one obtains $\left(\frac{1}{\hbar}[\overline{f},\og]\right)|_{\hbar=0}=\{\overline{f}|_{\hbar=0},\og|_{\hbar=0}\}$, i.e. the standard hydrodynamic Poisson bracket on the classical limit of the local functionals.\\

Notice that, given $\og \in  \hLambda_N^\hbar$, the morphism $[\cdot,\og]:\hcA_N^\hbar\to\hcA_N^\hbar$ is not a derivation of the commutative ring $\hcA_N^\hbar$ (while it is if we consider the non-commutative $\star$-product instead). This means that, while it makes sense to describe the simultaneous evolution along different time parameters $\tau_i$ (in the Heisenberg picture, to use the physical language) of a quantum differential polynomial $f \in \hcA_N^\hbar$ by a system of the form
\begin{gather}\label{eq:quantum Hamiltonian system}
\frac{\partial f}{\partial \tau_i} = \frac{1}{\hbar}[f,\oh_i], \ \alpha=1,\ldots,N ,\ i=1,2,\ldots,
\end{gather}
where $\oh_i\in\hLambda^{[\leq 0]}_N$ are quantum local functionals with the compatibility condition $[\oh_i,\oh_j]=0$, for $i,j\geq 1$, one should refrain from interpreting it as the evolution induced by composition with $\frac{\d u^\alpha}{\d \tau_i}=\frac{1}{\hbar} [u^\alpha,\oh_i]$, as the corresponding chain rule does not hold: $\frac{\d f}{\partial \tau_i} \neq \sum_{k\geq 0}\frac{\d f}{\d u^\alpha_k}\d_x^k \left(\frac{\d u^\alpha}{\d \tau_i}\right)$. This corresponds to the familiar concept that in quantum mechanics there are no trajectories in the phase space along which observables evolve.\\

A formal solution to the system (\ref{eq:quantum Hamiltonian system}) can be written in the form of an element in $\hcA_N^\hbar[[\tau_*]]$:
\begin{gather}\label{eq:quantum solution}
f^{\tau_*}(u^*_*;\eps,\hbar) :=  \exp\left(\sum_{i\geq 1} \frac{\tau_i}{\hbar}[\cdot,\oh_i]\right) f(u^*_*;\eps,\hbar) = \left(\prod_{i\geq 1}  \exp \left( \frac{\tau_i}{\hbar}[\cdot,\oh_i]\right) \right) f(u^*_*;\eps,\hbar) \end{gather}
where
\begin{gather}
\exp\left(\frac{\tau_i}{\hbar}[\cdot,\oh_i]\right) := \sum_{k\geq 0} \frac{\tau_i^k}{\hbar^k k!}[[\ldots [\cdot,\oh_i],\ldots,\oh_i],\oh_i]
\end{gather}
and $f\in \hcA_N^\hbar$ in the right hand side of (\ref{eq:quantum solution}) is interpreted as the initial datum. Lifting the quantum commutator $[\cdot,\cdot]$ to $\hcA^\hbar_N[[\tau_*]]$, it is easy to check that $f^{\tau_*}$ satisfies equation \eqref{eq:quantum Hamiltonian system}. We do insist that $f^{\tau_*}(u^*_*;\eps,\hbar) \neq f((u^*_*)^{\tau_*},\eps,\hbar)$.\\

%%%%%%%%%%%%%%%%%%%%%%%%%%%%%%%%%%%%%%%%%%%%%%

\subsection{Quantum double ramification hierarchy}\label{section:quantumDRH}
Given a cohomological field theory $c_{g,n}\colon V^{\otimes n} \to H^{\even}(\oM_{g,n};\mbC)$, we define the Hamiltonian densities of the \emph{quantum double ramification hierarchy} (qDR) as the following generating series:
\begin{equation}\label{density}
\begin{split}
G_{\alpha,d}:=&\sum_{\substack{g\ge 0,n\ge 0\\2g-1+n>0}}\frac{(i \hbar)^g}{n!}\times\\
&\times\sum_{\substack{a_1,\ldots,a_n\in\mbZ\\ \alpha_1,\ldots,\alpha_n}}\left(\int_{\DR_g\left(-\sum a_i,a_1,\ldots,a_n\right)}\Lambda\left(\frac{-\eps^2}{i \hbar}\right) \psi_1^d c_{g,n+1}\left(e_\alpha\otimes\otimes_{i=1}^n e_{\alpha_i}\right)\right)p^{\alpha_1}_{a_1}\ldots p^{\alpha_n}_{a_n}e^{ix\sum a_i},
\end{split}
\end{equation}
for $\alpha=1,\ldots,N$ and $d=0,1,2,\ldots$. Here $\Lambda\left(\frac{-\eps^2}{i \hbar}\right):=\left(1+ \left( \frac{-\eps^2}{i \hbar}\right) \lambda_1+\ldots + \left(\frac{-\epsilon^2}{i\hbar}\right)^g \lambda_g \right)$, with~$\lambda_i$ the $i$-th Chern class of the Hodge bundle. Notice also that, since $\Lambda(s)$ is itself a cohomological field theory depending on the formal parameter $s$, we could absorb such factor into $c_{g,n+1}\left(e_\alpha\otimes\otimes_{i=1}^n e_{\alpha_i}\right)$ obtaining densities for a CohFT analogue of the Symplectic Field Theory Hamiltonians of \cite{EGH00,FR10}.\\

As for the ``classical'' Hamiltonian densities $g_{\alpha,p}=G_{\alpha,p}|_{\hbar=0}$, we would like to rewrite the above expression in terms of formal jet variables $u^\alpha_s = \sum_{k\in\mbZ} (ik)^s p^\alpha_k e^{ikx}$, $\alpha=1,\ldots,N$, $s=0,1,2,\ldots$. Since the double ramification cycle $\DR_{g}(a_1,\ldots,a_n)$ is a non-homogeneous polynomial of degree at most~$2g$ in the variables $a_1,\ldots,a_n$ (as apparent from Pixton's formula \cite{JPPZ16}), we actually obtain that each $G_{\alpha,p}$ can be uniquely written as a quantum differential polynomial of degree $\deg G_{\alpha,p} \leq 0$ and such that $\deg \left.G_{\alpha,p}\right|_{\hbar=0} = 0$, i.e. $G_{\alpha,p} \in (\hcA^\hbar_N)^{[\leq 0]}$ and $\left.G_{\alpha,p}\right|_{\hbar=0} \in \hcA_N^{[0]}$. This means that the number of $x$-derivatives that can appear in the coefficient of $\eps^k\hbar^j$ is at most ~$k+2j$, and exactly $k$ in the coefficient of $\eps^k \hbar^0$.\\

We finally add manually $N$ extra densities $G_{\alpha,-1}:=\eta_{\alpha\mu} u^\mu$. Recall that by $\oG_{\alpha,p}= \int G_{\alpha,p} dx$ we denote the coefficient of $e^{i0x}$ in $G_{\alpha,p}$ considered also up to a constant, for all $\alpha=1,\ldots,N$, $p=-1,0,1,\ldots$.\\

The fact that the local functionals~$\oG_{\alpha,d}$ mutually commute with respect to the above commutator, $[\oG_{\alpha,p},\oG_{\beta,q}]=0$, was proved in~\cite{BR15} together with the fact that $\oG_{1,0}=\int\left(\frac{1}{2} \eta_{\mu\nu}u^\mu u^\nu\right)dx$, so that, for any $f\in \hcA_N^\hbar$, $\d_{t^1_0} f= \d_x f$.

%%%%%%%%%%%%%%%%%%%%%%%%%%%%%%%%%%%%%%

\subsection{Recursion for the qDR Hamiltonian densities}

We recall some of the properties of the DR hierarchies, in particular a recursion equation, proven in \cite{BR14} for the classical case and in \cite{BR15} for the quantum case, allowing to recover all the Hamiltonian densities $G_{\alpha,p}$, $\alpha=1,\ldots,N$, $p\geq0$, recursively from $G_{\alpha,-1}=\eta_{\alpha\mu} u^\mu$ starting from the knowledge of the functional $\oG_{1,1}$ only.\\

Let us define the following two-point potential for intersection numbers with the double ramification cycle
\begin{equation*}\label{density2}
\begin{split}
G_{\alpha,p;\beta,q}(x,y):=\sum_{\substack{g\ge 0,n\ge 0\\2g+n>0}}\frac{(i\hbar)^g}{n!} \sum_{\substack{a_0,\ldots,a_{n+1}\in\mbZ\\\sum a_i=0\\ \alpha_1,\ldots,\alpha_n}}&\left(\int_{\DR_g\left(a_0,a_1,\ldots,a_n,a_{n+1}\right)}\right. \Lambda\left(\frac{-\eps^2}{i \hbar}\right) \psi_0^p \psi_{n+1}^q \times \\
&\times c_{g,n+2}\left(e_\alpha\otimes\otimes_{i=1}^n e_{\alpha_i}\otimes e_\beta\right)\Bigg) p^{\alpha_1}_{a_1}\ldots p^{\alpha_n}_{a_n}e^{-i a_0 x-i a_{n+1} y},
\end{split}
\end{equation*}
for $\alpha,\beta=1,\ldots,N$ and $p,q=0,1,2,\ldots$.\\

In \cite{BR15} the following result was proven
\begin{lemma}[\cite{BR15}]\label{mainlemma}
For all  $\alpha,\beta=1,\ldots,N$ and $p,q=0,1,2,\ldots$, we have
\begin{equation}\label{eqmainlemma}
\partial_x G_{\alpha,p+1;\beta,q}(x,y) - \partial_y G_{\alpha,p;\beta,q+1}(x,y) =\frac{1}{\hbar} \left[ G_{\alpha,p}(x) , G_{\beta,q}(y)\right]
\end{equation}
\end{lemma}
From this lemma the following theorem can be deduced
\begin{theorem}[\cite{BR15}]\label{theorem:recursion}
For all $\alpha=1,\ldots,N$ and $p=-1,0,1,\ldots$, we have
\begin{gather}\label{eq:first recursion}
\partial_x (D-1) G_{\alpha,p+1} =\frac{1}{\hbar} \left[ G_{\alpha,p} , \oG_{1,1} \right],
\end{gather}
\begin{equation}\label{eq:second recursion}
\partial_x \frac{\partial G_{\alpha,p+1}}{\partial u^\beta} =\frac{1}{\hbar} \left[G_{\alpha,p}, \oG_{\beta,0} \right],
\end{equation}
where $D:=\eps\frac{\partial}{\partial\eps} + 2\hbar\frac{\partial}{\partial \hbar} + \sum_{s\ge 0} u^\alpha_s\frac{\partial}{\partial u^\alpha_s}$.
\end{theorem}
Notice how equation \eqref{eq:first recursion} can be used to recover recursively (up to a constant) $G_{\alpha,p}$, $\alpha=1,\ldots,N$, $p\geq0$ from $G_{\alpha,-1} = \eta_{\alpha,\mu} u^\mu$ and of course the knowledge of $\oG_{1,1}$. From equation \eqref{eq:second recursion} we can instead deduce the string equation (always up to a constant, actually)
\begin{equation}\label{eq:string equation}
\frac{\d G_{\alpha,p+1}}{\d u^1} = G_{\alpha,p}.
\end{equation}
Since we can prove such string equation separately from geometric considerations \cite{BR15}, the constant terms of the densities $G_{\alpha,p}$, that are left undetermined by the recursion \eqref{eq:first recursion}, can then be chosen uniquely as those that verify equation \eqref{eq:string equation}.

%%%%%%%%%%%%%%%%%%%%%%%%%%%%%%%%%%%%%%%
%%%%%%%%%%%%%%%%%%%%%%%%%%%%%%%%%%%%%%%

\section{Quantum double ramification hierarchy in genus $1$}

In \cite{BDGR16} we computed the genus $1$ term of the classical double ramification hierarchy for any cohomological field theory in terms of genus $0$ data. In this section we compute the quantum correction, always in genus $1$ and in terms of genus $0$ data plus the genus $1$ $\mathsf{G}$-function of the CohFT. As an application we compute the full quantum double ramification hierarchies for Witten's $3$- and $4$-spin classes.

\subsection{Genus-$1$ quantum correction}

Let $c_{g,n}:V^{\otimes n} \to H^*(\oM_{g,n},\mbC)$ be a cohomological field theory with $V$ an $N$-dimensional vector space endowed with a non-degenerate metric $\eta$ and basis $e_1,\ldots,e_N$, where $e_1$ is the unit of the CohFT. Let $\oG_{\alpha,d}$, $1\leq \alpha \leq N$, $d\geq -1$ be the corresponding quantum DR Hamiltonians and let $\oG =(D-2)^{-1} \oG_{1,1}$, with $D$ as in Theorem \ref{theorem:recursion}. Let $\og_{\alpha,d}$ and $\og$ be their classical counterparts and $g_{\alpha,d}^{[0]}=g_{\alpha,d}^{[0]}(u^1,\ldots,u^N)$ the genus $0$ Hamiltonian densities.

\begin{theorem}\label{theorem:genus1}
Let $F=F(u^1,\ldots,u^N)$ be the Frobenius potential (genus $0$ potential with no descendants) and $\mathsf{G}=\mathsf{G}(u^1,\ldots,u^N)$ the $\mathsf{G}$-function (genus $1$ potential with no descendants) of the CohFT. Let  $c_{\alpha\beta} = \frac{\d^2 F}{\d u^\alpha \d u^\beta}$, $c_{\alpha\beta\gamma} = \frac{\d ^3 F}{\d u^\alpha \d u^\beta \d u^\gamma}$,  $c_{\alpha\beta\gamma\delta} = \frac{\d ^4 F}{\d u^\alpha \d u^\beta \d u^\gamma \d u^\delta}$ and indices be raised and lowered by the metric $\eta$. Then we have
\begin{equation}\label{eq:genus1G}
\oG = \og + i \hbar \int \left[\left(\frac{1}{48} c_{\alpha \beta \mu}^\mu + \frac{1}{2} c_{\alpha \beta}^\mu  \frac{\d \mathsf{G}}{\d u^\mu}\right) u^\alpha_x u^\beta_x - \frac{1}{24}c^\mu_\mu \right] dx+O(\hbar^2)+O(\hbar \eps^2),
\end{equation}
\begin{equation}\label{eq:genus1Gad}
\begin{split}
\oG_{\alpha,d} = \og_{\alpha,d} + i \hbar \int &\left[\left( \frac{1}{48} \frac{\d^4 g_{\alpha,d}^{[0]}}{\d u^\gamma \d u^\beta \d u^\mu \d u^\nu } \eta^{\mu \nu} + \frac{1}{2} \frac{\d^3 g_{\alpha,d}^{[0]}}{\d u^\gamma \d u^\beta \d u^\mu} \eta^{\mu \nu} \frac{\d \mathsf{G}}{\d u^\nu}  \right. \right.\\
&\hspace{0.5cm}\left.\left. + \frac{1}{2} c_{\gamma \beta}^\mu  \frac{\d}{\d u^\mu} \left( \frac{1}{24}\frac{\d^2 g_{\alpha,d-1}^{[0]}}{\d u^\epsilon \d u^\delta} \eta^{\epsilon \delta}+ \frac{\d g_{\alpha,d-1}^{[0]}}{\d u^\epsilon} \eta^{\epsilon \delta} \frac{\d \mathsf{G}}{\d u^\delta}\right)\right) u^\gamma_x u^\beta_x\right. \\
& \hspace{0.5cm}\left.  -\frac{1}{24} \frac{\d^2 g_{\alpha,d}^{[0]}}{\d u^\mu \d u^\nu} \eta^{\mu \nu}\right] dx + O(\hbar^2)+O(\hbar\eps^2).
\end{split}
\end{equation}
\end{theorem}
\begin{proof}
Let us prove equation (\ref{eq:genus1G}). Recall that
\begin{equation*}
\oG:=\sum_{\substack{g\ge 0,n\ge 1\\2g-2+n>0}}\frac{(i \hbar)^g}{n!}\sum_{\substack{a_1,\ldots,a_n\in\mbZ\\ \alpha_1,\ldots,\alpha_n}}\left(\int_{\DR_g\left(a_1,\ldots,a_n\right)}\Lambda\left(\frac{-\eps^2}{i \hbar}\right) c_{g,n}\left(\otimes_{i=1}^n e_{\alpha_i}\right)\right)p^{\alpha_1}_{a_1}\ldots p^{\alpha_n}_{a_n},
\end{equation*}
so the relevant intersection numbers for the genus $1$ quantum corrections are
$$\int_{\DR_1(a_1,\ldots,a_n)} c_{1,n}\left(\otimes_{i=1}^n e_{\alpha_i}\right).$$
To compute them we use the following formulae for the DR cycle (see \cite{Hai11}), psi and lambda classes on $\oM_{1,n}$,
$$\DR_1(a_1,\ldots,a_n) = \sum_{i=1}^n \frac{\psi_i}{2} a_i^2 - \frac{1}{2} \left(\sum_{\substack{J\subset \{1,\ldots,n\}\\ |J|\geq 2}} \left(\sum_{j \in J}a_j\right)^2 \delta^J_0\right) -\lambda_1,$$
$$\psi_i = \frac{1}{24 }\delta_{\text{irr}}+\sum_{\substack{J\subset \{1,\ldots,n\}\\ |J|\geq 2, i\in J}} \delta^J_0, \qquad \lambda_1 = \frac{1}{24} \delta_{\text{irr}},$$
where $\delta_{\text{irr}}$ and $\delta^J_0$ denote the divisor in $\oM_{1,n}$ of singular curves with a non-separating node and of curves with a separating node whose rational component carries exactly the marked points labeled by $J$ (the points labeled by the complement $J^c$ belonging to the elliptic component), respectively. In particular we get
\begin{equation*}
\begin{split}
&\int_{\DR_1(a_1,\ldots,a_n)} c_{1,n}\left(\otimes_{k=1}^n e_{\alpha_k}\right)=\sum_{i=1}^n \frac{a_i^2}{48} \int_{\oM_{0,n+2}} c_{0,n+2}\left(\otimes_{k=1}^n e_{\alpha_k}\otimes e_\mu \otimes e_\nu\right)\eta^{\mu\nu}\\
& + \frac{1}{2} \sum_{\substack{J\subset \{1,\ldots,n\}\\ |J|\geq 2, i\in J}} a_i^2 \int_{\oM_{0,|J|+1}} c_{0,|J|+1}\left(\otimes_{k\in J} e_{\alpha_{k}}\otimes e_\mu\right) \eta^{\mu \nu} \int_{\oM_{1,n-|J|+1}}c_{1,n-|J|+1}(\otimes_{k\in J^c} e_{\alpha_k} \otimes e_\nu) \\
&-\frac{1}{2} \sum_{\substack{J\subset \{1,\ldots,n\}\\ |J|\geq 2}} (\sum_{j\in J}
a_j)^2 \int_{\oM_{0,|J|+1}} c_{0,|J|+1}\left(\otimes_{k\in J} e_{\alpha_{k}}\otimes e_\mu\right) \eta^{\mu \nu} \int_{\oM_{1,n-|J|+1}}c_{1,n-|J|+1}(\otimes_{k\in J^c} e_{\alpha_k} \otimes e_\nu)\\
& -\frac{1}{24}\int_{\oM_{0,n+2}} c_{0,n+2}\left(\otimes_{k=1}^n e_{\alpha_k}\otimes e_\mu \otimes e_\nu\right)\eta^{\mu\nu}.
\end{split}
\end{equation*}
In terms of generating functions, this becomes
\begin{equation*}
\begin{split}
\oG = \og + i\hbar \int \left[-\frac{1}{48} u^\alpha_{xx} c_{\alpha\mu}^\mu -\frac{1}{2}u^\alpha_{xx} c_\alpha^\mu \frac{\d \mathsf{G}}{\d u^\mu} + \frac{1}{2} \left(\d_x^2 \frac{\d F}{\d u^\mu}\right) \eta^{\mu \nu} \frac{\d \mathsf{G}}{\d u^\nu} - \frac{1}{24} c_\mu^\mu \right] dx + O(\hbar^2)
\end{split}
\end{equation*}
which can be brought to the form of equation \eqref{eq:genus1G} by integrating by parts.\\

The proof of equation \eqref{eq:genus1Gad} is completely analogous, the only difference being the insertion of a psi class to the power $d$ at an extra marked point, which makes it necessary to use genus $1$ topological recursion relations (see \cite{Wit91})
$$\frac{\d F_1(t^*_*)}{\d t^\alpha_d} = \frac{1}{24} \frac{\d^3 F_0(t^*_*)}{\d t^\alpha_{d-1} \d t^\epsilon_0 \d t^\delta_0} \eta^{\epsilon \delta}+ \frac{\d^2 F_0(t^*_*)}{\d t^\alpha_{d-1} \d t^\epsilon_0} \eta^{\epsilon \delta} \frac{\d F_1(t^*_*)}{\d t^\delta_0}$$
where $F_g(t^*_*)$ is the genus $g$ potential of the CohFT and whose right hand side, when restricted to $t^\alpha_0=u^\alpha$ and $t^*_p = 0$ for $p>0$, becomes the term
$$ \frac{1}{24}\frac{\d^2 g_{\alpha,d-1}^{[0]}}{\d u^\epsilon \d u^\delta} \eta^{\epsilon \delta}+ \frac{\d g_{\alpha,d-1}^{[0]}}{\d u^\epsilon} \eta^{\epsilon \delta} \frac{\d \mathsf{G}}{\d u^\delta}$$
in equation (\ref{eq:genus1Gad}).
\end{proof}

%%%%%%%%%%%%%%%%%%%%%%%%%%%%%%%%%%%%%%%%%%%%%%

\subsection{$3$- and $4$-spin quantum double ramification hierarchies}

As an application of the genus-$1$ computation of the previous section we compute the quantum DR hierarchy of Witten's $r$-spin class, for $r=3,4$. In light of the results of \cite{BG15,BDGR16}, which establish that the DR hierarchy in these cases coincides with the DZ hierarchy once we pass to the normal coordinates $\widetilde{u}^\alpha = \eta^{\alpha\mu}\frac{\delta \og_{\mu,0}}{\delta u^1}$ (which, for $r=4$ also changes the form of the Hamiltonian operator, see \cite{BDGR16}), and the fact that the DZ hierarchies for the $3$- and $4$-spin theories correspond in turn to the $3$- and $4$-KdV Gelfand-Dickey hierarchies \cite{DZ05,Dic03}, we obtain this way a quantization for such two well-known integrable systems.\\

Recall from \cite{Wi93,PV00} that, fixing $r\geq 2$ and an $(r-1)$-dimensional vector space $V$ with a basis $e_1,\ldots,e_{r-1}$, Witten's $r$-spin cohomological field theory $W_g(e_{a_1+1},\ldots,e_{a_n+1})=W_g(a_1,\ldots,a_n) \in H^*(\oM_{g,n};\mbQ)$ is a class of degree $\deg W_g(a_1,\ldots,a_n)=\frac{(r-2)(g-1)+\sum_{i=1}^n a_i}{r}$ if $a_i\in\{0,\ldots,r-2\}$ are such that this degree is a non-negative integer, and vanishes otherwise. By \cite{PPZ15}, this cohomological field theory is completely determined, thanks to generic semisimplicity, by the initial conditions $W_0(a_1,a_2,a_3)=1$ if $a_1+a_2+a_3=r-2$ (and zero otherwise) and $W_0(1,1,r-2,r-2)=\frac{1}{r}[\text{pt}]$ for $r\geq 3$ (while it vanishes for $r=2$). In particular, the metric $\eta$ takes the form $\eta_{\alpha\beta}=\delta_{\alpha+\beta,r}$.\\

\begin{theorem}
For $r=3,4$, the quantum double ramification hierarchies for Witten's $r$-spin classes are uniquely determined by
\begin{align*}
\oG_{1,1}^{3\text{-spin}} =&\int\left[\left(\frac{1}{2} \left(u^1\right)^2 u^2+\frac{\left(u^2\right)^4}{36}\right)+\left(-\frac{1}{12} \left(u_1^1\right){}^2-\frac{1}{24} u^2 \left(u_1^2\right){}^2\right) \eps ^2+\frac{1}{432} \left(u_2^2\right){}^2
   \eps ^4 - \frac{i\hbar}{12} u^1\right] dx,\\
\oG_{1,1}^{4\text{-spin}} =&\int \left[\left(\frac{ u^1 \left(u^2\right)^2}{2}+\frac{\left(u^1\right)^2 u^3}{2} +\frac{\left(u^2\right)^2 \left(u^3\right)^2}{8} +\frac{\left(u^3\right)^5}{320}\right)\right. \\
&\left.+\left(-\frac{\left(u_1^1\right){}^2}{8} -\frac{u^3
   \left(u_1^2\right){}^2}{16} -\frac{u^3 u_1^1 u_1^3}{32} +\frac{3}{64} \left(u^2\right)^2 u_2^3+\frac{1}{192} \left(u^3\right)^3 u_2^3\right) \eps ^2 \right. \\
   & \left. +\left(\frac{1}{160} \left(u_2^2\right){}^2+\frac{3}{640} u_2^1
   u_2^3+\frac{5 \left(u^3\right)^2 u_4^3}{4096}\right) \eps ^4-\frac{\left(u_3^3\right){}^2 \eps ^6}{8192}+ \right. \\
   &\left. +\left(  \frac{1}{96} \left(u^3_1\right)^2 - \frac{1}{96} \left(u^3\right)^2 -\frac{1}{8} u^1 \right)i\hbar - \frac{1}{1280} u^3 i\hbar \eps^2 \right]dx,
\end{align*}
\end{theorem}
\begin{proof}
The classical parts of the above formulae are copied from \cite{BR14}. Moreover, from dimension counting, we obtain that $\oG^{r\text{-spin}}_{1,1}$ is a homogeneous local functional of degree $2r+2$ with respect to the grading
$|u^{a+1}_k|=r-a$, $|\eps|=1$, $|\hbar|=r+2$. This means that the quantum correction in $\oG_{1,1}^{3\text{-spin}}$ is entirely in genus $1$ and hence determined by Theorem \ref{theorem:genus1} (recall that the $\mathsf{G}$-function for the $r$-spin theory vanishes identically, see e.g. \cite{Str03}). The quantum correction in $\oG_{1,1}^{4\text{-spin}}$, instead, has a part in genus $1$ (to be determined again using Theorem \ref{theorem:genus1}) but also the genus $2$ term $\int a u^3 i\hbar \eps^2dx$, with $a \in \mbQ$. The constant $a$ corresponds to the intersection number $a=-\int_{\DR_2(0,0)}\lambda_1  \psi_1 W_2(e_1,e_3) = - 3\int_{\DR_2(0)} \lambda_1 W_2(e_3)= -3\int_{\oM_{2,1}} \lambda_2 \lambda_1 W_2(e_3)$. Using the fact that, on $\oM_{2,0}$, $\lambda_2\lambda_1 = \frac{1}{5760} [\text{pt}]$ and that the class of a fiber of $\pi:\oM_{2,1}\to \oM_{2,0}$ is represented by the closure of the locus of singular genus $2$ curves with $3$ nodes (one separating, two non-separating) and a marked point on either of the two irreducible components we obtain $a=-3 \times \frac{1}{5760} \times 2 \int_{\oM_{0,3}} W_0(e_\mu,e_\nu,e_\epsilon) \eta^{\mu \nu} \eta^{\epsilon \delta} \int_{\oM_{0,4}} W_0(e_\delta,e_\alpha,e_\beta,e_3) \eta^{\alpha \beta}= -\frac{1}{1280}$.
\end{proof}

%%%%%%%%%%%%%%%%%%%%%%%%%%%%%%%%%%%%%%%%%%%%%%%%%%%%%%%
%%%%%%%%%%%%%%%%%%%%%%%%%%%%%%%%%%%%%%%%%%%%%%%%%%%%%%

\section{Tau-symmetry and tau-functions for quantum integrable systems}

In this section we introduce a quantum version of the notions of tau-structure and tau-functions for a Hamiltonian hierarchy.%, and define accordingly a partition function for the qDR hierarchy as the quantum tau-function of its string solution. The classical limit of the qDR partition function coincides with the DR partition function introduced in \cite{BDGR16}.

\begin{remark}
We note here that, for the time being, we will restrict our definitions to the case (relevant for the quantum double ramification hierarchy) of quantum Hamiltonian systems whose commutator $[\cdot,\cdot]$ is the one defined in Section \ref{section:quantum hamiltonian systems}. This means in particular that the semiclassical limit has the Poisson structure in standard form $\{\cdot,\cdot\}_{\eta\d_x}$. A more general theory of quantum tau-structures will require a study and classification of star-products and commutators on the space of quantum differential polynomials and local functionals. We plan to study this subject in a future work.
\end{remark}

\subsection{Tau-symmetric quantum Hamiltonian hierarchies}

Consider a quantum Hamiltonian system defined by a family of pairwise commuting quantum local functionals $\oH_{\beta,q}\in(\hLambda^{\hbar}_N)^{[\leq 0]}$, parameterized by two indices $1\le\beta\le N$ and $q\ge 0$, $[\oH_{\beta,q},\oH_{\gamma,p}]=0$, with respect to the quantum commutator introduced in section \ref{section:quantum hamiltonian systems}:
$$\frac{\d u^\alpha}{\d {t^\beta_q}}  = [u^\alpha,\oH_{\beta,q}]$$

Let us assume that $\oH_{1,0}=\frac{1}{2}\int \eta_{\mu \nu} u^\mu u^\nu$. Notice that, in this case, $\frac{1}{\hbar} [f,\oH_{1,0}] =\{f,\oH_{1,0}\} =\sum_{k\geq 0} \frac{\d f}{\d u^\alpha_k} u^\alpha_{k+1}=\d_x f $ for any $f \in \hcA^\hbar_N$.\\

A {\it tau-structure} for such hierarchy is a collection of quantum differential polynomials~$H_{\beta,q}\in(\hcA^\hbar_N)^{[\leq 0]}$, $1\le\beta\le N$, $q\ge -1$, such that the following conditions hold:
\begin{enumerate}

\item $\oH_{\beta,-1} := \int H_{\beta,-1} dx =\int \eta_{\beta\mu} u^\mu dx$,

\item For $q\ge 0$, the quantum differential polynomials~$H_{\beta,q}$ are densities for the Hamiltonians~$\oH_{\beta,q}$,
\begin{gather}\label{eq:tau-structure2}
\oH_{\beta,q}=\int H_{\beta,q}dx.
\end{gather}

\item Tau-symmetry:
\begin{gather}\label{eq:tau-structure3}
[H_{\alpha,p-1},\oH_{\beta,q}]=[H_{\beta,q-1},\oH_{\alpha,p}],\quad1\le\alpha,\beta\le N,\quad p,q\ge 0.
\end{gather}

\end{enumerate}

Existence of a tau-structure imposes non-trivial constraints on a quantum Hamiltonian hierarchy. A quantum Hamiltonian hierarchy with a fixed tau-structure will be called {\it tau-symmetric}.

\subsection{Sufficient condition for the existence of a tau-structure}

Consider again a quantum Hamiltonian hierarchy defined by a family of pairwise commuting quantum local functionals $\oH_{\beta,q}\in(\hLambda^{\hbar}_N)^{[\leq 0]}$, parameterized by two indices $1\le\beta\le N$ and $q\ge 0$. In the same way, as in the previous section, we assume that $\oH_{1,0}=\frac{1}{2}\int \eta_{\mu \nu} u^\mu u^\nu$. We have the following quantum analogue of a result from \cite{BDGR16}.

\begin{proposition}\label{proposition:sufficient condition}
Suppose that 
$$
\frac{\d\oH_{\beta,q}}{\d u^1}=
\begin{cases}
\oH_{\beta,q-1},&\text{if $q\ge 1$},\\
\int\eta_{\beta\mu}u^\mu dx,&\text{if $q=0$}.
\end{cases}
$$
Then the differential polynomials 
$$
H_{\beta,q}:=\frac{\delta\oH_{\beta,q+1}}{\delta u^1},\quad q\ge -1,
$$
define a tau-structure for the quantum hierarchy.
\end{proposition}
\begin{proof}
We have $\oH_{\beta,-1}=\int\eta_{\beta\mu}u^\mu dx$. Condition~\eqref{eq:tau-structure2} is clear, since for $q\ge 0$ we have
$$
\int H_{\beta,q}dx=\int\frac{\delta\oH_{\beta,q+1}}{\delta u^1}dx=\frac{\d}{\d u^1}\oH_{\beta,q+1}=\oH_{\beta,q}.
$$
Let us check the tau-symmetry condition~\eqref{eq:tau-structure3}. We have the commutativity $[\oH_{\alpha,p},\oH_{\beta,q}]=0$. Let us apply the variational derivative $\frac{\delta}{\delta u^1}$ to this equation. It is much easier to do it in the $p$-variables~\eqref{eq:u-p change}. We have $[\oH_{\alpha,p},\oH_{\beta,q}] =\oH_{\alpha,p} \left(  e^{\sum_{k>0} i \hbar k \eta^{\mu \nu} \overleftarrow{\frac{\partial }{\partial p^\mu_{k}}} \overrightarrow{\frac{\partial }{\partial p^\nu_{-k}}}} -e^{\sum_{k>0} i \hbar k \eta^{\mu \nu} \overleftarrow{\frac{\partial }{\partial p^\mu_{-k}}} \overrightarrow{\frac{\partial }{\partial p^\nu_{k}}}}\right) \oH_{\beta,q}$. For the variational derivative we have $\frac{\delta\oH}{\delta u^\gamma}=\sum_{n\in\mbZ}e^{-inx}\frac{\d\oH}{\d p^\gamma_n}$ for any $\oH\in(\hLambda^\hbar_N)^{[\leq0]}$. Therefore, we get
\begin{align*}
0&=\frac{\delta}{\delta u^1}[\oH_{\alpha,p},\oH_{\beta,q}]=\\
&=\sum_{n\in\mbZ}e^{-inx}\frac{\d}{\d p^1_n}\left(\oH_{\alpha,p} \left(  e^{\sum_{k>0} i \hbar k \eta^{\mu \nu} \overleftarrow{\frac{\partial }{\partial p^\mu_{k}}} \overrightarrow{\frac{\partial }{\partial p^\nu_{-k}}}} -e^{\sum_{k>0} i \hbar k \eta^{\mu \nu} \overleftarrow{\frac{\partial }{\partial p^\mu_{-k}}} \overrightarrow{\frac{\partial }{\partial p^\nu_{k}}}}\right) \oH_{\beta,q}\right)=\\
&=\left(\sum_{n\in\mbZ}e^{-inx}\frac{\d\oH_{\alpha,p}}{\d p^1_n}\right)\left(  e^{\sum_{k>0} i \hbar k \eta^{\mu \nu} \overleftarrow{\frac{\partial }{\partial p^\mu_{k}}} \overrightarrow{\frac{\partial }{\partial p^\nu_{-k}}}} -e^{\sum_{k>0} i \hbar k \eta^{\mu \nu} \overleftarrow{\frac{\partial }{\partial p^\mu_{-k}}} \overrightarrow{\frac{\partial }{\partial p^\nu_{k}}}}\right) \oH_{\beta,q}\\
&\phantom{=}+\oH_{\alpha,p}\left(  e^{\sum_{k>0} i \hbar k \eta^{\mu \nu} \overleftarrow{\frac{\partial }{\partial p^\mu_{k}}} \overrightarrow{\frac{\partial }{\partial p^\nu_{-k}}}} -e^{\sum_{k>0} i \hbar k \eta^{\mu \nu} \overleftarrow{\frac{\partial }{\partial p^\mu_{-k}}} \overrightarrow{\frac{\partial }{\partial p^\nu_{k}}}}\right) \left(\sum_{n\in\mbZ}e^{-inx}\frac{\d\oH_{\beta,q}}{\d p^1_n}\right)=\\
&=[H_{\alpha,p-1},\oH_{\beta,q}]-[H_{\beta,q-1},\oH_{\alpha,p}].
\end{align*}
The proposition is proved.
\end{proof}

\begin{corollary}
The quantum double ramification hierarchy $\{\oG_{\alpha,d}\}_{\substack{1 \leq \alpha\leq N, d\geq -1}}$, with $G_{\alpha,d}$ given by (\ref{density}), is tau-symmetric. A tau-structure is given by the densities $H_{\alpha,d} = \frac{\delta \oG_{\alpha,d+1}}{\delta u^1}$.
\end{corollary}

%%%%%%%%%%%%%%%%%%%%%%%%%%%%%%%%%%%%%%%%%%%%%

\subsection{Quantum tau-functions}

We consider again a quantum Hamiltonian hierarchy generated by Hamiltonians $\oH_{\alpha,p}$, $1\leq \alpha\leq N$, $p\geq -1$ where $\oH_{1,0}=\frac{1}{2}\int \eta_{\mu \nu} u^\mu u^\nu dx$. Suppose that the quantum differential polynomials $H_{\beta,q}$, $1\le\beta\le N$, $q\ge -1$, define a tau-structure for such hierarchy. From commutativity of the Hamiltonians we have
\begin{gather}\label{eq:conserved quantity}
\int [ H_{\alpha,p-1}, \oH_{\beta,q}]dx=0.
\end{gather}
The quantum differential polynomial $[ H_{\alpha,p-1}, \oH_{\beta,q}]$ has no constant term (because of the form of the quantum commutator), hence there exists a unique differential polynomial $\Omega^\hbar_{\alpha,p;\beta,q}\in(\hcA^\hbar_N)^{[\leq 0]}$ such that 
\begin{gather}\label{eq:definition of Omega}
\d_x\Omega^\hbar_{\alpha,p;\beta,q}= [ H_{\alpha,p-1}, \oH_{\beta,q}]\quad\text{and}\quad\left.\Omega^\hbar_{\alpha,p;\beta,q}\right|_{u^*_*=0}=0.
\end{gather}
The differential polynomial $\Omega^\hbar_{\alpha,p;\beta,q}$ is called the {\it two-point function} of the given tau-structure of the hierarchy. From condition~\eqref{eq:tau-structure3} it follows that 
\begin{gather}
\Omega^\hbar_{\alpha,p;\beta,q}=\Omega^\hbar_{\beta,q;\alpha,p}
\end{gather}
and, moreover, it implies that the differential polynomial 
\begin{gather}
[\Omega^\hbar_{\alpha,p;\beta,q}, \oH_{\gamma,r}]
\end{gather}
is symmetric with respect to all permutations of the pairs $(\alpha,p)$, $(\beta,q)$, $(\gamma,r)$. Since the Hamiltonian $\oH_{1,0}$ generates the spatial translations, equation~\eqref{eq:definition of Omega} implies that $\d_x\Omega^\hbar_{\alpha,p;1,0}=\d_x H_{\alpha,p-1}$, $p\ge 0$. Therefore,
\begin{gather}\label{eq:Omega-h relation}
\Omega^\hbar_{\alpha,p;1,0}-H_{\alpha,p-1}=C,\quad p\ge 0,
\end{gather}
where $C=C(\eps,\hbar)$ is a formal power series in $\eps$ and $\hbar$.\\

Consider also the evolved Hamiltonians
\begin{gather}
H^{t^*_*}_{\alpha,p} =  \exp\left(\sum_{q\geq 0} \frac{t^\beta_q}{\hbar}[\cdot,\oH_{\beta,q}]\right) H_{\alpha,p} \in \hcA_N^\hbar[[t^*_*]],
\end{gather}
and the evolved two-point functions
\begin{gather}
\Omega^{\hbar,t^*_*}_{\alpha,p;\gamma,r} = \exp\left(\sum_{q\geq 0} \frac{t^\beta_q}{\hbar}[\cdot,\oH_{\beta,q}]\right) \Omega^{\hbar}_{\alpha,p;\gamma,r} \in \hcA_N^\hbar[[t^*_*]].
\end{gather}
They satisfy, respectively, 
$$\frac{\d H^{t^*_*}_{\alpha,p}}{\d t^\beta_q}=\frac{1}{\hbar} [H^{t^*_*}_{\alpha,p},\oH_{\beta,q}],\qquad \left.H^{t^*_*}_{\alpha,p}\right|_{t^*_*=0} = H_{\alpha,p} \in \hcA^\hbar_N$$
and
$$\frac{\d \Omega^{\hbar,t^*_*}_{\alpha,p;\gamma,r}}{\d t^\beta_q}=\frac{1}{\hbar} [\Omega^{\hbar,t^*_*}_{\alpha,p;\gamma,r},\oH_{\beta,q}],\qquad \left.\Omega^{\hbar,t^*_*}_{\alpha,p;\gamma,r}\right|_{t^*_*=0} = \Omega^\hbar_{\alpha,p;\gamma,r} \in \hcA^\hbar_N$$
together with 
\begin{gather}\label{eq:two-point}
\frac{\d H^{t^*_*}_{\alpha,p-1}}{\d t^\beta_q} =\Omega^{\hbar,t^*_*}_{\alpha,p;\beta,q} =\Omega^{\hbar,t^*_*}_{\beta,q;\alpha,p} = \frac{\d H^{t^*_*}_{\beta,q-1}}{\d t^\alpha_p}.
\end{gather}
Moreover
\begin{gather}\label{eq:three-point}
\frac{\d \Omega^{\hbar,t^*_*}_{\alpha,p;\beta,q}}{\d t^\gamma_r}
\end{gather}
is symmetric with respect to all permutations of the pairs $(\alpha,p)$, $(\beta,q)$ and $(\gamma,r)$.\\

Then equation~\eqref{eq:two-point} and the symmetry of~\eqref{eq:three-point} imply that there exists a function $P\in\hcA_N^\hbar[[t^*_*]]$ such that 
$$
\Omega^{\hbar,t^*_*}_{\alpha,p;\beta,q}= \frac{\d^2 P}{\d t^\alpha_p\d t^\beta_q},\quad\text{for any $1\le\alpha,\beta\le N$ and $p,q\ge 0$}.
$$
To each initial condition $u^\alpha_k|_{x=t^*_*=0} = c^\alpha_k(\eps,\hbar)\in \mbC[[\eps,\hbar]]$ with $c^\alpha_k(0,0) = 0$ we can associate the restriction $\left. P \right|_{u^\alpha_k=c^\alpha_k(\eps,\hbar)}\in \mbC[[t^*_*,\eps,\hbar]]$ which is called (the logarithm of) the tau-function of the given solution.

%%%%%%%%%%%%%%%%%%%%%%%%%%%%%%%%%%%%%%%%%%%%%%%%%%%%%%%
%%%%%%%%%%%%%%%%%%%%%%%%%%%%%%%%%%%%%%%%%%%%%%%%%%%%%

\section{Hierarchies of double ramification type}

In this section we interpret the recursion (\ref{eq:first recursion}) as a system of functional derivative equations for $\oG_{1,1}$ and elevate it to the main axiom in the definition of a class of (quantum or classical) Hamiltonians producing integrable, tau-symmetric Hamiltonian systems.

\subsection{An integrability condition for Hamiltonian systems}

Let us consider the quantum Hamiltonian system defined by a Hamiltonian $\oH\in (\hLambda_N^\hbar)^{[\leq 0]}$ with respect to the standard quantum commutator introduced in Section \ref{section:quantum hamiltonian systems}. We give a sufficient condition for $\oH$ to be part of an integrable hierarchy. Consider the operator $\cD^\hbar_\oH:\hcA_N^\hbar[[z]]\to\hcA_N^\hbar[[z]]$ defined by
\begin{equation*}
\begin{split}
& \cD^\hbar_\oH f(z)= \d_x (D-1) f(z) - \frac{z}{\hbar}[f(z),\oH],\\
& f(z)=  f(u^*_*;\eps,\hbar;z) = \sum_{k\geq 0}  f_{k-1}(u^*_*;\eps,\hbar) z^k, \qquad  f_{k-1}(u^*_*;\eps,\hbar) \in (\hcA_N^\hbar)^{[\leq 0]}.\\
\end{split}
\end{equation*}
Suppose there exist $N$ solutions $G_{\alpha}(z) \in (\hcA_N^\hbar)^{[\leq 0]}[[z]]$, $\alpha=1,\ldots,N$, to $\cD^\hbar_\oH G_\alpha(z)=0$ with the initial conditions $G_{\alpha}(z=0)=\eta_{\alpha\mu}u^\mu$. Then a new vector of solutions can be found by the following transformation
\begin{equation}\label{eq:solution transf}
G_{\alpha}(z) \mapsto A^\mu_\alpha(z) G_\mu(z) + B_\alpha(z),
\end{equation}
where $A^\mu_\alpha(z)=\delta^\mu_\alpha + \sum_{i > 0} A^\mu_{\alpha,i} z^i \in \mbC[[z]]$ and $B_\alpha(z)=\sum_{i>0} B_{\alpha,i}(\eps,\hbar) z^i \in \mbC[[\eps,\hbar,z]]$.

\begin{theorem}\label{theorem:recursion->integrability}
Assume that $\oH \in (\hLambda_N^\hbar)^{[\leq 0]}$ has the following properties:
\begin{itemize}
\item[(a)] there exist $N$ independent solutions $G_{\alpha}(z) = \sum_{p\geq 0} G_{\alpha,p-1} z^p\in (\hcA_N^\hbar)^{[\leq 0]}[[z]]$, $\alpha=1,\ldots,N$, to the equation
\begin{equation}\label{eq:operator recursion}
\cD^\hbar_\oH G_\alpha(z) = 0
\end{equation}
with the initial conditions $ G_{\alpha}(z=0)=\eta_{\alpha\mu}u^\mu$,
\item[(b)] $\displaystyle \frac{\delta \oH}{\delta u^1} = \frac{1}{2}\eta_{\mu \nu } u^\mu u^\nu + \d_x R + c(\eps,\hbar), \qquad R\in (\hcA^\hbar_N)^{[\le -1]}, \quad c(\eps,\hbar) \in \mbC[[\eps,\hbar]]$,
\item[(c)] $\oG_{1,1} = \oH$.
\end{itemize}
Then, up to a transformation of type (\ref{eq:solution transf}), we have
\begin{itemize}
\item[(i)] $\displaystyle \oG_{1,0} =  \int \left( \frac{1}{2}\eta_{\mu \nu } u^\mu u^\nu \right)dx$,
\item[(ii)] $\displaystyle [\oG_{\alpha,p},\oG_{\beta,q}] = 0, \qquad \alpha,\beta=1,\ldots,N,\quad p,q\geq -1$,
\item[(iii)] $\displaystyle \frac{1}{\hbar}[G_{\alpha,p},\oG_{\beta,0}] = \d_x \frac{\d G_{\alpha,p+1}}{\d u^\beta}, \qquad \beta=1,\ldots,N, \quad p\geq -1,$ 
\item[(iv)] $\displaystyle \frac{\d G_{\alpha,p}}{\d u^1} = G_{\alpha,p-1}, \qquad \alpha=1,\ldots,N, \quad p\geq -1$,
\end{itemize}
hence in particular $\oH$ is part of a quantum integrable tau-symmetric hierarchy.
\end{theorem}
\begin{proof}
Equation \eqref{eq:operator recursion} implies in particular that $[\oG_{\alpha,p},\oH] = 0$ for every $\alpha=1,\ldots,N$, $p\geq-1$. Moreover we have
$$\d_x (D-1) G_{1,0} = \frac{1}{\hbar}[G_{1,-1},\oH] =  \d_x \frac{\delta \oH}{\delta u^1} =\d_x\left( \frac{1}{2} \eta_{\mu\nu} u^\mu u^\nu + \d_x R + c(\eps,\hbar)\right),$$
which proves $\mathrm{(i)}$.\\

%At the next step of the recursion we obtain $G_{1,1} = (D-1)^{-1} \d_x^{-1} \frac{1}{\hbar}[G_{1,0},\oH]$. A direct computation shows that $\frac{1}{\hbar}[\frac{1}{2}\eta_{\mu \nu} u^\mu u^\nu, \oH] = \d_x (D-1) H+ \d_x^2 S$ where $H,S \in \hcA^\hbar_N$ with $\oH = \int H dx$, so we deduce $\frac{1}{\hbar}[G_{1,0},\oH] = \d_x (D-1) H+ \d_x^2 S + \d^2_x (D-1)^{-1} R$. This implies
%$$\oG_{1,1}= \int \left[(D-1)^{-1} \d_x^{-1} \left(\d_x (D-1) H+ \d_x^2 S + \d^2_x (D-1)^{-1} R\right)\right] = \oH $$
%which proves $\mathrm{(ii)}$.\\

We write equation \eqref{eq:operator recursion} as $\d_x (D-1) G_{\alpha,p} = \frac{1}{\hbar} [G_{\alpha,p-1},\oG_{1,1}]$. To prove $\mathrm{(ii)}$ we will show that such recursion implies
\begin{equation*}
\frac{1}{\hbar} \left[ G_{\alpha,p}(x) , G_{\beta,q}(y)\right] = \partial_x G_{\alpha,p+1;\beta,q}(x,y) - \partial_y G_{\alpha,p;\beta,q+1}(x,y),
\end{equation*}
for $\alpha,\beta=1,\ldots,N$, $p,q\geq 0$ (which is equation \eqref{eqmainlemma}), for some opportunely defined $G_{\alpha,p;\beta,q}(x,y)$, symmetric with respect to simultaneous exchange of the indices $(\alpha,p,x)$ and $(\beta,q,y)$. We proceed by recursion starting from the fact that, for $p\geq 0$
\begin{equation*}
\frac{1}{\hbar} \left[ G_{\alpha,p}(x) , G_{\beta,-1}(y)\right] = \sum_{l\geq 0} \frac{\d G_{\alpha,p}}{\d u^\beta_l} \delta^{(l+1)}(x-y) = -\d_y \left( \sum_{l\geq 0}  \frac{\d G_{\alpha,p}}{\d u^\beta_l} \delta^{(l)}(x-y) \right),
\end{equation*}
so that we can pose
$$
G_{\alpha,p;\beta,0}(x,y)  := \sum_{l\geq 0}  \frac{\d G_{\alpha,p}}{\d u^\beta_l} \delta^{(l)}(x-y)=: G_{\beta,0;\alpha,p}(y,x),\qquad G_{\alpha,p;\beta,-1}(x,y) = G_{\beta,-1;\alpha,p}(y,x) = 0,
$$
and have
$$\frac{1}{\hbar} \left[ G_{\alpha,p}(x) , G_{\beta,-1}(y)\right] = \partial_x G_{\alpha,p+1;\beta,-1}(x,y) - \partial_y G_{\alpha,p;\beta,0}(x,y),$$
$$\frac{1}{\hbar} \left[ G_{\alpha,-1}(x) , G_{\beta,q}(y)\right] = \partial_x G_{\alpha,0;\beta,q}(x,y) - \partial_y G_{\alpha,-1;\beta,q+1}(x,y).$$
Now we assume
$$\frac{1}{\hbar} \left[ G_{\alpha,p}(x) , G_{\beta,q-1}(y)\right] = \partial_x G_{\alpha,p+1;\beta,q-1}(x,y) - \partial_y G_{\alpha,p;\beta,q}(x,y),$$
$$\frac{1}{\hbar} \left[ G_{\alpha,p-1}(x) , G_{\beta,q}(y)\right] = \partial_x G_{\alpha,p;\beta,q}(x,y) - \partial_y G_{\alpha,p-1;\beta,q+1}(x,y),$$
and obtain
\begin{equation*}
\begin{split}
D \frac{1}{\hbar} & [ G_{\alpha,p}(x) , G_{\beta,q}(y)] \\
&=\frac{1}{\hbar} [(D-1) G_{\alpha,p}(x),G_{\beta,q}(y)] + \frac{1}{\hbar} [ G_{\alpha,p}(x) , (D-1) G_{\beta,q}(y)]\\
& = \frac{1}{\hbar} \left[\d_x^{-1} \frac{1}{\hbar}\left[G_{\alpha,p-1}(x),\oG_{1,1}\right],G_{\beta,q}(y)\right] +  \frac{1}{\hbar} \left[G_{\alpha,p}(x),\d_y^{-1} \frac{1}{\hbar}\left[G_{\beta,q-1}(y),\oG_{1,1}\right]\right]\\
& = \d_x^{-1}\frac{1}{\hbar} \left[\frac{1}{\hbar}\left[G_{\alpha,p-1}(x),\oG_{1,1}\right],G_{\beta,q}(y)\right] +  \d_y^{-1}\frac{1}{\hbar} \left[G_{\alpha,p}(x), \frac{1}{\hbar}\left[G_{\beta,q-1}(y),\oG_{1,1}\right]\right]\\
&=\d_x^{-1}\left(-\frac{1}{\hbar}\left[\frac{1}{\hbar}\left[G_{\beta,q}(y),G_{\alpha,p-1}(x)\right],\oG_{1,1}\right]-\frac{1}{\hbar}\left[\frac{1}{\hbar}\left[\oG_{1,1},G_{\beta,q}(y)\right],G_{\alpha,p-1}(x)\right]\right)\\
& +\d_y^{-1}\left(-\frac{1}{\hbar}\left[\oG_{1,1},\frac{1}{\hbar}\left[G_{\alpha,p}(x),G_{\beta,q-1}(y)\right]\right]-\frac{1}{\hbar}\left[G_{\beta,q-1}(y),\frac{1}{\hbar}\left[\oG_{1,1},G_{\alpha,p}(x)\right]\right]\right)\\
&=\d_x^{-1}\left(\frac{1}{\hbar}\left[\partial_x G_{\alpha,p;\beta,q}(x,y) - \partial_y G_{\alpha,p-1;\beta,q+1}(x,y),\oG_{1,1}\right]+\frac{1}{\hbar}\left[\d_y (D-1) G_{\beta,q+1}(y),G_{\alpha,p-1}(x)\right]\right)\\
& +\d_y^{-1}\left(\frac{1}{\hbar}\left[\partial_x G_{\alpha,p+1;\beta,q-1}(x,y) - \partial_y G_{\alpha,p;\beta,q}(x,y),\oG_{1,1}\right]+\frac{1}{\hbar}\left[G_{\beta,q-1}(y),\d_x (D-1) G_{\alpha,p+1}(x)\right]\right)\\
&=-\d_y \d_x^{-1}\left(\frac{1}{\hbar}\left[G_{\alpha,p-1;\beta,q+1}(x,y),\oG_{1,1}\right]-\frac{1}{\hbar}\left[ (D-1) G_{\beta,q+1}(y),G_{\alpha,p-1}(x)\right]\right)\\
& +\d_x \d_y^{-1}\left(\frac{1}{\hbar}\left[ G_{\alpha,p+1;\beta,q-1}(x,y) ,\oG_{1,1}\right]+\frac{1}{\hbar}\left[G_{\beta,q-1}(y), (D-1) G_{\alpha,p+1}(x)\right]\right).
\end{split}
\end{equation*}
Hence we can define
$$G_{\alpha,p+1;\beta,q}(x,y) = D^{-1}\d_y^{-1}\left(\frac{1}{\hbar}\left[ G_{\alpha,p+1;\beta,q-1}(x,y) ,\oG_{1,1}\right]+\frac{1}{\hbar}\left[G_{\beta,q-1}(y), (D-1) G_{\alpha,p+1}(x)\right]\right),$$
$$G_{\alpha,p;\beta,q+1}(x,y) = D^{-1} \d_x^{-1}\left(\frac{1}{\hbar}\left[G_{\alpha,p-1;\beta,q+1}(x,y),\oG_{1,1}\right]-\frac{1}{\hbar}\left[ (D-1) G_{\beta,q+1}(y),G_{\alpha,p-1}(x)\right]\right), $$
which enjoy the correct symmetry property with respect to exchange of indices and variables. By induction we arrive then to the proof of $\mathrm{(ii)}$.\\

From the last equation we can deduce in particular that $\int G_{\alpha,p+1;\beta,0}(x,y) dy = \frac{\d G_{\alpha,p+1}(x)}{\d u^\beta}$. We also have
$$
\frac{1}{\hbar} \left[ G_{\alpha,p}(x) , G_{\beta,0}(y)\right] = \partial_x G_{\alpha,p+1;\beta,0}(x,y) - \partial_y G_{\alpha,p;\beta,1}(x,y)
$$
which, upon integration with respect to $y$, gives
$$\frac{1}{\hbar} \left[ G_{\alpha,p}(x) , \oG_{\beta,0}\right]=\int\d_x G_{\alpha,p+1;\beta,0}(x,y) dy=\d_x \frac{\d G_{\alpha,p+1}(x)}{\d u^\beta}$$
and proves $\mathrm{(iii)}$.\\

Point $\mathrm{(iv)}$ follows from point $\mathrm{(iii)}$ in the case $\beta = 1$, which gives $\d_x \frac{\d G_{\alpha,p+1}}{\d u^1} = \d_x G_{\alpha,p}$.
\end{proof}

We also have the following theorem, which is slightly stronger than the classical version of the above one. For a local functional $\oh\in\hLambda^{[0]}_N$ consider the operator $\cD_{\oh}:\hcA_N[[z]]\to\hcA_N[[z]]$ defined by
\begin{equation*}
\begin{split}
& \cD_{\oh} f(z)= \d_x (D-1) f(z) - z\{f(z),\oh\},\\
& f(z)=  f(u^*_*;\eps;z) = \sum_{k\geq 0}  f_{k-1}(u^*_*;\eps) z^k, \qquad  f_{k-1}(u^*_*;\eps) \in \hcA_N^{[0]}.\\
\end{split}
\end{equation*}
Suppose there exist $N$ solutions $g_{\alpha}(z) \in \hcA_N^{[0]}[[z]]$, $\alpha=1,\ldots,N$, to $\cD_{\oh} g_\alpha(z)=0$ with the initial conditions $g_{\alpha}(z=0)=\eta_{\alpha\mu}u^\mu$. Then a new vector of solutions can be found by the following transformation
\begin{equation}\label{eq:solution transf classical}
g_{\alpha}(z) \mapsto a^\mu_\alpha(z) g_\mu(z) + b_\alpha(z),
\end{equation}
where $a^\mu_\alpha(z)=\delta^\mu_\alpha + \sum_{i > 0} a^\mu_{\alpha,i} z^i \in \mbC[[z]]$ and $b_\alpha(z)=\sum_{i>0} b_{\alpha,i} z^i \in \mbC[[z]]$.

\begin{theorem}\label{theorem:recursion->integrability classical}
Assume that $\oh \in \hLambda_N^{[0]}$ has the following properties:
\begin{itemize}
\item[(a)] there exist $N$ independent solutions $g_{\alpha}(z) = \sum_{p\geq 0} g_{\alpha,p-1} z^p\in \hcA_N^{[0]}[[z]]$, $\alpha=1,\ldots,N$, to the equation
\begin{equation}\label{eq:operator recursion classical}
\cD_{\oh} g_\alpha(z) = 0
\end{equation}
with the initial conditions $ g_{\alpha}(z=0)=\eta_{\alpha\mu}u^\mu$,
\item[(b)] $\displaystyle \frac{\delta \oh}{\delta u^1} = \frac{1}{2}\eta_{\mu \nu } u^\mu u^\nu + \d_x^2 r, \qquad r\in \hcA_N^{[-2]}$.
\end{itemize}
Then, up to a transformation of type (\ref{eq:solution transf classical}), we have
\begin{itemize}
\item[(i)] $\displaystyle g_{1,0} =  \frac{1}{2}\eta_{\mu \nu } u^\mu u^\nu + \d_x^2 (D-1)^{-1} r$,
\item[(ii)] $\og_{1,1} = \oh$,
\item[(iii)] $\displaystyle \{\og_{\alpha,p},\og_{\beta,q}\} = 0, \qquad \alpha,\beta=1,\ldots,N,\quad p,q\geq -1$,
\item[(iv)] $\displaystyle \{g_{\alpha,p},\og_{\beta,0}\} = \d_x \frac{\d g_{\alpha,p+1}}{\d u^\beta}, \qquad \beta=1,\ldots,N, \quad p\geq -1,$ 
\item[(v)] $\displaystyle \frac{\d g_{\alpha,p}}{\d u^1} = g_{\alpha,p-1}, \qquad \alpha=1,\ldots,N, \quad p\geq -1$,
\end{itemize}
hence in particular $\oh$ is part of an integrable tau-symmetric hierarchy.
\end{theorem}
\begin{proof}
The only differences in the statement of this theorem from the classical limit of Theorem~\ref{theorem:recursion->integrability} are that hypothesis (b) has become stronger together with claim (i) and that hypothesis~(c) of Theorem \ref{theorem:recursion->integrability} has now become claim (ii) and so it needs to be proved. The proof of~(i) follows from equation (\ref{eq:operator recursion classical}):
$$\d_x (D-1) g_{1,0} = \{g_{1,-1},\oh\} =  \d_x \frac{\delta \oh}{\delta u^1} =\d_x\left( \frac{1}{2} \eta_{\mu\nu} u^\mu u^\nu + \d^2_x r\right).$$
Also from equation (\ref{eq:operator recursion classical}) we obtain that $g_{1,1} = (D-1)^{-1} \d_x^{-1} \{g_{1,0},\oh\}$. A direct computation shows that $\{\frac{1}{2}\eta_{\mu \nu} u^\mu u^\nu, \oh\} = \d_x (D-1) h+ \d_x^2 s$, where $h\in\hcA^{[0]}_N$, $s\in\hcA^{[-1]}_N$ with $\oh = \int h dx$, so we deduce $\{g_{1,0},\oh\} = \d_x (D-1) h+ \d_x^2 s + \d^2_x \{ (D-1)^{-1} r,\oh\}$, where we used that $D h = \left(\sum_{k\geq 0}(k+1)u^\alpha_k \frac{\d}{\d u^\alpha_k}\right)h$. This implies, always up to (\ref{eq:solution transf classical}),
$$
\og_{1,1}= \int \left[(D-1)^{-1} \d_x^{-1} \left(\d_x (D-1) h+ \d_x^2 s + \d^2_x \{(D-1)^{-1} r,\oh\}\right)\right] dx= \oh.
$$
\end{proof}

\begin{remark}
When we restrict to $\hbar=\eps=0$, a particular Hamiltonian satisfying conditions (a) and~(b) of Theorem \ref{theorem:recursion->integrability} is given by $\left.\oH\right|_{\hbar=\eps=0} =(D-2)\int F(u^1,\ldots,u^N) dx $, where the function $F=F(u^1,\ldots,u^N)$ is a solution to the WDVV equations
$$\frac{\d^3 F}{\d u^\alpha \d u^\beta \d u^\mu} \eta^{\mu \nu}\frac{\d^3 F}{\d u^\nu \d u^\gamma \d u^\delta} = \frac{\d^3 F}{\d u^\alpha \d u^\delta \d u^\mu} \eta^{\mu \nu}\frac{\d^3 F}{\d u^\nu \d u^\gamma \d u^\beta} ,$$
$$\frac{\d^3 F}{\d u^1 \d u^\alpha \d u^\beta} = \eta_{\alpha \beta}, $$
for $\alpha,\beta,\gamma,\delta = 1,\ldots,N$. This is because, at $\hbar=\eps=0$, equation (\ref{eq:operator recursion}) promptly reduces to an averaged (and hence weaker) form of genus $0$ topological recursion relations, and the WDVV equations are equivalent to the existence of $N$ independent solutions to such equations.  At that point, such $N$ solutions to $\eqref{eq:operator recursion}$ correspond to the $N$ generating functions of the classical ($\hbar=0$) dispersionless ($\eps=0$) Hamiltonian densities $g^{[0]}_{\alpha}(z) := \left. G_{\alpha}(z)\right|_{\hbar=\eps=0}$ of the principal hierarchy of the resulting (formal) Frobenius manifold, that is the $N$ flat coordinates of its deformed flat connection (see \cite{DZ05} for details). In such classical dispersionless context, Theorem \ref{theorem:recursion->integrability} is hence a generalization of results proved for instance in \cite{DZ05}.
\end{remark}

\begin{definition}
Let $\oH \in (\hLambda_N^\hbar)^{[\leq 0]}$ (resp. $\oh \in \hLambda_N^{[0]}$) satisfy the hypothesis of Theorem \ref{theorem:recursion->integrability} (resp. Theorem \ref{theorem:recursion->integrability classical}). Then we say that $\oH$ (resp. $\oh$) and the induced quantum (resp. classical) integrable tau-symmetric hierarchy are \emph{of double ramification (DR) type}.
\end{definition}

\begin{theorem}
The quantum double ramification hierarchy (\ref{density}) with $\oH=\oG_{1,1}$ and its classical limit are hierarchies of double ramification type. 
\end{theorem}
\begin{proof}
Hypothesis (a) of Theorem \ref{theorem:recursion->integrability} is satisfied thanks to recursion (\ref{eq:first recursion}). Hypothesis~(b) follows for instance from the string equation (\ref{eq:string equation}) together with the fact that $\oG_{1,0} = \int\left( \frac{1}{2}\eta_{\mu\nu} u^\mu u^\nu\right) dx$ and hypothesis (c) holds by definition of double ramification hierarchy. For the classical counterpart, hypothesis (b) of Theorem \ref{theorem:recursion->integrability classical} is a consequence of the divibility, for $g,n\geq 1$, of $\pi_*\left(\lambda_g\DR_g(-\sum_{i=1}^n a_i,a_1,\ldots,a_n)\right)$ by $a_n^2$, where $\pi:\oM_{g,n+1}\to \oM_{g,n}$ forgets the last marked point, which was proved in \cite{BDGR16}, and which implies the possibility of finding a density for $\og_{1,1}$ which is independent of $u^1_x$.
\end{proof}

\subsection{Classification of rank $1$ quantum integrable hierarchies of DR type}

In this section we study quantum deformations of DR type of the Riemann hierarchy, which is the genus $0$ double ramification hierarchy associated to the trivial cohomological field theory with $V=\mbC\ni e_1$ and $c_{g,n}(e_1^{\otimes n}) = 1 \in H^0(\oM_{g,n},\mbQ)$. At first we concentrate on purely quantum deformations of the Riemann hierarchy, which means that, in this classification problem, the variable $\eps$ will not appear. This amounts to classifying quantum Hamiltonians of the form
$$
\oG_1 = \int  \frac{u^3}{6} dx + \sum_{k\geq 1} \oG_1^k \hbar^k, \qquad \oG_1^k \in \Lambda_1^{[\le 2k]},
$$
satisfying the hypothesis of Theorem \ref{theorem:recursion->integrability}, with $\eps=0$. An explicit computation gives, modulo terms proportional to the Casimir $\int u dx$, the following classification up to order $3$ in $\hbar$:

\begin{align}\label{eq:quantum classification}
\oG_1 = \int \left[ \frac{u^3}{6}+\left( a u_1{}^2\right) i \hbar +\left(b u_2{}^2\right) (i\hbar) ^2+\left(c u_2{}^3+\frac{10 b^2
   - c }{7 a}u_3{}^2\right) (i\hbar)^3+O\left(\hbar ^4\right)\right] dx.
\end{align}

In the above formula we assume $a\neq 0$. In case $a=0$, the computation gives $b=c=0$ too. Let us compare this formula with the Hamiltonian $\oG_1$ of the dispersionless (i.e. $\eps=0$) quantum DR hierarchy for a rank~$1$ cohomological field theory with $\eta_{1,1}=1$. According to~\cite{Teleman} such CohFTs are parameterized by numbers $r_1,r_2,\ldots$ in the following way:
\begin{gather}
c_{g,n}(e_1^{\otimes n})=e^{-\sum_{i\ge 1}\frac{(2i)!}{B_{2i}}r_i\Ch_{2i-1}(\mathbb E)}.
\end{gather}
Here $\Ch_{2i-1}$ denotes the $(2i-1)$-th component of the Chern character and the $B_{2i}$ are Bernoulli numbers (see also Section \ref{section:DR and DZ rank 1}). A direct computation along the line of Section \ref{subsection:DR hierarchy is standard} gives, up to the term $-\frac{i \hbar}{24}\int u dx$, exactly equation (\ref{eq:quantum classification}) with
$$a=-\frac{1}{2}r_1, \qquad b=-\frac{1}{12} r_2-\frac{2}{5}r_1^3,\qquad c=-\frac{7}{480} r_3 r_1+\frac{5}{72}r_2^2-\frac{1}{3} r_2 r_1^3-\frac{8}{25} r_1^6,$$
suggesting that dispersionless quantum deformations of the Riemann hierarchy are in one to one correspondence with rank $1$ cohomological field theories with $\eta_{1,1}=1$.\\

Assuming this correspondence, it is possible to recover dispersive deformations too by defining new parameters $s_i$ as follows
$$e^{-\sum_{i\ge 1}\frac{(2i)!}{B_{2i}}r_i\Ch_{2i-1}(\mathbb E)} = \Lambda\left(\frac{-\eps^2}{i \hbar}\right) e^{-\sum_{i\ge 1}\frac{(2i)!}{B_{2i}}s_i\Ch_{2i-1}(\mathbb E)}.$$
This amounts to
$$r_i = s_i +\frac{B_{2i}}{2i(2i-1)}\left(\frac{\eps^2}{i \hbar}\right)^{2i-1},$$
which gives
\begin{align*} a=&\frac{1}{i\hbar}\left(-\frac{\eps ^2}{24}-\frac{1}{2} s_1 i\hbar\right),\\
b=&\frac{1}{(i\hbar)^2}\left(-\frac{1}{120} s_1 \eps ^4-\frac{1}{10} s_1^2 \  i\hbar\eps ^2 -\left(\frac{2}{5} s_1^3 +\frac{1}{12} s_2\right) (i\hbar) ^2\right),\\
c =&\frac{1}{(i\hbar)^3}\left( \left(-\frac{1}{360} s_1^3-\frac{s_2}{1728}\right)\eps^6- \frac{24 s_1^4 +5 s_1 s_2 }{720}i \hbar \eps ^4-\frac{4608 s_1^5 +2400 s_2 s_1^2 +35 s_3 }{28800}(i\hbar)^2\eps ^2\right.\\
&\left.-\frac{ 2304 s_1^6+2400 s_2 s_1^3+105 s_3 s_1-500 s_2^2 }{7200}(i\hbar)^3\right).\end{align*}
Once plugged into (\ref{eq:quantum classification}), this parametrization provides the quantum correction to the density (\ref{eq:standard density}) or (\ref{eq:DZ standard density}) up to genus $3$. Rescaling $\eps^2\to\eps^2 \gamma$ and $\hbar \to \hbar \gamma$ to keep track of the genus, we obtain
\begin{equation*}
\begin{split}
\oG_1=\int &\left[\frac{u^3}{6}+\left(\left(-\frac{\eps ^2}{24}-\frac{s_1}{2}  i\hbar\right)u_1^2-\frac{i \hbar}{24} u\right)\gamma\right.\\
&\left.+\left(\left(-\frac{s_1}{120}  \eps ^4-\frac{s_1^2}{10}   i\hbar\eps ^2 -\frac{24 s_1^3 +5 s_2}{60} (i\hbar) ^2\right)u_2^2\right)\gamma^2\right.\\
&\left.+\left(\left( -\frac{s_1^3}{360} \eps^6-\frac{s_2}{1728}\eps^6- \frac{24 s_1^4 +5 s_1 s_2 }{720}i \hbar \eps^4-\frac{4608 s_1^5 +2400 s_2 s_1^2 +35 s_3 }{28800}(i\hbar)^2\eps ^2\right.\right.\right.\\
&\left.\left.\left.\hspace{0.7cm}-\frac{2304 s_1^6+2400 s_2 s_1^3+105 s_3 s_1-500 s_2^2 }{7200}(i\hbar)^3\right)u_2^3+\left(-\frac{s_1^2}{420}  \eps ^6-\frac{96 s_1^3+5 s_2}{2520}i \hbar \eps^4\right.\right.\right.\\
&\left.\left.\left.\hspace{0.7cm}- \frac{24 s_1^4+5 s_2 s_1}{105}(i\hbar)^2\eps^2-\frac{4608 s_1^5+2400 s_2 s_1^2+35 s_3 }{8400}(i\hbar ^3)\right)u_3^2\right)\gamma^3\right.\\
&\left.+O\left(\gamma^4\right)\right]dx.
\end{split}
\end{equation*}

%%%%%%%%%%%%%%%%%%%%%%%%%%%%%%%%%%%%%%%%%%%%%%%%%%%%%%%%%%%%%%%%%%%%%%%
%%%%%%%%%%%%%%%%%%%%%%%%%%%%%%%%%%%%%%%%%%%%%%%%%%%%%%%%%%%%%%%%%%%%%%%

\section{Geometric formula for the double ramification correlators}\label{section:geometric formula}

The goal of this section is to prove a geometric formula for the double ramification correlators. In Section~\ref{subsection:DR correlators} we recall the construction of these correlators from~\cite{BDGR16}. In Section~\ref{subsection:stable trees} we introduce certain cohomology classes in~$\oM_{g,n}$. They are used in the formulation of the geometric formula for the double ramification correlators in Section~\ref{subsection:geometric formula}. In Section~\ref{subsection:main formulas} we collect main formulas with the double ramification cycles and then use them in Section~\ref{subsection:proof of the geometric formula} for the proof of the geometric formula.

Let $c_{g,n}\colon V^{\otimes n}\to H^\even(\oM_{g,n},\mbC)$ be an arbitrary cohomological field theory, where $V$ is an $N$-dimensional vector space, $\eta$ is its metric tensor, $e_1,\ldots,e_N$ is a basis in $V$ such that $e_1$ is the unit.

\subsection{Double ramification correlators}\label{subsection:DR correlators}

Here we briefly recall the construction of the double ramification correlators from~\cite{BDGR16}. Define differential polynomials $h^\DR_{\alpha,d}\in\hcA^{[0]}_N$, $d\ge -1$, by
$$
h^\DR_{\alpha,d}:=\frac{\delta\og_{\alpha,d+1}}{\delta u^1}.
$$
For $1\le\alpha,\beta\le N$ and $p,q\ge 0$ there exists a unique differential polynomial $\Omega_{\alpha,p;\beta,q}^\DR\in\hcA^{[0]}_N$ such that
$$
\d_x\Omega_{\alpha,p;\beta,q}^\DR=\frac{\d h^\DR_{\alpha,p-1}}{\d t^\beta_q}=\left\{h^\DR_{\alpha,p-1},\og_{\beta,q}\right\}_{\eta\d_x}\quad\text{and}\quad\left.\Omega^\DR_{\alpha,p;\beta,q}\right|_{u^*_*=0}=0.
$$
The string solution $(u^\str)^\alpha(x,t^*_*,\eps)$ of the double ramification hierarchy is specified by the initial condition
$$
\left.(u^\str)^\alpha\right|_{t^*_*=0}=\delta^{\alpha,1}x.
$$
Let $(u^\str)^\gamma_n:=\d_x^n(u^\str)^\gamma$. Then there exists a unique power series $F^\DR(t^*_*,\eps)\in\mbC[[t^*_*,\eps^2]]$ such that
\begin{align}
&\frac{\d^2 F^\DR}{\d t^\alpha_p\d t^\beta_q}=\left.\left(\left.\Omega^\DR_{\alpha,p;\beta,q}\right|_{u^\gamma_n=(u^\str)^\gamma_n}\right)\right|_{x=0},\notag\\
&\frac{\d F^\DR}{\d t^1_0}=\sum_{n\ge 0}t^\alpha_{n+1}\frac{\d F^\DR}{\d t^\alpha_n}+\frac{1}{2}\eta_{\alpha\beta}t^\alpha_0 t^\beta_0,\label{eq:string for DR}\\
&\frac{\d F^\DR}{\d t^1_1}=\sum_{n\ge 0}t^\alpha_n\frac{\d F^\DR}{\d t^\alpha_n}+\eps\frac{\d F^\DR}{\d\eps}-2F^\DR+\eps^2\frac{N}{24},\label{eq:dilaton for DR}\\
&\left.\Coef_{\eps^2}F^\DR\right|_{t^*_*=0}=0.\notag
\end{align}
We see that the first equation here determines $F^\DR$ uniquely up to constant and linear terms in the variables $t^\alpha_p$. The other equations fix this ambiguity. The power series $F^\DR$ is called the double ramification potential. Let 
$$
F^\DR(t^*_*,\eps)=\sum_{g\ge 0}\eps^{2g}F^\DR_g(t^*_*).
$$ 
The double ramification correlators $\<\tau_{d_1}(e_{\alpha_1})\ldots\tau_{d_n}(e_{\alpha_n})\>_g^\DR$ are defined as the coefficients of the expansion of $F^\DR_g$:
$$
F^\DR_g=\sum_{n\ge 0}\sum_{d_1,\ldots,d_n\ge 0}\<\tau_{d_1}(e_{\alpha_1})\ldots\tau_{d_n}(e_{\alpha_n})\>_g^\DR\frac{t^{\alpha_1}_{d_1}\ldots t^{\alpha_n}_{d_n}}{n!}.
$$
In~\cite[Sections 6.6,~6.7]{BDGR16} we proved that a double ramification correlator $\<\tau_{d_1}(e_{\alpha_1})\ldots\tau_{d_n}(e_{\alpha_n})\>_g^\DR$ vanishes unless 
$$
2g-2+n>0\quad\text{and}\quad 2g-1\le\sum d_i\le 3g-3+n.
$$

%%%%%%%%%%%%%%%%%%%%%%%%%%%%%%%%%%%%%%%%%%%%%%%%%%%%%%%%%%%%%%%%%%%%%

\subsection{Stable trees and cohomology classes in $\oM_{g,n}$}\label{subsection:stable trees}

In this section we collect notations and definitions related to stable graphs that will be needed for the formulation of our geometric formula for the double ramification correlators. We will use the notations from~\cite[Sections 0.2 and 0.3]{PPZ15}.

By stable tree we mean a stable graph
$$
\Gamma=(V,H,L,g\colon V\to\mbZ_{\ge 0},v\colon H\to V, \iota\colon H\to H),
$$
that is a tree. Let $H^e(\Gamma):=H(\Gamma)\backslash L(\Gamma)$. A path in $\Gamma$ is a sequence of pairwise distinct vertices $v_1,v_2,\ldots,v_k\in V$, $v_i\ne v_j$ for $i\ne j$, such that for any $1\le i\le k-1$ the vertices $v_i$ and $v_{i+1}$ are connected by an edge. For a vertex $v\in V(\Gamma)$ define a number $r(v)$ by
$$
r(v):=2g(v)-2+n(v).
$$ 

A stable rooted tree is a pair $(\Gamma,v_0)$, where $\Gamma$ is a stable tree and $v_0\in V(\Gamma)$. The vertex~$v_0$ is called the root. Denote by~$H_+(\Gamma)$ the set of half-edges of $\Gamma$ that are directed away from the root $v_0$. Clearly, $L(\Gamma)\subset H_+(\Gamma)$. Let $H^e_+(\Gamma):=H_+(\Gamma)\backslash L(\Gamma)$. A vertex $w$ is called a descendant of a vertex $v$, if $v$ is on the unique path from the root $v_0$ to $w$. Note that according to our definition the vertex $v$ is a descendant of itself. Denote by $\Desc[v]$ the set of all descendants of~$v$.

Let $g\ge 0$ and $m,n\ge 1$. Denote by $\ST^m_{g,n+1}$ the set of stable trees of genus~$g$ with $m$ vertices and with~$n+1$ legs marked by numbers $0,1,\ldots,n$. For a stable tree $\Gamma\in\ST^m_{g,n+1}$ denote by~$l_i(\Gamma)$ the leg in~$\Gamma$ that is marked by~$i$. We will always choose the vertex $v(l_0(\Gamma))$ as a root of~$\Gamma$. In this way a stable tree from~$\ST^m_{g,n+1}$ automatically becomes a stable rooted tree. For a leg $l\in L(\Gamma)$ denote by $0\le i(l)\le n$ its marking.

Consider a stable tree $\Gamma\in\ST^m_{g,n+1}$. We have the associated moduli space
$$
\oM_{\Gamma}:=\prod_{v\in V}\oM_{g(v),n(v)}
$$
and the canonical morphism 
$$
\xi_\Gamma\colon\oM_{\Gamma}\to\oM_{g(\Gamma),|L(\Gamma)|}.
$$
Consider integers $a_0,a_1,\ldots,a_n$ such that $a_0+a_1+\ldots+a_n=0$. To each half-edge $h\in H(\Gamma)$ we assign an integer~$a(h)$ in such a way that the following conditions hold:
\begin{itemize}

\item[a)] If $h\in L(\Gamma)$, then $a(h)=a_{i(l)}$;

\item[b)] If $h\in H^e(\Gamma)$, then $a(h)+a(\iota(h))=0$;

\item[c)] For any vertex $v\in V(\Gamma)$, we have $\sum_{h\in H[v]}a(h)=0$.

\end{itemize}
Since the graph~$\Gamma$ is a tree, it is easy to see that such a function $a\colon H(\Gamma)\to\mbZ$ exists and is uniquely determined by the numbers $a_0,a_1,\ldots,a_n$. For each moduli space $\oM_{g(v),n(v)}$, $v\in V(\Gamma)$, the numbers $a(h)$, $h\in H[v]$, define the double ramification cycle
$$
\DR_{g(v)}\left((a(h))_{h\in H[v]}\right)\in H^{2g(v)}(\oM_{g(v),n(v)},\mbQ).
$$
If we multiply all these cycles, we get the class
$$
\prod_{v\in V(\Gamma)}\DR_{g(v)}\left((a(h))_{h\in H[v]}\right)\in H^{2g}(\oM_\Gamma,\mbQ).
$$
We define a class $\DR_\Gamma(a_0,a_1,\ldots,a_n)\in H^{2(g+m-1)}(\oM_{g,n+1},\mbQ)$ by 
\begin{gather*}
\DR_\Gamma(a_0,a_1,\ldots,a_n):=\left(\prod_{h\in H^e_+(\Gamma)}a(h)\right)\cdot \xi_{\Gamma*}\left(\prod_{v\in V(\Gamma)}\DR_{g(v)}\left((a(h))_{h\in H[v]}\right)\right).
\end{gather*}
Note that in the case when the valency of some vertex $v$ in $\Gamma$ is equal to one, the class $\DR_\Gamma(a_0,a_1,\ldots,a_n)$ is equal to zero. This happens because, if $h$ is the half-edge incident to~$v$, then, obviously, $a(h)=0$. From Hain's formula~\cite{Hai11} it follows that for an arbitrary stable tree $\Gamma\in\ST^m_{g,n+1}$ the class
$$
\lambda_g\DR_\Gamma\left(-\sum_{i=1}^n a_i,a_1,\ldots,a_n\right)\in H^{2(2g+m-1)}(\oM_{g,n+1},\mbQ)
$$
is a polynomial in $a_1,\ldots,a_n$ homogeneous of degree $2g+m-1$. 

For a stable tree $\Gamma\in\ST_{g,n+1}^m$ define a combinatorial coefficient $C(\Gamma)$ by
$$
C(\Gamma):=\prod_{v\in V(\Gamma)}\frac{r(v)}{\sum_{\tv\in\Desc[v]}r(\tv)}.
$$

%%%%%%%%%%%%%%%%%%%%%%%%%%%%%%%%%%%%%%%%%%%%%%%%%%%%%%%%%%%%%%%%%%%%%

\subsection{Geometric formula for the correlators}\label{subsection:geometric formula}

Recall that a double ramification correlator $\<\tau_{d_1}(e_{\alpha_1})\ldots\tau_{d_n}(e_{\alpha_n})\>^{\DR}_g$ vanishes unless $\sum d_i\ge 2g-1$ (see~\cite[Section 6.7]{BDGR16}). 

\begin{theorem}\label{theorem:main geometric formula}
Suppose $g\ge 0$, $n\ge 1$ and $2g-2+n>0$. Let $d\ge 2g-1$ and $1\le\alpha_1,\ldots,\alpha_n\le N$. Then we have the following equality of polynomials in $a_1,\ldots,a_n$ of degree $d$:
\begin{multline}\label{eq:main geometric formula}
\sum_{\substack{d_1,\ldots,d_n\ge 0\\\sum d_i=d}}\<\tau_{d_1}(e_{\alpha_1})\ldots\tau_{d_n}(e_{\alpha_n})\>^{\DR}_ga_1^{d_1}\ldots a_n^{d_n}=\\
=\frac{1}{\sum a_i}\sum_{\Gamma\in\ST^{d-2g+2}_{g,n+1}}C(\Gamma)\int_{\oM_{g,n+1}}\DR_\Gamma\left(-\sum a_i,a_1,\ldots,a_n\right)\lambda_g c_{g,n+1}\left(e_1\otimes\otimes_{i=1}^n e_{\alpha_i}\right).
\end{multline}
\end{theorem}

Note that in the case $d=2g-1$ formula~\eqref{eq:main geometric formula} becomes particularly simple:
\begin{multline}\label{eq:geometric formula for k=0}
\sum_{\substack{d_1,\ldots,d_n\ge 0\\\sum d_i=2g-1}}\<\tau_{d_1}(e_{\alpha_1})\ldots\tau_{d_n}(e_{\alpha_n})\>^{\DR}_ga_1^{d_1}\ldots a_n^{d_n}=\\
=\frac{1}{\sum a_i}\int_{\oM_{g,n+1}}\DR_g\left(-\sum a_i,a_1,\ldots,a_n\right)\lambda_g c_{g,n+1}\left(e_1\otimes\otimes_{i=1}^n e_{\alpha_i}\right).
\end{multline}
We will prove Theorem~\ref{theorem:main geometric formula} in Section~\ref{subsection:proof of the geometric formula}.

%%%%%%%%%%%%%%%%%%%%%%%%%%%%%%%%%%%%%%%%%%%%%%%%%%%%%%%%%%%%%%%%%%%%%

\subsection{Main formulas with the double ramification cycles}\label{subsection:main formulas}

Here we collect main formulas with the double ramification cycles that we will use later.

\subsubsection{Double ramification cycle and fundamental class} Suppose $\pi\colon\oM_{g,n+g}\to\oM_{g,n}$ is the forgetful map, that forgets the last~$g$ marked points. Then we have \cite[Example 3.7]{BSSZ15}
\begin{gather}\label{eq:DR and fundamental class}
\pi_*\DR_g(a_1,\ldots,a_{n+g})=g!a_{n+1}^2\ldots a_{n+g}^2[\oM_{g,n}].
\end{gather}

\subsubsection{Divisibility properties}\label{subsubsection:divisibility properties} Let $g,n\ge 1$. Suppose $\pi\colon\oM_{g,n+1}\to\oM_{g,n}$ is the forgetful map that forgets the last marked point. Then the polynomial class 
\begin{gather*}
\left.\pi_*\DR_g\left(-\sum a_i,a_1,a_2,\ldots,a_n\right)\right|_{\cM_{g,n}^\ct}\in H^{2g-2}(\cM^{\ct}_{g,n},\mbQ)
\end{gather*}
is divisible by $a_n^2$.

Suppose $g,n,m\ge 1$. Then we have \cite[Section 5.1]{BDGR16}
\begin{align}
&\int_{\DR_g(-\sum a_i-\sum b_j,a_1,\ldots,a_n,b_1,\ldots,b_m)}\lambda_g\psi_2^d c_{g,n+m+1}(\otimes_{i=1}^{n+1} e_{\alpha_i}\otimes e_1^m)=\label{eq:divisibility}\\
=&\begin{cases}
\int_{\DR_g(-\sum a_i-\sum b_j,a_1+\sum b_j,a_2,\ldots,a_n)}\lambda_g\psi_2^{d-m}c_{g,n+1}(\otimes_{i=1}^{n+1}e_{\alpha_i})+O(b_1^2)+\ldots+O(b_m^2),&\text{if $d\ge m$};\\
O(b_1^2)+\ldots+O(b_m^2),&\text{if $d<m$}.
\end{cases}
\notag
\end{align}

\subsubsection{Double ramification cycle times a $\psi$-class}

Here we recall the formula from~\cite{BSSZ15} for the product of the double ramification cycle with a $\psi$-class. Denote by 
$$
\gl_k\colon\oM_{g_1,n_1+k}\times\oM_{g_2,n_2+k}\to\oM_{g_1+g_2+k-1,n_1+n_2}
$$
the gluing map that corresponds to gluing a curve from~$\oM_{g_1,n_1+k}$ to a curve from~$\oM_{g_2,n_2+k}$ along the last $k$ marked points on the first curve and the last $k$ marked points on the second curve. Suppose $n,m\ge k\ge 1$ and $a_1,\ldots,a_n$ and $b_1,\ldots,b_m$ are lists of integers with vanishing sums. Let
\begin{align*}
&\DR_{g_1}(a_1,\ldots,a_n)\boxtimes_k\DR_{g_2}(b_1,\ldots,b_m):=\\
=&\gl_{k*}\left(\DR_{g_1}(a_1,\ldots,a_n)\times\DR_{g_2}(b_1,\ldots,b_m)\right)\in H^{2(g_1+g_2+k)}(\oM_{g_1+g_2+k-1,n+m-2k},\mbQ).
\end{align*}

Let $a_1,\ldots,a_n$ be a list of integers with vanishing sum. Assume that $a_s \ne 0$ for some $1\le s\le n$. Then we have~\cite[Theorem 4]{BSSZ15}
\begin{align}
&a_s\psi_s \DR_g(a_1, \dots, a_n)=\label{eq:DR times psi}\\
=&\sum_{\substack{I\sqcup J=\{1,\ldots,n\}\\\sum_{i\in I}a_i>0}}\sum_{p\ge 1}\sum_{\substack{g_1,g_2\ge 0\\g_1+g_2+p-1=g}}\sum_{\substack{k_1,\ldots,k_p\ge 1\\\sum k_j=\sum_{i\in I}a_i}}\frac{\rho}{r}\frac{\prod_{i=1}^p k_i}{p!}\DR_{g_1}(a_I,-k_1,\ldots,-k_p)\boxtimes_p\DR_{g_2}(a_J,k_1,\ldots,k_p),\notag
\end{align}
where $a_I$ denotes the list $(a_i)_{i\in I}$, $r=2g-2+n$ and
$$
\rho=
\begin{cases} 
2g_2-2+|J|+p,    &\text{if $s\in I$},\\
-(2g_1-2+|I|+p), &\text{if $s\in J$}.
\end{cases}
$$

%%%%%%%%%%%%%%%%%%%%%%%%%%%%%%%%%%%%%%%%%%%%%%%%%%%%%%%%%%%%%%%%%%%%%

\subsection{Proof of the geometric formula}\label{subsection:proof of the geometric formula}

In this section we prove Theorem~\ref{theorem:main geometric formula}. The plan is the following. In Section~\ref{subsubsection:more about stable trees} we put combinatorial definitions and constructions that we will need for the proof. In Section~\ref{subsubsection:map phi} we show how to use the combinatorial map~$\phi$ defined in Section~\ref{subsubsection:more about stable trees} in order to simplify the geometric formula for the double ramification correlators from our previous work~\cite{BDGR16}. From this simplification we will see that Theorem~\ref{theorem:main geometric formula} follows from a certain relation in the cohomology of the moduli space of curves. This relation is proved in Section~\ref{subsubsection:relation}.

\subsubsection{More about stable trees}\label{subsubsection:more about stable trees}

In this section we collect a combinatorial material related to stable trees that we will need in the proof of the geometric formula. We will partly repeat the material from~\cite[Section 6.6.2]{BDGR16}.

Let $\Gamma\in\ST^m_{g,n+1}$. Introduce the following notations:
$$
L'(\Gamma):=L(\Gamma)\backslash\{l_0(\Gamma)\},\qquad H_+'(\Gamma):=H_+(\Gamma)\backslash\{l_0(\Gamma)\}.
$$
Clearly, for any vertex $v\in V(\Gamma)$ the set $H[v]\backslash H_+'[v]$ consists of exactly one element. The stable tree $\Gamma$ will be called admissible, if the following two conditions are satisfied:
\begin{itemize}
\item[a)] For any vertex $v\in V(\Gamma)$ we have $|L'[v]|\ge 1$;
\item[b)] For any two distinct vertices $v_1,v_2\in V(\Gamma)$ such that $v_2$ is a descendant of $v_1$ we have
$$
\min_{l\in L'[v_1]}i(l)<\min_{l\in L'[v_2]}i(l).
$$
\end{itemize}
The set of all admissible stable trees will be denoted by $\AST^m_{g,n+1}\subset\ST^m_{g,n+1}$.

Consider a stable tree $\Gamma\in\ST^m_{g,n+1}$ a vertex $v\in V(\Gamma)$ and a half-edge $h\in H^e_+[v]$. Denote by~$\Gamma_h$ the stable rooted tree formed by the descendants of $v(\iota(h))$ and all half-edges incident to them together with the vertex $v(\iota(h))$ as a root (see Fig.~\ref{fig:Gammah}).
\begin{figure}[t]
\begin{tikzpicture}[scale=0.4]
\draw (0,0) -- (-3,0);
\draw (0,0) -- (3,4);
\draw (0,0) -- (3,0);
\draw (0,0) -- (3,-3.5);
\draw (3,0) -- (6,1.5);
\draw (3,0) -- (6,0);
\draw (3,0) -- (6,-1.5);
\draw (3,4) -- (6,5.5);
\draw (3,4) -- (6,3.5);
\draw (6,5.5) -- (9,7);
\draw (6,5.5) -- (9,5.5);
\draw (6,5.5) -- (9,4);

\fill[white] (0,0) circle(7mm);
\draw (0,0) circle(7mm);
\fill[white] (3,0) circle(7mm);
\draw (3,0) circle(7mm);
\fill[white] (3,4) circle(7mm);
\draw (3,4) circle(7mm);
\fill[white] (6,5.5) circle(7mm);
\draw (6,5.5) circle(7mm);

\draw[dashed] (4.5,0) circle(2.8);

\coordinate [label=center: $\Gamma_h$] () at (7.9,-1.5);
\coordinate [label=above: $\scriptstyle{h}$] () at (1,-0.15);

\end{tikzpicture}
\caption{}
\label{fig:Gammah}
\end{figure}

Let us define splitting and contracting operations on stable trees. Consider a stable tree $\Gamma\in\ST^m_{g,n+1}$, a vertex $v\in V(\Gamma)$, a subset $I\subset H'_+[v]$ and an integer $0\le g_1\le g(v)$ such that $2g_1+|I|>0$ and $2g_2+|I^c|-1>0$, where $I^c:=H'_+[v]\backslash I$ and $g_2:=g(v)-g_1$. We define a stable tree~$\Spl(\Gamma,v,g_1,I)\in\ST^{m+1}_{g,n+1}$ in the following way. We split the vertex~$v$ in two vertices of genera~$g_1$ and~$g_2$ respectively, connect them by an edge, attach the half-edge from $H[v]\backslash H_+'[v]$ to the first vertex and then attach the half-edges from the set $I$ to the first vertex and the half-edges from the set $I^c$ to the second vertex (see Fig.~\ref{fig:split}). 
\begin{figure}[t]
\begin{tikzpicture}[scale=0.4]
\draw (0,0) -- (-5,0);
\draw (0,0) -- (3.7,3.5);
\draw (0,0) -- (4.8,1.5);
\draw (0,0) -- (4.8,-1.5);
\draw (0,0) -- (3.7,-3.5);

\fill[white] (0,0) circle(1);
\draw (0,0) circle(1) node {$g(v)$};
\draw[dashed] (0,0) circle(4);

\begin{scope}[shift={(17,0)}]
\draw (0,0) -- (-5,0);
\draw (0,0) -- (6,2);
\draw (0,0) -- (4.8,3.9);
\draw (0,0) -- (2,-2);
\draw (2,-2) -- (6,-3);
\draw (2,-2) -- (5,-4.7);

\fill[white] (0,0) circle(7mm);
\draw (0,0) circle(7mm) node {$g_1$};
\fill[white] (2,-2) circle(7mm);
\draw (2,-2) circle(7mm) node {$g_2$};
\draw[dashed] (1,-0.5) circle(4.5);
\end{scope}

\begin{scope}[shift={(3,0)}]
\draw [thick] (4.5,0) -- (6.5,0);
\draw [thick] (4.5,0.2) -- (4.5,-0.2);
\draw [thick] (6.5,0) -- (6.2,0.2);
\draw [thick] (6.5,0) -- (6.2,-0.2);
\coordinate [label=above: $\mathrm{Spl}$] () at (5.5,0);
\end{scope}

\end{tikzpicture}
\caption{Splitting operation}
\label{fig:split}
\end{figure}
This is the splitting operation. 

Let us define a contracting operation. Suppose $m\ge 2$. Let $\Gamma\in\ST^m_{g,n+1}$, $v\in V(\Gamma)$ and $h\in H^e[v]$. A stable tree $\Con(\Gamma,v,h)\in\ST^{m-1}_{g,n+1}$ is defined simply by contracting the edge corresponding to the half-edges $h$ and $\iota(h)$. 

A modified stable tree is a stable tree $\Gamma$ where we split the set of legs in two subsets: the set of legs of the first type and the set of legs of the second type, where we require that each vertex of the tree is incident to exactly one leg of the second type. The set of legs of the first type will be denoted by~$L_1(\Gamma)$ and the set of legs of the second type will be denoted by $L_2(\Gamma)$.  

For $g\ge 0$ and $m,n\ge 1$ denote by $\MST^m_{g,n+1}$ the set of modified stable trees of genus~$g$ with~$m$ vertices and with~$(m+n+1)$ legs. We mark the legs of first type by numbers $0,1,\ldots,n$ and the legs of the second type by numbers $n+1,\ldots,n+m$. In the same way, as for usual stable trees, for a modified stable tree $\Gamma\in\MST^m_{g,n+1}$ we use the notation $l_i(\Gamma)$ for the leg marked by~$i$ and the notation $i(l)$ for the marking of a leg $l\in L(\Gamma)$. We will always choose the vertex $v(l_0(\Gamma))$ as a root of $\Gamma$. In this way a modified stable tree from $\MST^m_{g,n+1}$ automatically becomes a stable rooted tree. An example of a modified stable tree from~$\MST^m_{g,n+1}$ is shown on the left-hand side of Fig.~\ref{fig:map phi}.
\begin{figure}[t]
\begin{tikzpicture}[scale=0.4]
\draw (0,0) -- (-3,0);
\draw (0,0) -- (3,4);
\draw (0,0) -- (3,0);
\draw (0,0) -- (3,-2);
\draw (0,0) -- (3,-3.5);
\draw[double] (0,0) -- (3,-5);
\draw (3,0) -- (6,1.5);
\draw (3,0) -- (6,0);
\draw[double] (3,0) -- (6,-1.5);
\draw (3,4) -- (6,5.5);
\draw (3,4) -- (6,4);
\draw[double] (3,4) -- (6,2.5);
\draw (6,5.5) -- (9,6.25);
\draw (6,5.5) -- (9,7.75);
\draw (6,5.5) -- (9,4.75);
\draw[double] (6,5.5) -- (9,3.25);

\fill[white] (0,0) circle(7mm);
\draw (0,0) circle(7mm) node {$g_1$};
\fill[white] (3,0) circle(7mm);
\draw (3,0) circle(7mm) node {$g_3$};
\fill[white] (3,4) circle(7mm);
\draw (3,4) circle(7mm) node {$g_2$};
\fill[white] (6,5.5) circle(7mm);
\draw (6,5.5) circle(7mm) node {$0$};

\coordinate [label=left: $0$] (B) at (-3,0);
\coordinate [label=right: $1$] (B) at (9,7.75);
\coordinate [label=right: $2$] (B) at (9,6.25);
\coordinate [label=right: $3$] (B) at (9,4.75);
\coordinate [label=right: $12$] (B) at (9,3.25);
\coordinate [label=right: $4$] (B) at (6,4);
\coordinate [label=right: $10$] (B) at (6,2.5);
\coordinate [label=right: $5$] (B) at (6,1.5);
\coordinate [label=right: $6$] (B) at (6,0);
\coordinate [label=right: $11$] (B) at (6,-1.5);
\coordinate [label=right: $7$] (B) at (3,-2);
\coordinate [label=right: $8$] () at (3,-3.5);
\coordinate [label=right: $9$] () at (3,-5);

\begin{scope}[shift={(20,0)}]
\draw (0,0) -- (-3,0);
\draw (0,0) -- (3,4);
\draw (0,0) -- (3,0);
\draw (0,0) -- (3,-3);
\draw (3,0) -- (6,0);
\draw (3,4) -- (6,5.5);
\draw (3,4) -- (6,2.5);

\fill[white] (0,0) circle(7mm);
\draw (0,0) circle(7mm) node {$g_1$};
\fill[white] (3,0) circle(7mm);
\draw (3,0) circle(7mm) node {$g_3$};
\fill[white] (3,4) circle(7mm);
\draw (3,4) circle(7mm) node {$g_2$};

\coordinate [label=left: $0$] (B) at (-3,0);
\coordinate [label=right: $4$] (B) at (6,5.5);
\coordinate [label=right: $2$] (B) at (6,2.5);
\coordinate [label=right: $3$] (B) at (6,0);
\coordinate [label=right: $1$] () at (3,-3);
\end{scope}

\begin{scope}[shift={(7,0)}]
\draw [thick] (4.5,0) -- (6.5,0);
\draw [thick] (4.5,0.2) -- (4.5,-0.2);
\draw [thick] (6.5,0) -- (6.2,0.2);
\draw [thick] (6.5,0) -- (6.2,-0.2);
\coordinate [label=above: $\phi$] () at (5.5,0);
\end{scope}

\end{tikzpicture}
\caption{Map $\phi\colon\MST^{m,e}_{g,n+1}\to\ST^{m-e}_{g,m+1}$}
\label{fig:map phi}
\end{figure}
The legs of the second type are drawn by double lines. The reader can see that in our example $n=8$ and $m=4$.

Consider a modified stable tree $\Gamma\in\MST^m_{g,n+1}$. Define a function $p\colon V(\Gamma)\to\{1,\ldots,m\}$ by $p(v):=i-n$, where $i$ is the marking of a unique leg of the second type incident to $v$. The modified stable tree~$\Gamma$ is called admissible, if for any two distinct vertices $v_1,v_2\in V(\Gamma)$ such that~$v_2$ is a descendant of~$v_1$, we have $p(v_2)>p(v_1)$. The subset of admissible modified stable trees will be denoted by $\AMST^m_{g,n+1}\subset\MST^m_{g,n+1}$. Note that the modified stable tree on the left-hand side of Fig.~\ref{fig:map phi} is admissible.

Consider a modified stable tree $\Gamma\in\MST^m_{g,n+1}$ and integers $a_0,a_1,\ldots,a_n$ with vanishing sum. Define a function $a\colon H(\Gamma)\to\mbZ$ by the properties
\begin{itemize}

\item[a)] If $h\in L_1(\Gamma)$, then $a(h)=a_{i(h)}$;

\item[b)] If $h\in L_2(\Gamma)$, then $a(h)=0$;

\item[c)] If $h\in H^e(\Gamma)$, then $a(h)+a(\iota(h))=0$;

\item[d)] For any vertex $v\in V(\Gamma)$, we have $\sum_{h\in H[v]}a(h)=0$.

\end{itemize}
In the same way, as in Section~\ref{subsection:stable trees}, we define the class $\DR_\Gamma(a_0,\ldots,a_n)\in H^{2(g+m-1)}(\oM_{g,m+n+1},\mbQ)$.

Suppose $\Gamma\in\MST^m_{g,n+1}$. It is useful to introduce the notations
\begin{align}
&L'_1(\Gamma):=L_1(\Gamma)\backslash\{l_0(\Gamma)\},\notag\\
&H'(\Gamma):=H^e(\Gamma)\cup\{l_0(\Gamma)\}.\label{eq:H'}
\end{align}
Clearly, for any vertex $v\in V(\Gamma)$ we have $|H[v]\backslash L_1'[v]|\ge 2$. A vertex $v\in V(\Gamma)$ will be called exceptional, if $g(v)=0$ and $|H[v]\backslash L_1'[v]|=2$. Otherwise, it will be called regular. The reader can see that one vertex in the graph on the left-hand side of Fig.~\ref{fig:map phi} is exceptional. Denote by~$V^\exc(\Gamma)$ and~$V^\reg(\Gamma)$ the sets of exceptional and regular vertices in~$\Gamma$ respectively. An edge in $\Gamma$ that is incident to an exceptional vertex will be called exceptional. The set of modified stable graphs with~$e$ exceptional vertices will be denoted by $\MST^{m,e}_{g,n+1}\subset\MST^m_{g,n+1}$. 

Consider $g\ge 0$ and $m,n\ge 1$ such that $2g+m-1>0$. Note that for any modified stable tree $\Gamma\in\MST^m_{g,n+1}$ the root is regular. So we have $|V^\exc(\Gamma)|\le m-1$. Let $0\le e\le m-1$. Let us define a map 
$$
\phi\colon\MST^{m,e}_{g,n+1}\to\ST^{m-e}_{g,m+1}
$$
in the following way. Suppose $\Gamma\in\MST^{m,e}_{g,n+1}$. We construct the graph $\phi(\Gamma)$ by contracting all exceptional edges and then by throwing away all legs from $L_1'(\Gamma)$. It is easy to see that the graph $\phi(\Gamma)$ has $m-e$ vertices and $m+1$ legs. We only have to specify how we mark them. A leg $l$ in $\phi(\Gamma)$ corresponds to some leg $l_i(\Gamma)$ in $\Gamma$, where $i=0$ or $n+1\le i\le m+n$. If $i=0$, then we mark $l$ by $0$ and if $n+1\le i\le m+n$, then we mark $l$ by $i-n$. An example of the action of the map $\phi$ is shown in Fig.~\ref{fig:map phi}. It is easy to see that for any $\Gamma\in\AMST^{m,e}_{g,n+1}$ we have $\phi(\Gamma)\in\AST^{m-e}_{g,m+1}$. So we have the map 
$$
\phi\colon\AMST^{m,e}_{g,n+1}\to\AST^{m-e}_{g,m+1}.
$$

\subsubsection{Map $\phi$ and integrals over double ramification cycles}\label{subsubsection:map phi}

The string equation~\eqref{eq:string for DR} for the double ramification correlators implies that
$$
\sum_{\substack{d_1,\ldots,d_n\ge 0\\\sum d_i=d}}\<\tau_{d_1}(e_{\alpha_1})\ldots\tau_{d_n}(e_{\alpha_n})\>^\DR_g\prod_{i=1}^n a_i^{d_i}=\frac{1}{\sum a_i}\sum_{\substack{d_1,\ldots,d_n\ge 0\\\sum d_i=d+1}}\<\tau_0(e_1)\tau_{d_1}(e_{\alpha_1})\ldots\tau_{d_n}(e_{\alpha_n})\>^\DR_g\prod_{i=1}^n a_i^{d_i}.
$$
Therefore, formula~\eqref{eq:main geometric formula} is equivalent to
\begin{multline}\label{eq:main formula with insertion}
\sum_{\substack{d_1,\ldots,d_n\ge 0\\\sum d_i=d+1}}\<\tau_0(e_1)\tau_{d_1}(e_{\alpha_1})\ldots\tau_{d_n}(e_{\alpha_n})\>^{\DR}_ga_1^{d_1}\ldots a_n^{d_n}=\\
=\sum_{\Gamma\in\ST^{d-2g+2}_{g,n+1}}C(\Gamma)\int_{\oM_{g,n+1}}\DR_\Gamma\left(-\sum a_i,a_1,\ldots,a_n\right)\lambda_g c_{g,n+1}\left(e_1\otimes\otimes_{i=1}^n e_{\alpha_i}\right).
\end{multline}
In~\cite[Section 6.6.3]{BDGR16} we proved that the correlator $\<\tau_0(e_1)\tau_{d_1}(e_{\alpha_1})\ldots\tau_{d_n}(e_{\alpha_n})\>^\DR_g$ is equal to the coefficient of $b_1 b_2\ldots b_{2g+n-1}$ in the polynomial
$$
\frac{1}{(2g+n-1)!}\sum_{\Gamma\in\AMST^n_{g,2g+n}}\int_{\DR_\Gamma(-\sum b_i,b_1,\ldots,b_{2g+n-1})}\lambda_g c_{g,2g+2n}(e_1^{2g+n}\otimes\otimes_{i=1}^n e_{\alpha_i})\prod_{i=1}^n\psi_{2g+n-1+i}^{d_i}.
$$
On the left-hand side of Fig.~\ref{fig:transformation of the integral} we schematically represent an example of an integral from this formula. Note that the modified stable tree $\Gamma$ in this example coincides with the modified stable tree on the left-hand side of Fig.~\ref{fig:map phi}.
\begin{figure}[t]
\begin{tikzpicture}[scale=0.4]

\draw plot [smooth,tension=1.8] coordinates {(0,6) (-7,0) (0,-6)};
\draw plot [smooth,tension=1.8] coordinates {(0,6) (1.5,4.5) (0,3)};
\draw plot [smooth,tension=1.8] coordinates {(0,3) (1.5,1.5) (0,0)};
\draw plot [smooth,tension=1.8] coordinates {(0,0) (1.5,-1.5) (0,-3)};
\draw plot [smooth,tension=1.8] coordinates {(0,-3) (1.5,-4.5) (0,-6)};

\draw plot [smooth,tension=1.8] coordinates {(6,3.1) (1.5,4.5) (6,5.9)};
\draw plot [smooth,tension=1.8] coordinates {(6,3.1) (6.7,3.8) (6,4.5)};
\draw plot [smooth,tension=1.8] coordinates {(6,4.5) (6.7,5.2) (6,5.9)};

\draw plot [smooth,tension=1.8] coordinates {(6,0.1) (1.5,1.5) (6,2.9)};
\draw plot [smooth,tension=1.8] coordinates {(6,0.1) (6.7,0.8) (6,1.5)};
\draw plot [smooth,tension=1.8] coordinates {(6,1.5) (6.7,2.2) (6,2.9)};

\draw plot [smooth,tension=1.8] coordinates {(11.2,3.7) (6.7,5.2) (11.2,6.7)};
\draw plot [smooth,tension=1.8] coordinates {(11.2,3.7) (11.7,4.2) (11.2,4.7)};
\draw plot [smooth,tension=1.8] coordinates {(11.2,4.7) (11.7,5.2) (11.2,5.7)};
\draw plot [smooth,tension=1.8] coordinates {(11.2,5.7) (11.7,6.2) (11.2,6.7)};

\coordinate [label=right: $\scriptstyle{\psi^2 e_{\alpha_4}}$] (B) at (7.8,5.5);
\coordinate [label=center: $\scriptstyle{\bullet}$] (B) at (7.8,5.5);
\coordinate [label=left: $\scriptstyle{a_3}$] (B) at (11.7,4.2);
\coordinate [label=right: $\scriptstyle{e_1}$] (B) at (11.7,4.2);
\coordinate [label=left: $\scriptstyle{a_2}$] (B) at (11.7,5.2);
\coordinate [label=right: $\scriptstyle{e_1}$] (B) at (11.7,5.2);
\coordinate [label=left: $\scriptstyle{a_1}$] (B) at (11.7,6.2);
\coordinate [label=right: $\scriptstyle{e_1}$] (B) at (11.7,6.2);
\coordinate [label=left: $\scriptstyle{a_4}$] (B) at (6.7,3.8);
\coordinate [label=right: $\scriptstyle{e_1}$] (B) at (6.7,3.8);
\coordinate [label=left: $\scriptstyle{a_5}$] (B) at (6.7,2.2);
\coordinate [label=right: $\scriptstyle{e_1}$] (B) at (6.7,2.2);
\coordinate [label=left: $\scriptstyle{a_6}$] (B) at (6.7,0.8);
\coordinate [label=right: $\scriptstyle{e_1}$] (B) at (6.7,0.8);
\coordinate [label=left: $\scriptstyle{a_7}$] (B) at (1.5,-1.5);
\coordinate [label=right: $\scriptstyle{e_1}$] (B) at (1.5,-1.5);
\coordinate [label=left: $\scriptstyle{a_8}$] (B) at (1.5,-4.5);
\coordinate [label=right: $\scriptstyle{e_1}$] (B) at (1.5,-4.5);

\draw (-4.5,0) arc (-126.87:-90+36.87:2.5);
\draw plot [smooth,tension=1.2] coordinates {(-4,-0.29) (-3,0) (-2,-0.29)};

\coordinate [label=right: $\scriptstyle{\psi^{d_1} e_{\alpha_1}}$] (B) at (-3.5,-2.3);
\coordinate [label=center: $\scriptstyle{\bullet}$] (B) at (-3.5,-2.3);

\begin{scope}[shift={(6,4.5)},scale=0.7]
\draw (-4.5,0) arc (-126.87:-90+36.87:2.5);
\draw plot [smooth,tension=1.2] coordinates {(-4,-0.29) (-3,0) (-2,-0.29)};
\end{scope}
\begin{scope}[shift={(6,1.5)},scale=0.7]
\draw (-4.5,0) arc (-126.87:-90+36.87:2.5);
\draw plot [smooth,tension=1.2] coordinates {(-4,-0.29) (-3,0) (-2,-0.29)};
\end{scope}

\coordinate [label=right: $\scriptstyle{\psi^{d_2} e_{\alpha_2}}$] (B) at (2.9,5);
\coordinate [label=center: $\scriptstyle{\bullet}$] (B) at (2.9,5);

\coordinate [label=right: $\scriptstyle{\psi^{d_3} e_{\alpha_3}}$] (B) at (2.9,2);
\coordinate [label=center: $\scriptstyle{\bullet}$] (B) at (2.9,2);

\coordinate [label=center: $\scriptstyle{\lambda_{g_1}}$] (B) at (-4,-5.7);
\coordinate [label=center: $\scriptstyle{\lambda_{g_2}}$] (B) at (4,6.3);
\coordinate [label=center: $\scriptstyle{\lambda_{g_3}}$] (B) at (4,-0.3);

\begin{scope}[shift={(20,0)}]

\draw plot [smooth,tension=1.8] coordinates {(0,6) (-7,1.5) (0,-3)};
\draw plot [smooth,tension=1.8] coordinates {(0,6) (1.5,4.5) (0,3)};
\draw plot [smooth,tension=1.8] coordinates {(0,3) (1.5,1.5) (0,0)};
\draw plot [smooth,tension=1.8] coordinates {(0,0) (1.5,-1.5) (0,-3)};
\draw plot [smooth,tension=1.8] coordinates {(6,3.1) (1.5,4.5) (6,5.9)};
\draw plot [smooth,tension=1.8] coordinates {(6,3.1) (6.7,3.8) (6,4.5)};
\draw plot [smooth,tension=1.8] coordinates {(6,4.5) (6.7,5.2) (6,5.9)};

\draw plot [smooth cycle,tension=1.2] coordinates {(1.5,1.5) (4.1,2.5) (6.7,1.5) (4.1,0.5)};

\draw (-4.5,1.5) arc (-126.87:-90+36.87:2.5);
\draw plot [smooth,tension=1.2] coordinates {(-4,-0.29+1.5) (-3,1.5) (-2,-0.29+1.5)};

\begin{scope}[shift={(6,4.5)},scale=0.7]
\draw (-4.5,0) arc (-126.87:-90+36.87:2.5);
\draw plot [smooth,tension=1.2] coordinates {(-4,-0.29) (-3,0) (-2,-0.29)};
\end{scope}
\begin{scope}[shift={(5.2,1.5)},scale=0.7]
\draw (-4.5,0) arc (-126.87:-90+36.87:2.5);
\draw plot [smooth,tension=1.2] coordinates {(-4,-0.29) (-3,0) (-2,-0.29)};
\end{scope}

\coordinate [label=center: $\scriptstyle{\lambda_{g_1}}$] (B) at (-4,-3);
\coordinate [label=center: $\scriptstyle{\lambda_{g_2}}$] (B) at (4,6.3);
\coordinate [label=center: $\scriptstyle{\lambda_{g_3}}$] (B) at (4,-0.1);
\coordinate [label=left: $\scriptstyle{a_1+a_2+a_3}$] (B) at (6.7,5.2);
\coordinate [label=right: $\scriptstyle{e_{\alpha_4}}$] (B) at (6.7,5.2);
\coordinate [label=left: $\scriptstyle{a_4}$] (B) at (6.7,3.8);
\coordinate [label=right: $\scriptstyle{\psi^{d_2-1}e_{\alpha_2}}$] (B) at (6.7,3.8);
\coordinate [label=left: $\scriptstyle{a_5+a_6}$] (B) at (6.7,1.5);
\coordinate [label=right: $\scriptstyle{\psi^{d_3-2}e_{\alpha_3}}$] (B) at (6.7,1.5);
\coordinate [label=left: $\scriptstyle{a_7+a_8}$] (B) at (1.5,-1.5);
\coordinate [label=right: $\scriptstyle{\psi^{d_1-2}e_{\alpha_1}}$] (B) at (1.5,-1.5);

\coordinate [label=left: $\scriptstyle{e_1}$] (B) at (-7,1.5);

\end{scope}

\begin{scope}[shift={(4.5,0)}]
\draw [thick] (4.5,0) -- (6.5,0);
\draw [thick] (4.5,0.2) -- (4.5,-0.2);
\draw [thick] (6.5,0) -- (6.2,0.2);
\draw [thick] (6.5,0) -- (6.2,-0.2);
\end{scope}

\coordinate [label=left: $\scriptstyle{e_1}$] (B) at (-7,0);

\end{tikzpicture}
\caption{Map $\phi$ and integrals over double ramification cycles}
\label{fig:transformation of the integral}
\end{figure}
We have
\begin{align}
&\int_{\DR_\Gamma(-\sum b_i,b_1,\ldots,b_{2g+n-1})}\lambda_g c_{g,2g+2n}(e_1^{2g+n}\otimes\otimes_{j=1}^n e_{\alpha_j})\prod_{j=1}^n\psi_{2g+n-1+j}^{d_j}=\label{big integral}\\
=&\prod_{h\in H^e_+(\Gamma)}b(h)\sum_{\nu\colon H^e(\Gamma)\to\{1,\ldots,N\}}\eta^{\nu(h)\nu(\iota(h))}\times\notag\\
&\times\prod_{v\in V(\Gamma)}\int_{\DR_{g(v)}\left(0,(b(l))_{l\in L_1[v]},(b(h))_{h\in H^e[v]}\right)}\lambda_{g(v)}\psi_0^{d_{p(v)}}c_{g,|H[v]|}(e_{\alpha_{p(v)}}\otimes e_1^{|L_1[v]|}\otimes\otimes_{h\in H^e[v]}e_{\nu(h)}).\notag
\end{align}
In order to simplify our formulas a little bit, it is convenient to use the notation~\eqref{eq:H'} and set $\nu(l_0(\Gamma)):=1$. Then we can rewrite formula~\eqref{big integral} in the following way:
\begin{align}
&\int_{\DR_\Gamma(-\sum b_i,b_1,\ldots,b_{2g+n-1})}\lambda_g c_{g,2g+2n}(e_1^{2g+n}\otimes\otimes_{j=1}^n e_{\alpha_j})\prod_{j=1}^n\psi_{2g+n-1+j}^{d_j}=\label{big integral2}\\
=&\prod_{h\in H^e_+(\Gamma)}b(h)\sum_{\nu\colon H^e(\Gamma)\to\{1,\ldots,N\}}\eta^{\nu(h)\nu(\iota(h))}\times\notag\\
&\times\prod_{v\in V(\Gamma)}\int_{\DR_{g(v)}\left(0,(b(l))_{l\in L'_1[v]},(b(h))_{h\in H'[v]}\right)}\lambda_{g(v)}\psi_0^{d_{p(v)}}c_{g,|H[v]|}(e_{\alpha_{p(v)}}\otimes e_1^{|L'_1[v]|}\otimes\otimes_{h\in H'[v]}e_{\nu(h)}).\notag
\end{align}
Suppose that $v\in V^\reg(\Gamma)$. Then equation~\eqref{eq:divisibility} implies that the integral
\begin{gather}\label{part of the big integral}
\int_{\DR_{g(v)}\left(0,(b(l))_{l\in L'_1[v]},(b(h))_{h\in H'[v]}\right)}\lambda_{g(v)}\psi_0^{d_{p(v)}}c_{g,|H[v]|}(e_{\alpha_{p(v)}}\otimes e_1^{|L'_1[v]|}\otimes\otimes_{h\in H'[v]}e_{\nu(h)})
\end{gather}
is equal to 
\begin{gather*}
\int_{\DR_{g(v)}\left(\sum_{l\in L_1'[v]}b(l),(b(h))_{h\in H'[v]}\right)}\lambda_{g(v)}\psi_0^{d_{p(v)}-|L_1'[v]|}c_{g,|H[v]|-|L_1'[v]|}(e_{\alpha_{p(v)}}\otimes\otimes_{h\in H'[v]}e_{\nu(h)})+\sum_{i=1}^{2g+n-1}O(b_i^2)
\end{gather*}
in the case $d_{p(v)}\ge |L_1'[v]|$ and is equal to $\sum_{i=1}^{2g+n-1}O(b_i^2)$, if $d_{p(v)}<|L_1'[v]|$. Suppose that $v\in V^\exc(\Gamma)$. Then the set $H'[v]$ consists of only one element, $H'[v]=\{l\}$. The integral~\eqref{part of the big integral} is equal to $\eta_{\alpha_{p(v)}\nu(l)}$, if $|L'_1[v]|=d_{p(v)}+1$, and is equal to zero otherwise. 

We say that an admissible modified stable tree $\Gamma\in\AMST^n_{g,2g+n}$ is compatible with an $n$-tuple of non-negative integers $(d_1,\ldots,d_n)$ if the following two conditions are satisfied:
\begin{enumerate}
\item[a)] For any $v\in V^\reg(\Gamma)$ we have $d_{p(v)}\ge |L_1'[v]|$.
\item[b)] For any $v\in V^\exc(\Gamma)$ we have $d_{p(v)}+1=|L_1'[v]|$.
\end{enumerate}
We obtain that the coefficient of $b_1b_2\ldots b_{2g+n-1}$ in~\eqref{big integral} can be non-zero only if $\Gamma$ is compatible with $(d_1,\ldots,d_n)$. Suppose that an admissible modified stable tree $\Gamma\in\AMST^n_{g,2g+n}$ is compatible with an $n$-tuple $(d_1,\ldots,d_n)$, where $\sum d_i=d+1$. Then from the computations above it follows that the coefficient of $b_1b_2\ldots b_{2g+n-1}$ in~\eqref{big integral} is equal to the coefficient of~$b_1b_2\ldots b_{2g+n-1}$ in
\begin{align*}
&\left(\prod_{v\in V^\exc(\Gamma)}\sum_{l\in L_1'[v]}b(l)\right)\times\\
&\times\int_{\DR_{\phi(\Gamma)}\left(-\sum b_i,\left(\sum_{l\in L_1'[v(l_{i+2g+n-1}(\Gamma))]}b(l)\right)_{1\le i\le n}\right)}\lambda_g c_{g,n+1}(e_1\otimes\otimes_{i=1}^n e_{\alpha_i})\prod_{v\in V^\reg(\Gamma)}\psi_{p(v)}^{d_{p(v)}-|L_1'[v]|}.
\end{align*}
An example of an integral from this formula is illustrated on the right-hand side of Fig.~\ref{fig:transformation of the integral}. It is easy to see that the coefficient of $b_1b_2\ldots b_{2g+n-1}$ in the last expression is equal to the coefficient of $a_1^{d_1}\ldots a_n^{d_n}$ in
$$
\left(\prod_{v\in V(\Gamma)}|L_1'[v]|!\right)\int_{\DR_{\phi(\Gamma)}(-\sum a_i,a_1,\ldots,a_n)}\lambda_g c_{g,n+1}(e_1\otimes\otimes_{i=1}^n e_{\alpha_i})\prod_{v\in V^{\reg}(\Gamma)}(a_{p(v)}\psi_{p(v)})^{d_{p(v)}-|L_1'[v]|}.
$$
Let $e:=|V^{\exc}(\Gamma)|$. Note that 
$$
\sum_{p\in V^\reg(\Gamma)}(d_{p(v)}-|L_1'[v]|)=d+1-(2g+n-1-e).
$$
Note also that for any $v\in V^\reg(\Gamma)$ the leg $l_{p(v)}(\phi(\Gamma))\in L(\phi(\Gamma))$ satisfies the property:
$$
p(v)=\min_{l'\in L'[v(l_{p(v)}(\phi(\Gamma)))]}i(l').
$$
This motivates the following definition. For $0\le e\le n-1$ and an admissible stable tree $\Gamma\in\AST^{n-e}_{g,n+1}$ define a set $S_{\Gamma,d}\subset\mbZ^n_{\ge 0}$ by
$$
S_{\Gamma,d}:=\left\{(c_1,\ldots,c_n)\in\mbZ_{\ge 0}^n\left|
\begin{smallmatrix}
\text{$c_i=0$, if $i\notin\{\min_{l\in L'[v]}i(l)\}_{v\in V(\Gamma)}$}\\
\sum c_i=d+1-(2g+n-1-e)
\end{smallmatrix}\right.\right\}.
$$
We obtain the following equation:
\begin{align*}
&\sum_{\substack{d_1,\ldots,d_n\ge 0\\\sum d_i=d+1}}\<\tau_0(e_1)\tau_{d_1}(e_{\alpha_1})\ldots\tau_{d_n}(e_{\alpha_n})\>^\DR_g\prod_{i=1}^n a_i^{d_i}=\\
=&\sum_{e=0}^{n-1}\sum_{\Gamma\in\AST^{n-e}_{g,n+1}}\sum_{(c_1,\ldots,c_n)\in S_{\Gamma,d}}\int_{\DR_\Gamma(-\sum a_i,a_1,\ldots,a_n)}\lambda_g c_{g,n+1}(e_1\otimes\otimes_{i=1}^n e_{\alpha_i})\prod_{i=1}^n(a_i\psi_i)^{c_i}.\notag
\end{align*}
We can now see that the following relation in the cohomology of~$\oM_{g,n+1}$ implies formula~\eqref{eq:main formula with insertion}:
\begin{multline}\label{eq:relation in cohomology}
\sum_{e=0}^{n-1}\sum_{\Gamma\in\AST^{n-e}_{g,n+1}}\sum_{(c_1,\ldots,c_n)\in S_{\Gamma,d}}\lambda_g\DR_\Gamma(a_0,a_1,\ldots,a_n)\prod_{i=1}^n(a_i\psi_i)^{c_i}=\\
=\sum_{\Gamma\in\ST^{d-2g+2}_{g,n+1}}C(\Gamma)\lambda_g\DR_\Gamma(a_0,a_1,\ldots,a_n),
\end{multline}
where $a_0:=-\sum_{i=1}^n a_i$. This relation will be proved in the next section.

%%%%%%%%%%%%%%%%%%%%%%%%%%%%%%%%%%%%%%%%%%%%%%%%%%%%%%%%%%

\subsubsection{Relation in the cohomology of $\oM_{g,n}$}\label{subsubsection:relation}

We prove relation~\eqref{eq:relation in cohomology} by the double induction on~$d$ and on~$n$. The base cases are when $d=2g-1$ or $n=1$. If $d=2g-1$, then the condition $\sum c_i=d+1-(2g+n-1-e)$ in the definition of the set $S_{\Gamma,d}$ immediately implies that $e=n-1$ and that the left-hand side of~\eqref{eq:relation in cohomology} is equal to $\DR_g(a_0,\ldots,a_n)$. The right-hand side of~\eqref{eq:relation in cohomology} is clearly the same. Suppose $n=1$. Then the left-hand side of~\eqref{eq:relation in cohomology} is equal to 
\begin{gather}\label{eq:base case,left}
\lambda_g\DR_g(-a_1,a_1)(a_1\psi_1)^{d+1-2g},
\end{gather}
while the right-hand side of~\eqref{eq:relation in cohomology} is equal to
\begin{gather}\label{eq:base case,right}
\sum_{\Gamma\in\ST^{d-2g+2}_{g,2}}C(\Gamma)\lambda_g\DR_\Gamma(-a_1,a_1).
\end{gather}
The class $\DR_\Gamma(-a_1,a_1)$ is zero unless $\Gamma$ is a chain. Therefore, applying formula~\eqref{eq:DR times psi} to~\eqref{eq:base case,left} $d+1-2g$ times, we get~\eqref{eq:base case,right}. So, the base cases for our induction are proved.

Suppose now that $d\ge 2g$ and $n\ge 2$. We rewrite the left-hand side of~\eqref{eq:relation in cohomology} in the following way:
\begin{align}
\sum_{e=0}^{n-1}\sum_{\Gamma\in\AST^{n-e}_{g,n+1}}\sum_{(c_1,\ldots,c_n)\in S_{\Gamma,d}}&\lambda_g\DR_\Gamma(a_0,\ldots,a_n)\prod_{i=1}^n(a_i\psi_i)^{c_i}=\notag\\
=&a_1\psi_1\sum_{e=0}^{n-1}\sum_{\Gamma\in\AST^{n-e}_{g,n+1}}\sum_{(c_1,\ldots,c_n)\in S_{\Gamma,d-1}}\lambda_g\DR_\Gamma(a_0,\ldots,a_n)\prod_{i=1}^n(a_i\psi_i)^{c_i}+\label{part 1 of induction}\\
&+\sum_{e=0}^{n-1}\sum_{\Gamma\in\AST^{n-e}_{g,n+1}}\sum_{(0,c_2,\ldots,c_n)\in S_{\Gamma,d}}\lambda_g\DR_\Gamma(a_0,\ldots,a_n)\prod_{i=2}^n(a_i\psi_i)^{c_i}.\label{part 2 of induction}
\end{align}
By the induction assumption, expression~\eqref{part 1 of induction} is equal to
\begin{align}
&a_1\psi_1\sum_{\Gamma\in\ST^{d-2g+1}_{g,n+1}}C(\Gamma)\lambda_g\DR_\Gamma(a_0,\ldots,a_n)\stackrel{\text{by~\eqref{eq:DR times psi}}}{=}\notag\\
=&\sum_{\Gamma\in\ST^{d-2g+1}_{g,n+1}}C(\Gamma)\lambda_g\sum_{\substack{g_1,g_2\ge 0\\g_1+g_2=g(v(l_1(\Gamma)))}}\sum_{\substack{I\sqcup J=H_+'[v(l_1(\Gamma))]\\l_1(\Gamma)\in J\\2g_1+|I|>0\\2g_2+|J|-1}}\frac{2g_1+|I|}{r(v(l_1(\Gamma)))}\DR_{\Spl(\Gamma,v(l_1(\Gamma)),g_1,I)}(a_0,\ldots,a_n)\label{final part 1,first line}\\
&-\sum_{\Gamma\in\ST^{d-2g+1}_{g,n+1}}C(\Gamma)\lambda_g\sum_{\substack{g_1,g_2\ge 0\\g_1+g_2=g(v(l_1(\Gamma)))}}\sum_{\substack{I\sqcup J=H_+'[v(l_1(\Gamma))]\\l_1(\Gamma)\in I\\2g_2+|J|-1>0}}\frac{2g_2+|J|-1}{r(v(l_1(\Gamma)))}\DR_{\Spl(\Gamma,v(l_1(\Gamma)),g_1,I)}(a_0,\ldots,a_n)\label{final part 1,second line}.
\end{align}

Let us now analyze expression~\eqref{part 2 of induction}. From the definition of an admissible stable tree it immediately follows that for any $\Gamma\in\AST^{n-e}_{g,n+1}$ the leg~$l_1(\Gamma)$ is incident to the root of~$\Gamma$. The stable rooted tree $\Gamma$ is obtained by attaching the stable rooted trees $\Gamma_h$, $h\in H^e[v(l_0(\Gamma))]$ together with the legs from $L[v(l_0(\Gamma))]$ to the vertex $v(l_0(\Gamma))$. Note that the number of legs in each tree $\Gamma_h$ is strictly less than $n+1$. Therefore, the induction assumption implies that expression~\eqref{part 2 of induction} is equal to
\begin{gather}\label{final part 2}
\sum_{\substack{\Gamma\in\ST^{d-2g+2}_{g,n+1}\\v(l_1(\Gamma))=v(l_0(\Gamma))}}\tC(\Gamma)\lambda_g\DR_\Gamma(a_0,\ldots,a_n),
\end{gather}
where
$$
\tC(\Gamma):=\prod_{v\in V(\Gamma)\backslash\{v(l_0(\Gamma))\}}\frac{r(v)}{\sum_{\tv\in\Desc[v]}r(\tv)}.
$$
It remains to prove that the sum of~\eqref{final part 1,first line},~\eqref{final part 1,second line} and~\eqref{final part 2} is equal to the right-hand side of~\eqref{eq:relation in cohomology}. We see that all expressions~\eqref{final part 1,first line},~\eqref{final part 1,second line},~\eqref{final part 2} and the right-hand side of~\eqref{eq:relation in cohomology} are sums of classes 
\begin{gather}\label{main classes}
\lambda_g\DR_\Gamma(a_0,\ldots,a_n),\quad \Gamma\in\ST^{d-2g+2}_{g,n+1},
\end{gather}
with some rational coefficients. Consider a stable tree $\Gamma\in\ST^{d-2g+2}_{g,n+1}$. It remains to check that the coefficients of the class~\eqref{main classes} in the sum of~\eqref{final part 1,first line},~\eqref{final part 1,second line}~\eqref{final part 2} and in the right-hand side of~\eqref{eq:relation in cohomology} are equal. Let $v:=v(l_1(\Gamma))$. Introduce the notations
\begin{align*}
R:=&\sum_{v'\in\Desc[v]}r(v'),\\
R_h:=&\sum_{v'\in\Desc[v(\iota(h))]}r(v'),\quad\text{for $h\in H^e_+[v]$}.
\end{align*}
There are two cases.

{\it Case 1}. Suppose $v\ne v(l_0(\Gamma))$. Clearly, the set $H^e[v]\backslash H^e_+[v]$ consists of a unique element. Let us denote it by $h_-$ and let $\tv:=v(\iota(h_-))$ (see Fig~\ref{fig:example of a tree}).
\begin{figure}[t]
\begin{tikzpicture}[scale=0.4]
\draw[dashed] (-2,0) -- (-4,0);
\draw (0,0) -- (-2,0);
\draw (0,0) -- (3,3);
\draw (0,0) -- (2,0);
\draw[dashed] (2,0) -- (4,0);
\draw (0,0) -- (2,-1.5);
\draw[dashed] (2,-1.5) -- (4,-3);
\draw (3,3) -- (6,8);
\draw (3,3) -- (6,5);
\draw (3,3) -- (6,3);
\draw (3,3) -- (6,1.5);
\draw (3,3) -- (6,0);
\draw (6,8) -- (8.5,7.5);
\draw[dashed] (8.5,7.5) -- (11,7);
\draw (6,8) -- (8.5,8);
\draw[dashed] (8.5,8) -- (11,8);
\draw (6,8) -- (8.5,8.5);
\draw[dashed] (8.5,8.5) -- (11,9);
\draw (6,5) -- (8.5,5.25);
\draw[dashed] (8.5,5.25) -- (11,5.5);
\draw (6,5) -- (8.5,4.75);
\draw[dashed] (8.5,4.75) -- (11,4.5);

\fill[white] (0,0) circle(7mm);
\draw (0,0) circle(7mm) node {$\widetilde{v}$};
\fill[white] (3,3) circle(7mm);
\draw (3,3) circle(7mm) node {$v$};
\fill[white] (6,8) circle(7mm);
\draw (6,8) circle(7mm);
\fill[white] (6,5) circle(7mm);
\draw (6,5) circle(7mm);

\coordinate [label=-45: $1$] () at (6,0);
\coordinate [label=center: $\scriptstyle{h_-}$] () at (1.8,2.5);

\end{tikzpicture}
\caption{}
\label{fig:example of a tree}
\end{figure}
Let
\begin{align*}
\tR:=&\sum_{v'\in\Desc[\tv]}r(v'),\\
B:=&\prod_{v'\in V(\Gamma)\backslash(\{v,\tv\}\cup\cup_{h\in H^e_+[v]}v(\iota(h)))}\frac{r(v')}{\sum_{v''\in\Desc[v']}r(v'')}.
\end{align*}
So the constant $C(\Gamma)$ can be written as
$$
C(\Gamma)=B\cdot\frac{r(\tv)}{\tR}\frac{r(v)}{R}\prod_{h\in H^e_+[v]}\frac{r(v(\iota(h)))}{R_h}.
$$
Clearly, the stable tree $\Gamma$ can be obtained by splitting of the tree $\Con(\Gamma,v,h_-)$. Therefore, the coefficient of the class~\eqref{main classes} in~\eqref{final part 1,first line} is equal to
\begin{gather}\label{eq:case 1,first contribution}
\frac{r(\tv)}{r(\tv)+r(v)}\cdot B\cdot\frac{r(\tv)+r(v)}{\tR}\prod_{h\in H^e_+[v]}\frac{r(v(\iota(h)))}{R_h}=B\cdot\frac{r(\tv)}{\tR}\prod_{h\in H^e_+[v]}\frac{r(v(\iota(h)))}{R_h}.
\end{gather}
On the other hand, for any $h\in H^e_+[v]$ the stable tree $\Gamma$ can be obtained by splitting of the tree $\Con(\Gamma,v,h)$. Therefore, the coefficient of the class~\eqref{main classes} in~\eqref{final part 1,second line} is equal to
\begin{align}
&-\sum_{h\in H^e_+[v]}\frac{r(v(\iota(h)))}{r(v(\iota(h)))+r(v)}\cdot B\cdot\frac{r(\tv)}{\tR}\frac{r(v)+r(v(\iota(h)))}{R}\prod_{h'\in H^e_+[v]\backslash\{h\}}\frac{r(v(\iota(h')))}{R_{h'}}=\notag\\
=&-B\cdot\sum_{h\in H^e_+[v]}\frac{r(\tv)}{\tR}\frac{r(v(\iota(h)))}{R}\prod_{h'\in H^e_+[v]\backslash\{h\}}\frac{r(v(\iota(h')))}{R_{h'}}.\label{eq:case 1,second contribution}
\end{align}
Obviously, the class~\eqref{main classes} does not appear in~\eqref{final part 2}. Summing~\eqref{eq:case 1,first contribution} and~\eqref{eq:case 1,second contribution}, we get
$$
B\cdot\frac{r(\tv)}{\tR}\left(\prod_{h\in H^e_+[v]}\frac{r(v(\iota(h)))}{R_h}\right)\left(1-\sum_{h\in H^e_+[v]}\frac{R_h}{R}\right)=C(\Gamma).
$$
So, this case is done.

{\it Case 2}. Suppose $v=v(l_0(\Gamma))$. Let
\begin{gather*}
B:=\prod_{v'\in V(\Gamma)\backslash(\{v\}\cup\cup_{h\in H^e_+[v]}v(\iota(h)))}\frac{r(v')}{\sum_{v''\in\Desc[v']}r(v'')}.
\end{gather*}
Therefore,
$$
C(\Gamma)=B\cdot\frac{r(v)}{R}\prod_{h\in H^e_+[v]}\frac{r(v(\iota(h)))}{R_h}.
$$
It is easy to see that the class~\eqref{main classes} does not appear in~\eqref{final part 1,first line}. By the same arguments, as in the first case, the class~\eqref{main classes} appears in~\eqref{final part 1,second line} with the coefficient
\begin{align}
&-\sum_{h\in H^e_+[v]}\frac{r(v(\iota(h)))}{r(v(\iota(h)))+r(v)}\cdot B\cdot\frac{r(v(\iota(h)))+r(v)}{R}\prod_{h'\in H^e_+[v]\backslash\{h\}}\frac{r(v(\iota(h')))}{R_{h'}}=\notag\\
=&-B\cdot\sum_{h\in H^e_+[v]}\frac{r(v(\iota(h)))}{R}\prod_{h'\in H^e_+[v]\backslash\{h\}}\frac{r(v(\iota(h')))}{R_{h'}}.\label{eq:case 2,first contribution}
\end{align}
One can easily see that the coefficient of the class~\eqref{main classes} in~\eqref{final part 2} is equal to
\begin{gather}\label{eq:case 2,second contribution}
\tC(\Gamma)=B\cdot\prod_{h\in H^e[v]}\frac{r(v(\iota(h)))}{R_{h}}.
\end{gather}
Summing~\eqref{eq:case 2,first contribution} and~\eqref{eq:case 2,second contribution}, we obtain
$$
B\cdot\left(\prod_{h\in H^e_+[v]}\frac{r(v(\iota(h)))}{R_h}\right)\left(1-\sum_{h\in H^e_+[v]}\frac{R_h}{R}\right)=C(\Gamma).
$$
Case 2 is also done. Relation~\eqref{eq:relation in cohomology} is proved and, hence, Theorem~\ref{theorem:main geometric formula} is also proved.

%%%%%%%%%%%%%%%%%%%%%%%%%%%%%%%%%%%%%%%%%%%%%%%%%%%%%%%%%%%%%%%%%%%%%%%
%%%%%%%%%%%%%%%%%%%%%%%%%%%%%%%%%%%%%%%%%%%%%%%%%%%%%%%%%%%%%%%%%%%%%%%

\section{Miura transformation for the Dubrovin-Zhang operator}\label{section:Miura for DZ}

In this section we show that our strong DR/DZ equivalence conjecture \cite[Section~7.3]{BDGR16} together with formula~\eqref{eq:geometric formula for k=0} give a simple description of a Miura transformation that should reduce the Hamiltonian operator of the Dubrovin-Zhang hierarchy to the standard form. Remarkably, the description is given purely in terms of the potential of the cohomological field theory. The main goal of this section is to prove that this Miura transformation indeed reduces the Dubrovin-Zhang operator to the standard form. This gives a new evidence for the strong DR/DZ equivalence conjecture. 

In Sections~\ref{subsection:recall of DZ theory} and~\ref{subsection:strong DR/DZ equivalence} we briefly recall the theory of the Dubrovin-Zhang hierarchies and our strong DR/DZ equivalence conjecture from~\cite{BDGR16}. The Miura transformation for the Dubrovin-Zhang operator is given at the end of Section~\ref{subsection:strong DR/DZ equivalence}. The main result is proved in Section~\ref{subsection:proof of Miura for DZ}.

Throughout this section we fix a semisimple cohomological field theory $c_{g,n}\colon V^{\otimes n}\to H^\even(\oM_{g,n},\mbC)$ with $\dim V=N$. 

\subsection{Brief recall of the Dubrovin-Zhang theory}\label{subsection:recall of DZ theory}

Here we recall the construction of the Dubrovin-Zhang hierarchy. We follow the approach from~\cite{BPS12b} (see also~\cite{BPS12a}).

The potential of the cohomological field theory is defined by
\begin{align*}
F(t^*_*,\eps):=&\sum_{g\ge 0}F_g(t^*_*)\eps^{2g},\\
F_g(t^*_*):=&\sum_{\substack{n\ge 0\\2g-2+n>0}}\sum_{d_1,\ldots,d_n\ge 0}\<\tau_{d_1}(e_{\alpha_1})\ldots\tau_{d_n}(e_{\alpha_n})\>_g \frac{t^{\alpha_1}_{d_1}\ldots t^{\alpha_n}_{d_n}}{n!},
\end{align*}
where
$$
\<\tau_{d_1}(e_{\alpha_1})\ldots\tau_{d_n}(e_{\alpha_n})\>_g:=\int_{\oM_{g,n}}c_{g,n}(\otimes_{i=1}^n e_{\alpha_1})\psi_1^{d_1}\ldots\psi_n^{d_n}.
$$
Recall the string and the dilaton equations for $F$:
\begin{align*}
&\frac{\d F}{\d t^1_0}=\sum_{n\ge 0}t^\alpha_{n+1}\frac{\d F}{\d t^\alpha_n}+\frac{1}{2}\eta_{\alpha\beta}t^\alpha_0t^\beta_0+\eps^2\<\tau_0(e_1)\>_1,\\
&\frac{\d F}{\d t^1_1}=\sum_{n\ge 0}t^\alpha_n\frac{\d F}{\d t^\alpha_n}+\eps\frac{\d F}{\d\eps}-2F+\eps^2\frac{N}{24}.
\end{align*}

We will use rings of differential polynomials in different variables and Miura transformations between them. We refer the reader to~\cite[Section~3.4]{BDGR16} for the corresponding notations. We also refer the reader to~\cite[Section 3.1]{BDGR16} for a brief review of the theory of tau-symmetric Hamiltonian hierarchies.

Introduce power series $(w^\top)^\alpha\in\mbC[[x,t^*_*,\eps]]$ by
$$
(w^\top)^\alpha:=\left.\eta^{\alpha\mu}\frac{\d^2 F}{\d t^\mu_0\d t^1_0}\right|_{t^1_0\mapsto t^1_0+x}.
$$
Let $(w^\top)^\alpha_n:=\d_x^n(w^\top)^\alpha$. For $k\ge 0$ denote by $\mbC[[t^*_*]]^{(k)}$ the vector subspace of $\mbC[[t^*_*]]$ spanned by monomials $t^{\alpha_1}_{d_1}\ldots t^{\alpha_n}_{d_n}$ with $\sum d_i\ge k$. From the string equation for~$F$ it follows that
\begin{gather}\label{eq:property of wtop1}
\left.(w^\top)^\alpha_n\right|_{\eps=x=0}-t^\alpha_n-\delta_{n,1}\delta^{\alpha,1}\in\mbC[[t^*_*]]^{(n+1)}.
\end{gather}
Therefore, any power series in $t^\alpha_n$ and $\eps$ can be expressed as a power series in $\left(\left.(w^\top)^\alpha_n\right|_{x=0}-\delta_{n,1}\delta^{\alpha,1}\right)$ and $\eps$ in a unique way. Consider formal variables $w^1,\ldots,w^N$. In~\cite{BPS12b} the authors proved that for any $1\le\alpha,\beta\le N$ and $p,q\ge 0$ there exists a unique differential polynomial $\Omega^{\DZ}_{\alpha,p;\beta,q}\in\hcA^{[0]}_{w^1,\ldots,w^N}$ such that
$$
\Omega^\DZ_{\alpha,p;\beta,q}(w^\top,w^\top_x,\ldots;\eps)=\left.\frac{\d^2 F}{\d t^\alpha_p\d t^\beta_q}\right|_{t^1_0\mapsto t^1_0+x}.
$$
In particular, $\Omega^\DZ_{\alpha,0;1,0}=\eta_{\alpha\mu}w^\mu$. The equations of the Dubrovin-Zhang hierarchy are given by
\begin{gather}\label{eq:DZ system}
\frac{\d w^\alpha}{\d t^\beta_q}=\eta^{\alpha\mu}\d_x\Omega^\DZ_{\mu,0;\beta,q},\quad 1\le\alpha,\beta\le N,\quad q\ge 0.
\end{gather}
Clearly, the series~$(w^\top)^\alpha$ is a solution of these equations. It is called the topological solution.

The system~\eqref{eq:DZ system} has a Hamiltonian structure. The Hamiltonians are given by
\begin{gather}\label{eq:DZ Hamiltonians}
\oh^\DZ_{\alpha,p}=\int\Omega^\DZ_{\alpha,p+1;1,0}dx,\quad p\ge 0.
\end{gather}
The construction of the Hamiltonian operator is more complicated. Let 
$$
(v^\top)^\alpha:=\left.(w^\top)^\alpha\right|_{\eps=0}.
$$
Then any power series in $t^\alpha_n$ and $\eps$ can be expressed as a power series in $\left(\left.(v^\top)^\alpha_n\right|_{x=0}-\delta_{n,1}\delta^{\alpha,1}\right)$ and $\eps$ in a unique way. In particular, for $g\ge 1$ we can express the function $F_g$ as a function of $\left.(v^\top)^\alpha_n\right|_{x=0}$. Then $F_g$ depends only on $\left.(v^\top)^\alpha_n\right|_{x=0}$ with $n\le 3g-2$ (see e.g.~\cite[Proposition 4]{BPS12b}). This property is called the $3g-2$ property. Consider formal variables $v^1,\ldots,v^N$. Let~$\cA^\wk_{v^1,\ldots,v^N}$ be the ring of formal power series in $(v^\alpha_n-\delta^{\alpha,1}\delta_{n,1})$ with complex coefficients. We have a natural inclusion
$$
\cA_{v^1,\ldots,v^N}\subset\cA^\wk_{v^1,\ldots,v^N}.
$$
Let $\hcA^\wk_{v^1,\ldots,v^N}:=\cA^\wk_{v^1,\ldots,v^N}\otimes\mbC[[\eps]]$. Clearly, there exists a unique element $w^\alpha(v^*_*;\eps)\in\hcA^\wk_{v^1,\ldots,v^N}$ such that
$$
w^\alpha(v^\top,v^\top_x,\ldots;\eps)=(w^\top)^\alpha.
$$
We have 
\begin{gather}\label{eq:w-v change of variables}
w^\alpha(v^*_*;\eps)=v^\alpha+\sum_{g\ge 1}\eps^{2g}f^\alpha_g(v^*_*).
\end{gather}
The $3g-2$ property implies that the function $f^\alpha_g(v^*_*)\in\cA^\wk_{v^1,\ldots,v^N}$ depends only on $v^\gamma_n$ with $n\le 3g$. Then formula~\eqref{eq:w-v change of variables} can be considered as a change of variables between $v^\gamma$ and $w^\gamma$. Define an operator $K^\DZ(v^*_*;\eps)=\left((K^\DZ)^{\alpha\beta}(v^*_*;\eps)\right)$ by
\begin{gather}\label{eq:definition of the DZ operator}
(K^\DZ)^{\alpha\beta}(v^*_*;\eps)=\sum_{p,q\ge 0}\frac{\d w^\alpha(v^*_*;\eps)}{\d v^\mu_p}\d_x^p\circ\eta^{\mu\nu}\d_x\circ(-\d_x)^q\circ\frac{\d w^\beta(v^*_*;\eps)}{\d v^\nu_q}.
\end{gather}
Since $f^\alpha_g(v^*_*)$ depends only on $v^\gamma_n$ with $n\le 3g$, the expression on the right-hand side of~\eqref{eq:definition of the DZ operator} is well-defined. We have
$$
(K^\DZ)^{\alpha\beta}(v^*_*;\eps)=\sum_{i\ge 0}(K^\DZ)^{\alpha\beta}_i(v^*_*;\eps)\d_x^i.
$$
Let $(K^\DZ)^{\alpha\beta}_i(w^*_*;\eps)$ be the function $(K^\DZ)^{\alpha\beta}_i(v^*_*;\eps)$ expressed in the variables $w^\gamma$ using the change of variables~\eqref{eq:w-v change of variables}. We have $(K^\DZ)^{\alpha\beta}_i(w^*_*;\eps)\in\hcA^\wk_{w^1,\ldots,w^N}$. In~\cite{BPS12b} the authors proved that we actually have
$$
(K^\DZ)^{\alpha\beta}_i(w^*_*;\eps)\in\hcA^{[-i+1]}_{w^1,\ldots,w^N}.
$$
The operator $K^\DZ=\sum_{i\ge 0}K^\DZ_i(w^*_*;\eps)\d_x^i$ is Hamiltonian. Together with the local functionals~\eqref{eq:DZ Hamiltonians} it defines the Hamiltonian structure for the Dubrovin-Zhang system~\eqref{eq:DZ system}.

Finally, the tau-structure for the Dubrovin-Zhang hierarchy is given by the differential polynomials
$$
h^\DZ_{\alpha,p}=\Omega^\DZ_{\alpha,p+1;1,0},\quad p\ge -1.
$$
Since $h^\DZ_{\alpha,-1}=\eta_{\alpha\mu}w^\mu$, we see that the coordinates $w^\alpha$ are normal. 

%%%%%%%%%%%%%%%%%%%%%%%%%%%%%%%%%%%%%%%%%%%%%%%%%%%%%%%%%%%%%%%%%%%%%%

\subsection{Strong DR/DZ equivalence conjecture}\label{subsection:strong DR/DZ equivalence}

In~\cite[Section 7.3]{BDGR16} we proved that there exists a unique differential polynomial $\cP\in\hcA^{[-2]}_{w^1,\ldots,w^N}$ such that the power series $F^\red\in\mbC[[t^*_*,\eps]]$, defined by
\begin{gather}\label{eq:definition of Fred}
F^\red:=F+\left.\cP(w^\top,w^\top_x,w^\top_{xx},\ldots;\eps)\right|_{x=0},
\end{gather}
satisfies the following vanishing property:
\begin{gather}\label{eq:property of Fred}
\<\tau_{d_1}(e_{\alpha_1})\ldots\tau_{d_n}(e_{\alpha_n})\>^{\red}_g=0,\quad\text{if}\quad \sum d_i\le 2g-2,
\end{gather}
where $\<\tau_{d_1}(e_{\alpha_1})\ldots\tau_{d_n}(e_{\alpha_n})\>^{\red}_g$ are the coefficients of the expansion of $F^{\red}$:
$$
F^\red(t^*_*,\eps):=\sum_{g,n\ge 0}\frac{\eps^{2g}}{n!}\sum_{d_1,\ldots,d_n\ge 0}\<\tau_{d_1}(e_{\alpha_1})\ldots\tau_{d_n}(e_{\alpha_n})\>^\red_g \frac{t^{\alpha_1}_{d_1}\ldots t^{\alpha_n}_{d_n}}{n!}.
$$
We called the power series $F^\red$ the reduced potential of the cohomological field theory. We proved that the reduced potential~$F^\red$ satisfies the string and the dilaton equations:
\begin{align*}
&\frac{\d F^\red}{\d t^1_0}=\sum_{n\ge 0}t^\alpha_{n+1}\frac{\d F^\red}{\d t^\alpha_n}+\frac{1}{2}\eta_{\alpha\beta}t^\alpha_0 t^\beta_0,\\
&\frac{\d F^\red}{\d t^1_1}=\eps\frac{\d F^\red}{\d\eps}+\sum_{n\ge 0}t^\alpha_n\frac{\d F^\red}{\d t^\alpha_n}-2F^{\red}+\eps^2\frac{N}{24}.
\end{align*}
Recall (see~\cite[Section 4]{BDGR16}) that the tau-structure for the double ramification hierarchy is given by the differential polynomials $h^\DR_{\alpha,p}=\frac{\delta\og_{\alpha,p+1}}{\delta u^1}$. The normal coordinates for this tau-structure are
\begin{gather}\label{eq:normal coordinates for DR}
\tu^\alpha(u^*_*;\eps)=\eta^{\alpha\mu}h^\DR_{\mu,-1}=\eta^{\alpha\mu}\frac{\delta\og_{\mu,0}}{\delta u^1}.
\end{gather}
In~\cite[Section 7.3]{BDGR16} we proposed the following conjecture.
\begin{conjecture}\label{conjecture:strong}
The normal Miura transformation defined by the differential polynomial~$\cP$ transforms the Dubrovin-Zhang hierarchy to the double ramification hierarchy written in the normal coordinates $\tu^\alpha$.
\end{conjecture}
\noindent
We called this conjecture the strong DR/DZ equivalence conjecture. In~\cite[Section~7.3]{BDGR16} we proved that the strong DR/DZ equivalence conjecture is true if and only if $F^\DR=F^\red$.

Note that formulas~\eqref{eq:normal coordinates for DR} and~\eqref{eq:geometric formula for k=0} together with the string equation for~$F^\DR$ imply that that the normal coordinates $\tu^\alpha(u^*_*;\eps)$ can be described using the double ramification correlators:
\begin{gather}\label{eq:u-tu in terms of FDR}
\tu^\alpha(u^*_*;\eps)=u^\alpha+\sum_{g,n\ge 1}\frac{\eps^{2g}}{n!}\sum_{d_1+\ldots+d_n=2g}\eta^{\alpha\mu}\<\tau_0(e_1)\tau_0(e_\mu)\prod\tau_{d_i}(e_{\alpha_i})\>^{\DR}_g\prod u^{\alpha_i}_{d_i}.
\end{gather}
If Conjecture~\ref{conjecture:strong} is true, then $\<\tau_0(e_1)\tau_0(e_\mu)\prod\tau_{d_i}(e_{\alpha_i})\>^{\DR}_g=\<\tau_0(e_1)\tau_0(e_\mu)\prod\tau_{d_i}(e_{\alpha_i})\>^{\red}_g$. Together with equation~\eqref{eq:u-tu in terms of FDR}, it motivates the following theorem.
%Suppose that Conjecture~\ref{conjecture:strong} is true, then $F^\DR=F^\red$. Using equation~\eqref{eq:u-tu in terms of FDR} we obtain the following theorem.

\begin{theorem}\label{theorem:Miura for DZ}
Define Miura transformations $w^\alpha\mapsto\tu^\alpha(w^*_*;\eps)$ and $u^\alpha\mapsto\tu^\alpha(u^*_*;\eps)$ by
\begin{align*}
\tu^\alpha(w^*_*;\eps)=&w^\alpha+\eta^{\alpha\mu}\d_x\{\cP,\oh^\DZ_{\mu,0}\}_{K^\DZ},\\
\tu^\alpha(u^*_*;\eps)=&u^\alpha+\sum_{g,n\ge 1}\frac{\eps^{2g}}{n!}\sum_{d_1+\ldots+d_n=2g}\eta^{\alpha\mu}\<\tau_0(e_1)\tau_0(e_\mu)\prod\tau_{d_i}(e_{\alpha_i})\>^{\red}_g\prod u^{\alpha_i}_{d_i}.
\end{align*}
Then the Miura transformation $w^\alpha\mapsto u^\alpha(w^*_*;\eps)$ transforms the operator $K^\DZ$ to $\eta\d_x$.
\end{theorem}
\noindent
We will prove this theorem in the next section.

%%%%%%%%%%%%%%%%%%%%%%%%%%%%%%%%%%%%%%%%%%%%%%%%%%%%%%%%%%%%%%%%%%%%%%

\subsection{Proof of Theorem~\ref{theorem:Miura for DZ}}\label{subsection:proof of Miura for DZ}

We split the proof into three steps. In Section~\ref{subsubsection:rational Miura} we introduce rational Miura transformations and discuss their properties. In Section~\ref{subsubsection:rationality of w} we prove that the change of variables $v^\alpha\mapsto w^\alpha(v^*_*;\eps)$ from Section~\ref{subsection:recall of DZ theory} is a rational Miura transformation. Finally, in Section~\ref{subsubsection:final step} we prove Theorem~\ref{theorem:Miura for DZ}.

\subsubsection{Rational Miura transformations}\label{subsubsection:rational Miura}

For $d\in\mbZ$ let~$\cA^{\rt,[d]}_{v^1,\ldots,v^N}$ be the vector space spanned by expressions of the form
\begin{gather}\label{eq:rational function}
\sum_{i\ge m}\frac{P_i(v^*_*)}{(v^1_x)^i},
\end{gather}
where $m\in\mbZ$, $P_i\in\cA^{[d+i]}_{v^1,\ldots,v^N}$ and $\frac{\d P_i}{\d v^1_x}=0$. Let $\cA^{\rt}_{v^1,\ldots,v^N}:=\bigoplus_{d\in\mbZ}\cA^{\rt,[d]}_{v^1,\ldots,v^N}$. Since 
$$
\frac{1}{(v^1_x)^i}=(1+(v^1_x-1))^{-i}=\sum_{k\ge 0}{-i\choose k}(v^1_x-1)^k,
$$
we have a natural inclusion 
$$
\cA^{\rt}_{v^1,\ldots,v^N}\subset\cA^\wk_{v^1,\ldots,v^N}.
$$
In the same way, as for differential polynomials, we introduce a grading by $\deg v^\alpha_i=i$. Then the subspace $\cA^{\rt,[d]}_{v^1,\ldots,v^N}\subset \cA^{\rt}_{v^1,\ldots,v^N}$ consists precisely of elements of degree $d$. For an element $f(v^*_*)=\sum_{i\ge m}\frac{P_i(v^*_*)}{(v^1_x)^i}\in\cA^\rt_{v^1,\ldots,v^N}$ define the polynomial part by
\begin{gather*}
f(v^*_*)^\pol:=\sum_{i=m}^0\frac{P_i(v^*_*)}{(v^1_x)^i}\in\cA_{v^1,\ldots,v^N}.
\end{gather*}
Define the extended space $\hcA^\rt_{v^1,\ldots,v^N}:=\cA^\rt_{v^1,\ldots,v^N}[[\eps]]$. Denote by 
$$
\hcA^{\rt,[d]}_{v^1,\ldots,v^N}\subset\hcA^\rt_{v^1,\ldots,v^N}
$$
the subspace of elements of degree $d$, where we, as usual, set $\deg\eps=-1$. 

A rational function~\eqref{eq:rational function} is called tame, if there exists a non-negative integer~$C$ such that $\frac{\d P_i}{\d v^\alpha_k}=0$ for $k>C$. The subspace of tame elements in~$\cA^\rt_{v^1,\ldots,v^N}$ will be denoted by 
$$
\cA^{\rt,\t}_{v^1,\ldots,v^N}\subset\cA^\rt_{v^1,\ldots,v^N}.
$$
An element $f(v^*_*;\eps)=\sum_{g\ge 0}\eps^g f_g(v^*_*)\in\hcA^{\rt}_{v^1,\ldots,v^N}$ will be called tame if all functions $f_g\in\cA^{\rt,\t}_{v^1,\ldots,v^N}$ are tame. The subspace of tame elements in~$\hcA^\rt_{v^1,\ldots,v^N}$ will be denoted by 
$$
\hcA^{\rt,\t}_{v^1,\ldots,v^N}\subset\hcA^\rt_{v^1,\ldots,v^N}.
$$
Consider changes of variables of the form
\begin{gather}
v^\alpha\mapsto w^\alpha(v^*_*;\eps)=v^\alpha+\eps f^\alpha(v^*_*;\eps),\quad \alpha=1,\ldots,N,\quad f^\alpha\in\hcA^{\rt,\t,[1]}_{v^1,\ldots,v^N}\label{eq:rational Miura transformation}.
\end{gather}
We will call them rational Miura transformations. These transformations form a group. Any tame rational function $f(v^*_*;\eps)\in\hcA^{\rt,\t}_{v^1,\ldots,v^N}$ can be rewritten as a tame rational function in the new variables $w^\alpha$. The resulting tame rational function will be denoted by $f(w^*_*;\eps)$. Clearly, the polynomial part~$w^\alpha(v^*_*;\eps)^\pol$ of a rational Miura transformation~\eqref{eq:rational Miura transformation} is a usual Miura transformation. 

Define a subspace $S_{v^1,\ldots,v^N}\subset\hcA^{\rt,\t}_{v^1,\ldots,v^N}$ by
$$
S_{v^1,\ldots,v^N}:=\left\{\left.f\in\hcA^{\rt,\t}_{v^1,\ldots,v^N}\right|f^\pol=0,\frac{\d f}{\d v^1}=0\right\}.
$$
It is easy to see that the subspace $S_{v^1,\ldots,v^N}$ is closed under multiplication and also under the derivations $\d_x$ and $\frac{\d}{\d v^\gamma_n}$.

\begin{lemma}\label{lemma:rational Miura}
Let $v^\alpha\mapsto w^\alpha(v^*_*;\eps)$ be a rational Miura transformation such that 
$$
(w^\alpha)^\pol(v^*_*;\eps)=v^\alpha\quad\text{and}\quad \frac{\d w^\alpha(v^*_*;\eps)}{\d v^1}=\delta^{\alpha,1}.
$$
Consider an operator $K=(K^{\alpha\beta})$ defined by 
\begin{gather}\label{eq:hamiltonian operator after singular transformation}
K^{\alpha\beta}:=\sum_{p,q\ge 0}\frac{\d w^\alpha(v^*_*;\eps)}{\d v^\mu_p}\d_x^p\circ\eta^{\mu\nu}\d_x\circ(-\d_x)^q\circ\frac{\d w^\beta(v^*_*;\eps)}{\d v^\nu_q}=\sum_{i\ge 0}K^{\alpha\beta}_i(v^*_*;\eps)\d_x^i.
\end{gather}
Suppose that $K^{\alpha\beta}_i(w^*_*;\eps)\in\hcA_{w^1,\ldots,w^N}$. Then $K^{\alpha\beta}=\eta^{\alpha\beta}\d_x$.
\end{lemma}
\begin{proof}
From formula~\eqref{eq:hamiltonian operator after singular transformation} one can easily see that
$$
K^{\alpha\beta}_i(v^*_*;\eps)-\delta_{i,1}\eta^{\alpha\beta}\in S_{v^1,\ldots,v^N}.
$$
Observe that if $f(v^*_*;\eps)\in S_{v^1,\ldots,v^N}$, then $f(w^*_*;\eps)\in S_{w^1,\ldots,w^N}$. Since $S_{w_1,\ldots,w^N}\cap \hcA_{w_1,\ldots,w^N}=0$, we get $K^{\alpha\beta}_i(w^*_*,\eps)-\delta_{i,1}\eta^{\alpha\beta}=0$. The lemma is proved.
\end{proof}

\begin{lemma}\label{lemma:composition of rational Miura transformations}
Consider three sets of variables $v^\alpha$, $u^\alpha$ and $w^\alpha$. Suppose that we have rational Miura transformations $v^\alpha\mapsto u^\alpha(v^*_*;\eps)$ and $u^\alpha\mapsto w^\alpha(u^*_*;\eps)$ such that
\begin{align*}
&\frac{\d u^\alpha(v^*_*;\eps)}{\d v^1}=\delta^{\alpha,1},&& \frac{\d u^\alpha(v^*_*;\eps)^\pol}{\d v^1_x}=0,\\
&\frac{\d w^\alpha(u^*_*;\eps)}{\d u^1}=\delta^{\alpha,1},&& \frac{\d w^\alpha(u^*_*;\eps)^\pol}{\d u^1_x}=0.
\end{align*}
Then the polynomial part of the composition of these rational Miura transformations is equal to the composition of their polynomial parts.
\end{lemma}
\begin{proof}
The proof is straightforward. One should just notice that the singularities of~$w^\alpha(u^*_*;\eps)$ and $u^\alpha(v^*_*;\eps)$ cannot give a non-trivial contribution in the polynomial part of the composition of these rational Miura transformations. 
\end{proof}

Let us formulate one more technical statement in this section.
\begin{lemma}\label{lemma:inverse}
Consider variables $u^\alpha$ and $w^\alpha$. Suppose we have a Miura transformation $u^\alpha\mapsto w^\alpha(u^*_*;\eps)$ such that $\frac{\d w^\alpha(u^*_*;\eps)}{\d u^1}=\delta^{\alpha,1}$ and $\frac{\d w^\alpha(u^*_*;\eps)}{\d u^1_x}=0$. Then the inverse Miura transformation $w^\alpha\mapsto u^\alpha(w^*_*;\eps)$ satisfies the same properties: $\frac{\d u^\alpha(w^*_*;\eps)}{\d w^1}=\delta^{\alpha,1}$ and $\frac{\d u^\alpha(w^*_*;\eps)}{\d w^1_x}=0$.
\end{lemma}
\begin{proof}
This is a direct computation based on the chain rule.
\end{proof}

\subsubsection{Rationality of the function $w^\alpha(v^*_*,\eps)$}\label{subsubsection:rationality of w}

Consider the function~$w^\alpha(v^*_*;\eps)$ from Section~\ref{subsection:recall of DZ theory}.
\begin{proposition}\label{proposition:rationality}
We have $w^\alpha(v^*_*;\eps)\in\hcA^{\rt,\t,[0]}_{v^1,\ldots,v^N}$ and, moreover, $\frac{\d w^\alpha(v^*_*;\eps)}{\d v^1}=\delta^{\alpha,1}$.
\end{proposition}
\begin{remark}
While this work was under preparation, we were informed that this proposition was independently proved by S.~Shadrin, D.~Lewanski and A.~Popolitov.
\end{remark}
\begin{proof}[Proof of Proposition~\ref{proposition:rationality}]
The proof is very similar to the construction of the differential polynomial $\cP$ from~\cite[Section 7.3]{BDGR16}. Consider the $\eps$-expansion of the topological solution~$(w^\top)^\alpha$:
$$
(w^\top)^\alpha(x,t^*_*,\eps)=\sum_{g\ge 0}\eps^{2g}(w^\top)^{\alpha,[g]}(x,t^*_*).
$$
Define a linear differential operator $O_\dil$ by
$$
O_\dil:=\frac{\d}{\d t^1_1}-x\frac{\d}{\d x}-\sum_{n\ge 0}t^\gamma_n\frac{\d}{\d t^\gamma_n}.
$$
For $g\ge 1$ let us construct a sequence of functions $w^{\alpha,[g,k]}\in\cA^{\rt,[2g]}_{v^1,\ldots,v^N}$, $k\ge -1$, such that
\begin{align}
&w^{\alpha,[g,k]}=w^{\alpha,[g,k-1]}+(v^1_x)^{2g-k}P^{\alpha,[g,k]},\quad k\ge 0,\quad P^{\alpha,[g,k]}\in\cA^{[k]}_{v^1,\ldots,v^N},\quad\frac{\d P^{\alpha,[g,k]}}{\d v^1_x}=0,\label{eq:construction of rational,eq1}\\
&\left.\left((w^\top)^{\alpha,[g]}-w^{\alpha,[g,k]}(v^\top,v^\top_x,\ldots)\right)\right|_{x=0}\in\mbC[[t^*_*]]^{(k+1)},\label{eq:construction of rational,eq2}\\
&O_{\dil}w^{\alpha,[g,k]}(v^\top,v^\top_x,\ldots)=2g\cdot w^{\alpha,[g,k]}(v^\top,v^\top_x,\ldots).\label{eq:construction of rational,eq3}
\end{align}
Let $w^{\alpha,[g,-1]}:=0$. Suppose that $k\ge 0$ and that $w^{\alpha,[g,k-1]}$ is already constructed. Let 
$$
\<\tau_{d_1}(e_{\alpha_1})\ldots\tau_{d_n}(e_{\alpha_n})\>^{\alpha,[g,k-1]}:=\left.\frac{\d^n w^{\alpha,[g,k-1]}(v^\top,v^\top_x,\ldots)}{\d t^{\alpha_1}_{d_1}\ldots\d t^{\alpha_n}_{d_n}}\right|_{x=t^*_*=0}.
$$
Define 
\begin{align}
&w^{\alpha,[g,k]}:=w^{\alpha,[g,k-1]}+\notag\\
&+\sum_{n\ge 0}\frac{\eps^{2g}}{n!}\sum_{d_1+\ldots+d_n=k}\underline{\left(\eta^{\alpha\mu}\<\tau_0(e_1)\tau_0(e_\mu)\prod_{i=1}^n\tau_{d_i}(e_{\alpha_i})\>_g-\<\prod_{i=1}^n\tau_{d_i}(e_{\alpha_i})\>^{\alpha,[g,k-1]}\right)}(v^1_x)^{2g-k}\prod_{i=1}^n v^{\alpha_i}_{d_i}.\label{eq:underlined difference}
\end{align}
Let us prove properties~\eqref{eq:construction of rational,eq1}-\eqref{eq:construction of rational,eq3}. We have
$$
O_\dil\left((w^\top)^{\alpha,[g]}-w^{\alpha,[g,k-1]}(v^\top,v^\top_x,\ldots)\right)=2g\left((w^\top)^{\alpha,[g]}-w^{\alpha,[g,k-1]}(v^\top,v^\top_x,\ldots)\right).
$$
Using~\eqref{eq:construction of rational,eq2} for~$w^{\alpha,[g,k-1]}$, we see that the underlined expression in~\eqref{eq:underlined difference} is equal to zero, if $\alpha_i=d_i=1$ for some $i$. Therefore, formula~\eqref{eq:construction of rational,eq1} is clear. Equation~\eqref{eq:construction of rational,eq3} follows from the fact that $O_\dil(v^\top)^\alpha_n=n(v^\top)^\alpha_n$. Property~\eqref{eq:construction of rational,eq2} follows from~\eqref{eq:property of wtop1}.

From~\eqref{eq:construction of rational,eq1} it follows that the limit $w^{\alpha,[g]}:=\lim_{k\to\infty}w^{\alpha,[g,k]}\in\cA^{\rt,[2g]}_{v^1,\ldots,v^N}$ is well-defined. Formula~\eqref{eq:construction of rational,eq2} implies that
$$
(v^\top)^\alpha+\sum_{g\ge 1}\eps^{2g}w^{\alpha,[g]}(v^\top,v^\top_x,\ldots)=(w^\top)^{\alpha,[g]}.
$$
Therefore, $w^\alpha(v^*_*;\eps)=v^\alpha+\sum_{g\ge 1}\eps^{2g}w^{\alpha,[g]}\in\hcA^{\rt,[0]}_{v^1,\ldots,v^N}$. The tameness of $w^\alpha(v^*_*;\eps)$ was already explained in Section~\ref{subsection:recall of DZ theory}.

It remains to show that $\frac{\d w^\alpha(v^*_*;\eps)}{\d v^1}=\delta^{\alpha,1}$. Let 
$$
O_\str:=\frac{\d}{\d t^1_0}-\sum_{n\ge 0}t^\gamma_{n+1}\frac{\d}{\d t^\gamma_n}.
$$
From the string equation for the potential~$F$ it follows that $O_\str(w^\top)^\alpha=O_{\str}(v^\top)^\alpha=\delta^{\alpha,1}$. Therefore, $\frac{\d w^\alpha(v^*_*;\eps)}{\d v^1}=\delta^{\alpha,1}$. The proposition is proved.
\end{proof}

\subsubsection{Final step}\label{subsubsection:final step}

Consider the rational Miura transformation $v^\alpha\mapsto w^\alpha(v^*_*;\eps)$ from the previous section. Since the variables $u^\alpha$ and $\tu^\alpha$ are related to $w^\alpha$ by Miura transformations, we see that they are related to the variables~$v^\alpha$ by rational Miura transformations, that we denote by~$u^\alpha(v^*_*;\eps)$ and~$\tu^\alpha(v^*_*;\eps)$ respectively. From equation~\eqref{eq:definition of the DZ operator} it follows that the operator $K^\DZ$ in the variables $u^\alpha$ is equal to
$$
(K^\DZ)^{\alpha\beta}_u=\sum_{p,q\ge 0}\frac{\d u^\alpha(v^*_*;\eps)}{\d v^\mu_p}\d_x^p\circ\eta^{\mu\nu}\d_x\circ(-\d_x)^q\circ\frac{\d u^\beta(v^*_*;\eps)}{\d v^\nu_q}.
$$
Lemma~\ref{lemma:rational Miura} implies that it is sufficient to show that 
\begin{gather}\label{eq:two sufficient conditions}
\frac{\d u^\alpha(v^*_*;\eps)}{\d v^1}=\delta^{\alpha,1}\quad\text{and}\quad u^\alpha(v^*_*;\eps)^\pol=v^\alpha.
\end{gather}
We have
$$
\tu^\alpha(v^\top,v^\top_x,\ldots;\eps)=\left.\eta^{\alpha\mu}\frac{\d^2 F^\red}{\d t^\mu_0\d t^1_0}\right|_{t^1_0\mapsto t^1_0+x}.
$$
The string equation for~$F^\red$ implies that $O_\str\tu^\alpha(v^\top,v^\top_x,\ldots;\eps)=\delta^{\alpha,1}$. Therefore, $\frac{\d\tu^\alpha(v^*_*;\eps)}{\d v^1}=\delta^{\alpha,1}$. From the string equation for~$F^\red$ and property~\eqref{eq:property of Fred} it follows that $\frac{\d\tu^\alpha(u^*_*,\eps)}{\d u^1}=\delta^{\alpha,1}$. Thus, $\frac{\d u^\alpha(v^*_*;\eps)}{\d v^1}=\delta^{\alpha,1}$.

Let us now prove the second equation in~\eqref{eq:two sufficient conditions}. Let
$$
\tu^\alpha(v^*_*;\eps)=v^\alpha+\sum_{g\ge 1}\eps^{2g}\sum_{k\ge -2g}\frac{P^\alpha_{g,k}(v^*_*)}{(v^1_x)^k},\quad P^\alpha_{g,k}\in\cA^{[2g+k]}_{v^1,\ldots,v^N}.
$$
Property~\eqref{eq:property of Fred} together with the string equation for $F^\red$ imply that
$$
\left.\Coef_{\eps^{2g}}\tu^\alpha(v^\top,v^\top_x,\ldots;\eps)\right|_{x=0}\in\mbC[[t^*_*]]^{(2g)}.
$$
Using also~\eqref{eq:property of wtop1}, we conclude that $P^\alpha_{g,k}=0$ for $k<0$ and
\begin{gather}\label{eq:P and red}
P^\alpha_{g,0}(v^*_*)=\sum_{n\ge 1}\sum_{d_1+\ldots+d_n=2g}\eta^{\alpha\mu}\<\tau_0(e_\mu)\tau_0(e_1)\prod\tau_{d_i}(e_{\alpha_i})\>^\red_g \frac{\prod v^{\alpha_i}_{d_i}}{n!}.
\end{gather}
Thus, 
$$
\tu^\alpha(u^*_*;\eps)=\left.\tu^\alpha(v^*_*;\eps)^\pol\right|_{v^\gamma_n=u^\gamma_n}.
$$
The rational Miura transformation $v^\alpha\mapsto u^\alpha(v^*_*;\eps)$ is the composition of the transformations \begin{gather}\label{eq:two transformations}
v^\alpha\mapsto\tu^\alpha(v^*_*;\eps)\quad\text{and}\quad\tu^\alpha\mapsto u^\alpha(\tu^*_*;\eps).
\end{gather}
We already know that $\frac{\d\tu^\alpha(v^*_*;\eps)}{\d v^1}=\delta^{\alpha,1}$. Equations~\eqref{eq:P and red},~\eqref{eq:property of Fred} and the string and the dilaton equations for $F^\red$ imply that $\frac{\d P^\alpha_{g,0}}{\d v^1_x}=0$. Therefore, $\frac{\d\tu^\alpha(v^*_*;\eps)^\pol}{\d v^1_x}=0$. So, the first transformation in~\eqref{eq:two transformations} satisfies the assumptions of Lemma~\ref{lemma:composition of rational Miura transformations}. Using Lemma~\ref{lemma:inverse} we see that the second transformation in~\eqref{eq:two transformations} also satisfies the assumptions of Lemma~\ref{lemma:composition of rational Miura transformations}. We conclude that $u^\alpha(v^*_*;\eps)^\pol=v^\alpha$. Theorem~\ref{theorem:Miura for DZ} is proved.

%%%%%%%%%%%%%%%%%%%%%%%%%%%%%%%%%%%%%%%%%%%%%%%%%%%%%%%%%%%%%%%%%%%%%%%
%%%%%%%%%%%%%%%%%%%%%%%%%%%%%%%%%%%%%%%%%%%%%%%%%%%%%%%%%%%%%%%%%%%%%%%

\section{Double ramification and Dubrovin-Zhang hierarchies of rank $1$}\label{section:DR and DZ rank 1}

In this section we focus on cohomological field theories of rank $1$, i.e.~ $\dim V=1$, and the corresponding double ramification and Dubrovin-Zhang hierarchies.

In Section~\ref{subsection:tau-symmetric deformations of the Riemann hierarchy} we recall certain definitions and the main conjecture from the work~\cite{DLYZ16} about tau-symmetric deformations of the Riemann hierarchy. In Section~\ref{subsection:DR hierarchy is standard} we show that the double ramification hierarchy is a standard deformation of the Riemann hierarchy in the sense of~\cite{DLYZ16}. In Section~\ref{subsection:standard for DZ} we prove an existence of a normal Miura transformation that reduces the Dubrovin-Zhang hierarchy to its unique standard form. This proves a part of the conjecture from~\cite{DLYZ16} for the Dubrovin-Zhang hierarchies. In Section~\ref{subsection:DR/DZ equivalence up to genus 5} we prove the strong DR/DZ equivalence conjecture at the approximation up to genus~$5$.

Since $\dim V=1$, a rank $1$ cohomological field theory is described by classes $c_{g,n}=c_{g,n}(e_1^{\otimes n})\in H^\even(\oM_{g,n},\mbC)$. Recall that, according to~\cite{Teleman}, rank~$1$ cohomological field theories with $\eta_{1,1}=\alpha$ are parameterized by numbers $s_1,s_2,\ldots$ in the following way:
\begin{gather}\label{eq:general rank 1 CohFT}
c_{g,n}=\alpha^{1-g}e^{-\sum_{i\ge 1}\frac{(2i)!}{B_{2i}}s_i\Ch_{2i-1}(\mathbb E)}.
\end{gather}
Here $\Ch_{2i-1}$ denotes the $(2i-1)$-th component of the Chern character and we use the same rescaling of the coefficient of $\Ch_{2i-1}(\mathbb E)$ in the exponent, as in~\cite[page 384]{DLYZ16}. Since $\dim V=1$, we will omit Greek indices in many notations. For example, the correlators of a cohomological field theory will be denoted by $\<\tau_{d_1}\ldots\tau_{d_n}\>_g$, the Hamiltonians of the Dubrovin-Zhang hierarchy will be denoted by $\oh_d^\DZ$, ... .

\subsection{Tau-symmetric deformations of the Riemann hierarchy}\label{subsection:tau-symmetric deformations of the Riemann hierarchy}

The Riemann hierarchy is the tau-symmetric Hamiltonian hierarchy given by the Hamiltonians 
$$
\oh^\mathrm{R}_{d}:=\int\frac{u^{d+2}}{(d+2)!}dx,\quad d\ge 0,
$$
the Hamiltonian operator $\d_x$ and the tau-symmetric densities $h^\mathrm{R}_d:=\frac{u^{d+2}}{(d+2)!}$, $d\ge -1$. 

A tau-symmetric deformation of the Riemann hierarchy is a tau-symmetric Hamiltonian hierarchy given by Hamiltonians $\oh_d$, $d\ge 0$, Hamiltonian operator $K$ and tau-symmetric densities~$h_d$, $d\ge -1$, such that
$$
\oh_d|_{\eps=0}=\oh^\mathrm{R}_d,\qquad K|_{\eps=0}=\d_x,\qquad h_d|_{\eps=0}=h^\mathrm{R}_d,\qquad K\frac{\delta\oh_0}{\delta u}=u_x.
$$
Here the last condition means that the Hamiltonian $\oh_0$ generates the spatial translations.

Denote by $\cP_n$ the set of all partitions of $n$. For a partition $\lambda=(\lambda_1,\ldots,\lambda_l)$, $\lambda_1\ge\ldots\lambda_l\ge 1$, let $l(\lambda):=l$. Introduce a subset $\cP'_n\subset\cP_n$ by
$$
\cP'_n:=\left\{\lambda\in\cP_n\left|
\begin{smallmatrix}
l(\lambda)\ge 2,\\
\lambda_1=\lambda_2,\\
\lambda_i\ge 2.
\end{smallmatrix}\right.
\right\}.
$$
For a partition $\lambda\in\cP_n$ let $u_\lambda:=\prod_{i=1}^{l(\lambda)}u_{\lambda_i}$. A tau-symmetric deformation of the Riemann hierarchy is said to be standard, if $K=\d_x$ and a density $\th_1$ for the Hamiltonian $\oh_1$ can be chosen in the following form:
\begin{gather}\label{eq:standard density}
\th_1=\frac{u^3}{6}-\frac{\eps^2}{24}a_0 u_x^2+\sum_{g\ge 2}\eps^{2g}\sum_{\lambda\in\cP'_{2g}}\alpha_\lambda u_\lambda,
\end{gather}
for some complex coefficients $a_0$ and $\alpha_\lambda$. It is easy to show that if such a density exists, then it is unique. In~\cite{DLYZ16} the authors proposed the following conjecture.

\begin{conjecture}\label{conjecture:DLYZ16}
Consider an arbitrary tau-symmetric deformation of the Riemann hierarchy.

\begin{enumerate}

\item[1.] Suppose that the deformation is standard. Then for the unique density of the form~\eqref{eq:standard density} we have the following.

\begin{itemize}
\item[a)] If $a_0=0$, then $\alpha_\lambda=0$ for all $\lambda$.
\item[b)] If $a_0\ne 0$, then all coefficients $\alpha_\lambda$ are uniquely determined by the coefficients $a_0$ and~$\alpha_{(2^g)}$, $g\ge 2$.
\end{itemize}

\item[2.] There exists a unique normal Miura transformation that transforms the hierarchy to a standard deformation. This deformation is called the standard form of the hierarchy.

\end{enumerate}
\end{conjecture}
\noindent
The authors of~\cite{DLYZ16} checked the uniqueness statement in the second part of the conjecture. Moreover they verified the conjecture at the approximation up to $\eps^{12}$.

Consider a cohomological field theory of rank $1$ with $\eta_{1,1}=1$. Clearly, the corresponding double ramification and Dubrovin-Zhang hierarchies are tau-symmetric deformations of the Riemann hierarchy. In the next section we will prove that the double ramification hierarchy is a standard deformation. In Section~\ref{subsection:standard for DZ} we will prove that part 2 of Conjecture~\ref{conjecture:DLYZ16} is true for the Dubrovin-Zhang hierarchy.

%%%%%%%%%%%%%%%%%%%%%%%%%%%%%%%%%%%%%%%%%%%%%%%%%%%%%%%%%%%%%%%%%%%%%%

\subsection{Double ramification hierarchy as a standard deformation}\label{subsection:DR hierarchy is standard}

Introduce a subset $\cP^\circ_n\subset\cP_n$ by
$$
\cP^\circ_n:=\left\{\lambda\in\cP_n\left|
\begin{smallmatrix}
l(\lambda)\ge 2,\\
\lambda_1=\lambda_2.
\end{smallmatrix}\right.
\right\}.
$$
\begin{lemma}\label{lemma:unique density}
Let $d\ge 2$. Consider a differential polynomial $h=\sum_{\lambda\in\cP_d}h_\lambda(u)u_\lambda\in\cA^{[d]}_u$, where~$h_\lambda(u)$ are formal power series in $u$.
\begin{enumerate}

\item[1.] For the local functional $\oh=\int h dx$ there exists a unique density $\th\in\cA^{[d]}_u$ of the form
\begin{gather}\label{eq:unique density}
\th=\sum_{\lambda\in\cP^\circ_d}\th_\lambda(u)u_\lambda,
\end{gather}
where $\th_\lambda(u)$ are formal power series in~$u$.

\item[2.] Let $d=2g$. Suppose that $\frac{\d h_\lambda(u)}{\d u}=0$ for all $\lambda$ and that $h_\lambda=0$ unless $\lambda_i\ge 2$. Then $\frac{\d\th_\lambda(u)}{\d u}=0$ for all~$\lambda$ and $\th_\lambda=0$ for $\lambda\in\cP_d\backslash\cP'_d$. Moreover, we have $\th_{(2^g)}=h_{(2^g)}$.

\end{enumerate}
\end{lemma}
\begin{proof}
1. Let us prove the existence of such density. Suppose that the set
\begin{gather}\label{eq:bad partitions}
\{\lambda\in\cP_d\backslash\cP^\circ_d|h_\lambda(u)\ne 0\}
\end{gather}
is non-empty. Let~$\lambda^{(0)}$ be the lexicographically maximal partition in the set~\eqref{eq:bad partitions} and~$m$ be the multiplicity of the part~$\lambda^{(0)}_1-1$ in~$\lambda^{(0)}$. Define a differential polynomial~$h^{(1)}$ by
\begin{gather*}
h^{(1)}:=h-\d_x\left(\frac{u_{\lambda^{(0)}_1-1}^{m+1}}{m+1}h_{\lambda^{(0)}}(u)\prod_{i=m+2}^{l(\lambda^{(0)})}u_{\lambda^{(0)}_i}\right)=\sum_{\lambda\in\cP_d}h^{(1)}_\lambda(u)u_\lambda.
\end{gather*}
Obviously, $h^{(1)}$ is a density for $\oh$. It is also clear that the lexicographically maximal partition in the set~$\{\lambda\in\cP_d\backslash\cP^\circ_d|h^{(1)}_\lambda(u)\ne 0\}$ is lexicographically smaller than $\lambda^{(0)}$. Continuing this process, after a finite number of steps, we come to a density of~$\oh$ of the form~\eqref{eq:unique density}. 

The uniqueness part follows from the fact that a non-zero differential polynomial of the form~\eqref{eq:unique density} does not belong to the image of the operator~$\d_x$.

Part 2 of the lemma is clear from the proof of part 1.
\end{proof}

\begin{proposition}\label{proposition:DR is standard}
Consider an arbitrary cohomological field theory of rank $1$ with $\eta_{1,1}=1$. Then we have the following.
\begin{enumerate}

\item[1.] The corresponding double ramification hierarchy is a standard tau-symmetric deformation of the Riemann hierarchy.

\item[2.] For the unique density $\tg_1$ for $\og_1$ of the form~\eqref{eq:standard density},
\begin{equation*}
\tg_1=\frac{u^3}{6}-\frac{\eps^2}{24}a^\DR_0 u_x^2+\sum_{g\ge 2}\eps^{2g}\sum_{\lambda\in\cP'_{2g}}\alpha^\DR_\lambda u_\lambda,
\end{equation*}
we have 
$$
a^\DR_0=1,\qquad \alpha^\DR_{(2^g)}=(3g-2)\int_{\oM_g}\lambda_g c_{g,0}.
$$

\end{enumerate}
\end{proposition}
\begin{proof}
We have
$$
\og_1=\sum_{g\ge 0,\,n\ge 2}\frac{(-\eps^2)^g}{n!}\sum_{a_1+\cdots+a_n=0}\left(\int_{\DR_g(0,a_1,\ldots,a_n)}\lambda_g\psi_1 c_{g,n+1}\right)\prod_{i=1}^n p_{a_i}.
$$
For $g\ge 1$ and $n\ge 2$ we have
\begin{gather}\label{eq:DR is standard,tmp1}
\int_{\DR_g(0,a_1,\ldots,a_n)}\lambda_g\psi_1 c_{g,n+1}=(2g-2+n)\int_{\DR_g(a_1,\ldots,a_n)}\lambda_g c_{g,n}.
\end{gather}
For $k\le n$ denote by $\pi_k\colon\oM_{g,n}\to\oM_{g,n-k}$ the forgetful map that forgets the last $k$ marked points. Using~\eqref{eq:DR and fundamental class}, we see that if $g=1$, then the right-hand side of~\eqref{eq:DR is standard,tmp1} is equal to
\begin{gather}\label{eq:rank 1,genus 1}
n\int_{\DR_1(a_1,\ldots,a_n)}\lambda_1 c_{1,n}=
\begin{cases}
0,&\text{if $n\ge 3$};\\
2a_1^2\int_{\oM_{1,1}}\lambda_1 c_{1,1}\stackrel{\text{by~\eqref{eq:general rank 1 CohFT}}}{=}\frac{a_1^2}{12},&\text{if $n=2$}.
\end{cases}
\end{gather}
Suppose $g\ge 2$. Then
\begin{gather}\label{eq:rank 1, higher g}
(2g-2+n)\int_{\DR_g(a_1,\ldots,a_n)}\lambda_g c_{g,n}=(2g-2+n)\int_{\pi_{n*}\DR_g(a_1,\ldots,a_n)}\lambda_g c_{g,0}.
\end{gather}
Note that the right-hand side is equal to zero unless $n\le g$. We also see that for $n=g$ the right-hand side of~\eqref{eq:rank 1, higher g} is equal to
\begin{gather}\label{eq:DR is standard,tmp2}
(3g-2)\int_{\pi_{g*}\DR_g(a_1,\ldots,a_g)}\lambda_g c_{g,0}=(3g-2)g!a_1^2\cdots a_g^2\int_{\oM_{g}}\lambda_g c_{g,0}.
\end{gather}
For an arbitrary $n\le g$ we write
$$
(2g-2+n)\int_{\pi_{n*}\DR_g(a_1,\ldots,a_n)}\lambda_g c_{g,0}=\frac{2g-2+n}{2g-2}\int_{\pi_{n*}\DR_g(0,a_1,\ldots,a_n)}\psi_1\lambda_g c_{g,1}.
$$
The divisibility property from Section~\ref{subsubsection:divisibility properties} implies that the integral $\int_{\pi_{n*}\DR_g(0,a_1,\ldots,a_n)}\psi_1\lambda_g c_{g,1}$ can be expressed as a polynomial 
$$
P(a_1,\ldots,a_n)=\sum_{d_1+\cdots+d_n=2g}P_{d_1,\ldots,d_n}a_1^{d_1}\cdots a_n^{d_n},\quad P_{d_1,\ldots,d_n}\in\mbC,
$$
where the coefficient $P_{d_1,\ldots,d_n}$ is equal to zero unless $d_i\ge 2$ for all $1\le i\le n$. Therefore, we obtain
$$
\og_1=\int\left(\frac{u^3}{6}-\frac{\eps^2}{24}u_x^2+\sum_{g\ge 2}\eps^{2g}\sum_{\lambda\in\cP_{2g}}\beta_\lambda u_\lambda\right)dx,
$$
for some constants $\beta_\lambda\in\mbC$ such that $\beta_\lambda=0$ unless $\lambda_i\ge 2$ for all $1\le i\le l(\lambda)$. Moreover, by~\eqref{eq:DR is standard,tmp2}, we have
\begin{gather*}
\beta_{(2^g)}=(3g-2)\int_{\oM_g}\lambda_g c_{g,0}.
\end{gather*}
Lemma~\ref{lemma:unique density} completes the proof of the proposition.
\end{proof}

We obtain the following formula for the constants $\alpha^\DR_{(2^g)}$ in terms of the parameters $s_i$ from~\eqref{eq:general rank 1 CohFT}:
\begin{gather}\label{eq:alpha-s relation}
\alpha^\DR_{(2^{g})}=(3g-2)\int_{\oM_g}\lambda_g e^{-\sum_{i\ge 1}\frac{(2i)!}{B_{2i}}s_i\Ch_{2i-1}(\mathbb E)}.
\end{gather}
In particular,
\begin{align}
\alpha^\DR_{(2^2)}=&-48s_1\int_{\oM_2}\lambda_2\lambda_1=-\frac{s_1}{120},\label{eq:correspondence1}\\
\alpha^\DR_{(2^3)}=&\left(-4032s_1^3-840s_2\right)\int_{\oM_3}\lambda_3\lambda_2\lambda_1=-\frac{s_1^3}{360}-\frac{s_2}{1728},\label{eq:correspondence2}\\
\alpha^\DR_{(2^4)}=&\left(-331776 s_1^5 - 172800 s_1^2 s_2 - 2520 s_3\right)\int_{\oM_4}\lambda_4\lambda_3\lambda_2=-\frac{2s_1^5}{525}-\frac{s_1^2s_2}{504}-\frac{s_3}{34560},\label{eq:correspondence3}\\
\alpha^\DR_{(2^5)}=&s_1^7\int_{\oM_5}\left(\frac{207028224}{35}\lambda_5\lambda_4\lambda_3-\frac{51757056}{5}\lambda_5\lambda_4\lambda_2\lambda_1\right)\label{eq:correspondence4}\\
&+s_1^4 s_2\int_{\oM_5}\left(10782720\lambda_5\lambda_4\lambda_3-10782720\lambda_5\lambda_4\lambda_2\lambda_1\right)\notag\\
&+s_1^2 s_3\int_{\oM_5}\left(943488\lambda_5\lambda_4\lambda_3-471744\lambda_5\lambda_4\lambda_2\lambda_1\right)\notag\\
&-s_4\int_{\oM_5}3120\lambda_5\lambda_4\lambda_3,\notag\\
&+s_1s_2^2\int_{\oM_5}\left(2246400\lambda_5\lambda_4\lambda_2\lambda_1-8985600 \lambda_5\lambda_4\lambda_3\right)=\notag\\
=&-\frac{754 s_1^7}{67375}-\frac{13 s_1^4s_2}{1320}-\frac{13 s_1^2 s_3}{52800}-\frac{13 s_4}{10644480}-\frac{13 s_1 s_2^2}{22176}.\notag
\end{align}
Here we use the formulas \cite{FP00,DLYZ16}
\begin{align*}
&\int_{\oM_g}\lambda_g\lambda_{g-1}\lambda_{g-2}=\frac{1}{2(2g-2)!}\frac{|B_{2g-2}|}{2g-2}\frac{|B_{2g}|}{2g},\quad g\ge 2,\\
&\int_{\oM_5}\lambda_5\lambda_4\lambda_2\lambda_1=\frac{1}{766402560}.
\end{align*}

%%%%%%%%%%%%%%%%%%%%%%%%%%%%%%%%%%%%%%%%%%%%%%%%%%%%%%%%%%%%%%%%%%%%%%%%%

\subsection{Standard form for the Dubrovin-Zhang hierarchy of rank $1$}\label{subsection:standard for DZ}

\begin{theorem}\label{theorem:standard form of DZ}
Consider a cohomological field theory of rank $1$ with $\eta_{1,1}=1$. Then part 2 of Conjecture~\ref{conjecture:DLYZ16} is true for the corresponding Dubrovin-Zhang hierarchy.
\end{theorem}
\begin{proof}
Consider the normal Miura transformation $w\mapsto\tu(w_*;\eps)$ and the Miura transformation $u\mapsto\tu(u_*;\eps)$ from Theorem~\ref{theorem:Miura for DZ}. From equation~\eqref{eq:property of Fred} and the string equation for~$F^\red$ it follows that $\<\tau_0^2\prod\tau_{d_i}\>^\red_g=0$, if $\sum d_i=2g$ and $g\ge 1$. Therefore, $\tu(u_*;\eps)=u$. By Theorem~\ref{theorem:Miura for DZ}, $K^\DZ_u=\d_x$. Let us prove that the Hamiltonian~$\oh_1^\DZ[u]$ has a density of the form~\eqref{eq:standard density}. Let 
$$
u^\red(x,t_*,\eps):=\left.\frac{\d^2 F^\red}{\d t_0^2}\right|_{t_0\mapsto t_0+x}.
$$
Denote by $h^\red_p\in\hcA^{[0]}_u$, $p\ge -1$, the tau-symmetric densities of the Dubrovin-Zhang hierarchy after the normal Miura transformation $w\mapsto u(w_*;\eps)$. The differential polynomial $h^\red_p$ is uniquely determined by the condition 
\begin{gather}\label{eq:equation for hpred}
h^\red_p(u^\red,u^\red_x,\ldots;\eps)=\left.\frac{\d^2 F^\red}{\d t_0\d t_{p+1}}\right|_{t_0\mapsto t_0+x}.
\end{gather}
The string equation for $F^\red$ implies that
$$
\frac{\d h^\red_p}{\d u}=h^\red_{p-1},\quad p\ge 0.
$$
Since $K^\DZ_u=\d_x$ and the Hamiltonian $\oh^\DZ_0[u]$ generates the spatial translations, we get $u_x=\d_x\frac{\delta\oh^\DZ_0[u]}{\delta u}$. Therefore,
$$
\oh^\DZ_0[u]=\int\frac{u^2}{2}dx.
$$
We obtain
$$
\frac{\d\oh^\DZ_1[u]}{\d u}=\oh_0^\DZ[u]=\int\frac{u^2}{2}dx.
$$
Therefore, 
$$
\oh_1^\DZ[u]=\int\left(\frac{u^3}{6}-\frac{\eps^2}{24}a_0u_x^2\right)dx+O(\eps^4)
$$
for some constant $a_0$.

\begin{lemma}\label{lemma:sufficient condition for good density}
Suppose $d\ge 4$ and $\oh\in\Lambda^{[d]}_u$. Then $\oh$ has a density $\th$ of the form
\begin{gather}\label{eq:good form of the density}
\th=\sum_{\lambda\in\cP'_d}C_\lambda u_\lambda,\quad C_\lambda\in\mbC,
\end{gather}
if and only if 
\begin{gather}\label{eq:sufficient condition for good density}
\frac{\d\oh}{\d u}=0\quad\text{and}\quad\frac{\d}{\d u_x}\frac{\delta\oh}{\delta u}=0.
\end{gather}
\end{lemma}
\begin{proof}
Obviously, if a density of the form~\eqref{eq:good form of the density} exists, then equations~\eqref{eq:sufficient condition for good density} are satisfied. Suppose now that the conditions~\eqref{eq:sufficient condition for good density} are true. Consider the unique density~$\th$ for~$\oh$ of the form~\eqref{eq:unique density}. The first condition in~\eqref{eq:sufficient condition for good density} immediately implies that $\frac{\d\th_\lambda(u)}{\d u}=0$. Then we compute
\begin{align*}
\frac{\d}{\d u_x}\frac{\delta\oh}{\delta u}=&\frac{\d}{\d u_x}\sum_{n\ge 0}(-\d_x)^n\frac{\d\th}{\d u_n}=-\sum_{n\ge 1}n(-\d_x)^{n-1}\frac{\d}{\d u_n}\frac{\d\th}{\d u}+\sum_{n\ge 0}(-\d_x)^n\frac{\d}{\d u_n}\frac{\d\th}{\d u_x}=\frac{\delta}{\delta u}\frac{\d\th}{\d u_x}.
\end{align*}
We obtain $\frac{\delta}{\delta u}\frac{\d\th}{\d u_x}=0$ and, therefore, $\frac{\d\th}{\d u_x}$ is $\d_x$-exact. Clearly, the differential polynomial $\frac{\d\th}{\d u_x}$ has the form~\eqref{eq:unique density}, so it can be $\d_x$-exact only if it is zero. Thus, $\th$ has the form~\eqref{eq:good form of the density} and the lemma is proved.
\end{proof}

We see that it remains to prove that $\frac{\d}{\d u_x}\frac{\delta\oh^\DZ_1[u]}{\delta u}=0$. We have (see~\cite[Section 3.7]{BDGR16}) $\frac{\delta\oh_1^\DZ[u]}{\delta u}=h^\red_0$. Let us prove that
\begin{gather}\label{eq:formula for h0red}
h^\red_0=\frac{u^2}{2}+\sum_{g,n\ge 1}\frac{\eps^{2g}}{n!}\sum_{d_1+\ldots+d_n=2g}\<\tau_0\tau_1\prod\tau_{d_i}\>^\red_g\prod u_{d_i}.
\end{gather}
From~\eqref{eq:property of Fred} and the string equation for $F^\red$ it follows that
$$
\left.u^\red_d\right|_{x=0}=t_d+\delta_{d,1}+\sum_{g\ge 0}\eps^{2g}R_{g,d}(t_*),
$$
where $R_{g,d}\in\mbC[[t_*]]^{(2g+d+1)}$. Denote the right-hand side of~\eqref{eq:formula for h0red} by $Q$. Using~\eqref{eq:equation for hpred}, we see that
$$
\left.\left(h_0^\red(u^\red,u^\red_x,\ldots;\eps)-Q(u^\red,u^\red_x,\ldots;\eps)\right)\right|_{x=0}=\sum_{g\ge 0}\eps^{2g}R_g(t_*),
$$
where $R_g\in\mbC[[t_*]]^{(2g+1)}$. The proof of equation~\eqref{eq:formula for h0red} is completed by the following lemma. 
\begin{lemma}
Suppose for a differential polynomial $P\in\hcA^{[0]}_u$ we have
\begin{gather}\label{eq:sufficient property to be zero}
\left.P(u^\red,u^\red_x,\ldots;\eps)\right|_{x=0}=\sum_{g\ge 0}\eps^g T_g(t_*),
\end{gather}
where $T_g\in\mbC[[t_*]]^{(g+1)}$. Then $P=0$.
\end{lemma}
\begin{proof}
Suppose that
$$
P(u_*;\eps)=\sum_{g\ge g_0}\eps^g P_g(u_*),\quad P_g\in\cA^{[g]}_u,\quad P_{g_0}\ne 0.
$$
Let 
$$
P_{g_0}(u_*)=\sum_{k=0}^{k_0}P_{g_0,k}(u_*)u_x^k,
$$
where $\frac{\d P_{g_0,k}}{\d u_x}=0$ and $P_{g_0,k_0}\ne 0$. Clearly, we have
\begin{gather}\label{eq:contradiction1}
\left.P(u^\red,u^\red_x,\ldots;\eps)\right|_{x=0}=\eps^{g_0}(P_{g_0,k_0}|_{u_d=t_d}+R(t_*))+O(\eps^{g_0+1}),
\end{gather}
where $R(t_*)\in\mbC[[t_*]]^{(g_0-k_0+1)}$. Since $P_{g_0,k_0}\ne 0$, we see that equation~\eqref{eq:contradiction1} contradicts~\eqref{eq:sufficient property to be zero}. Therefore, $P=0$ and the lemma is proved.
\end{proof}
Equations~\eqref{eq:equation for hpred},~\eqref{eq:property of Fred} and the string and the dilaton equations for~$F^\red$ imply that $\frac{\d h^\red_0}{\d u_x}=0$. Therefore, $\frac{\d}{\d u_x}\frac{\delta\oh^\DZ_1[u]}{\delta u}=0$ and the theorem is proved.
\end{proof}

%%%%%%%%%%%%%%%%%%%%%%%%%%%%%%%%%%%%%%%%%%%%%%%%%%%%%%%%%%%%%%%%%%%%%%%%%%%%%%%%%%%%%%%%%%

\subsection{Strong DR/DZ equivalence up to genus $5$}\label{subsection:DR/DZ equivalence up to genus 5}

In Section~\ref{subsubsection:strong DR/DZ} we recall a sufficient condition for the strong DR/DZ equivalence conjecture to be true. In Section~\ref{subsubsection:reduction to alpha=1} we consider a rank~$1$ cohomological field theory~\eqref{eq:general rank 1 CohFT} and show that the strong DR/DZ equivalence conjecture for general $\alpha$ follows from the case $\alpha=1$. Finally, in Section~\ref{subsubsection:proof up to genus 5} we prove the strong conjecture at the approximation up to genus~$5$.

\subsubsection{Sufficient condition for the strong DR/DZ equivalence conjecture}\label{subsubsection:strong DR/DZ}

Consider an arbitrary semisimple cohomological field theory, $c_{g,n}\colon V^{\otimes n}\to H^\even(\oM_{g,n},\mbC)$, where $\dim V=N$. Recall that by $\tu^\alpha(u^*_*;\eps)$ we denote the normal coordinates~\eqref{eq:normal coordinates for DR} for the double ramification hierarchy. Denote by $K^\DR_\tu$ the operator $\eta\d_x$ in the coordinates $\tu^\alpha$. In~\cite{BDGR16} we proved the following proposition.
\begin{proposition}[{\cite[Section 7.3]{BDGR16}}]\label{proposition:sufficient condition for strong}
Suppose that the Hamiltonians and the Hamiltonian operators of the double ramification hierarchy in the coordinates~$\tu^\alpha$ and the Dubrovin-Zhang hierarchy are related by a Miura transformation of the form
\begin{gather}\label{eq:sufficient form of a transformation}
\tu^\alpha\mapsto w^\alpha(\tu^*_*;\eps)=\tu^\alpha+\eta^{\alpha\mu}\d_x\left\{\cQ,\og_{\mu,0}[\tu]\right\}_{K^\DR_{\tu}},
\end{gather}
where $\cQ\in\hcA^{[-2]}_{\tu^1,\ldots,\tu^N}$ and $\frac{\d\cQ}{\d\tu^1}=\eps^2\<\tau_0(e_1)\>_1$. Then the strong DR/DZ equivalence conjecture is true.
\end{proposition}

\subsubsection{Reduction to the case $\alpha=1$}\label{subsubsection:reduction to alpha=1}

Consider a rank $1$ cohomological field theory~\eqref{eq:general rank 1 CohFT}. Then both potentials $F$ and $F^\red$ are power series in $t_0,t_1,\ldots$ and $\eps$ that additionally depend on the parameters $s_1,s_2,\ldots$ and $\alpha$. Define an operator $O$ by $O:=\alpha\frac{\d}{\d\alpha}+\frac{1}{2}\eps\frac{\d}{\d\eps}$. From Theorem~\ref{theorem:main geometric formula} we immediately see that 
\begin{gather}\label{eq:OFDR}
O F^\DR=F^\DR.
\end{gather}
Clearly, we have $O F=F$. Since $\eta^{1,1}=\frac{1}{\alpha}$, we get $O w^\top=0$. Then from the construction of the reduced potential $F^\red$ in~\cite[Section 7.3]{BDGR16} we can easily see that 
\begin{gather}\label{eq:OFred}
OF^\red=F^\red.
\end{gather}
Formulas~\eqref{eq:OFDR} and~\eqref{eq:OFred} imply that if $F^\DR$ and $F^\red$ are equal for $\alpha=1$, then they are equal for an arbitrary $\alpha$. Therefore, if the strong DR/DZ equivalence conjecture is true for $\alpha=1$, then it is true for an arbitrary $\alpha$.

\subsubsection{Proof of the equivalence up to genus $5$}\label{subsubsection:proof up to genus 5}

Consider a cohomological field theory~\eqref{eq:general rank 1 CohFT}. Let us prove the strong DR/DZ equivalence conjecture at the approximation up to genus~$5$. From the previous section we know that it is enough to consider the case $\alpha=1$. By Theorem~\ref{theorem:standard form of DZ}, the normal Miura transformation
$$
w\mapsto u(w_*;\eps)=w+\d_x^2\cP,
$$
transforms the Dubrovin-Zhang hierarchy to its standard form. We have the following formula for the unique density $\th_1$ for~$\oh_1^\DZ[u]$ of the form~\eqref{eq:standard density}:
\begin{equation}\label{eq:DZ standard density}
\begin{split}
\th_{1}=&\frac{u^3}{6}-\frac{\eps^2}{24}u_x^2-\frac{\eps^4}{120} s_1 u_{xx}^2-
\eps^6\left[\left(\frac{s_1^3}{360} +\frac{s_2}{1728}\right)u_{xx}^3+\frac{s_1^2}{420}u_{xxx}^2\right]\\
&-\eps^8\left[\left(\frac{2s_1^5}{525} +\frac{s_1^2 s_2}{504}+\frac{s_3}{34560}\right)u_{xx}^4+\left(\frac{11 s_1^4}{1400}+\frac{11 s_1 s_2}{6720}\right)u_{xxx}^2u_{xx}+\left(\frac{s_1^3}{1260}+\frac{s_2}{60480}\right)u_{xxxx}^2\right]\\
&-\eps^{10}\left[\left(\frac{754 s_1^7}{67375}+\frac{13 s_2 s_1^4}{1320}+\frac{13 s_3 s_1^2}{52800}+\frac{13 s_2^2 s_1}{22176}+\frac{13 s_4}{10644480}\right)u_{xx}^5\right.\\
&\hspace{1.2cm}+\left(\frac{58 s_1^6}{1375}+\frac{7s_2s_1^3}{330}+\frac{7 s_3 s_1}{26400}+\frac{s_2^2}{3168}\right)u_{xxx}^2 u_{xx}^2\\
&\hspace{1.2cm}\left.+\left(\frac{71 s_1^5}{12600}+\frac{s_1^2s_2}{756}+\frac{s_3}{276480}\right)u_{xxxx}^2u_{xx}+\left(\frac{s_1^4}{3465}+\frac{s_2 s_1}{66528}\right)u_{xxxxx}^2\right]\\
&+O(\eps^{12}).
\end{split}
\end{equation}
This formula is given in~\cite[page 433]{DLYZ16} at the approximation up to genus~$4$, and we are grateful to the authors of~\cite{DLYZ16} for providing us a software that computes the density~$\th_1$ at the approximation up to genus~$5$. We see here that $a_0=1$ and 
\begin{align*}
\alpha_{(2^2)}=&-\frac{s_1}{120},\\
\alpha_{(2^3)}=&-\frac{s_1^3}{360}-\frac{s_2}{1728},\\
\alpha_{(2^4)}=&-\frac{2s_1^5}{525}-\frac{s_1^2 s_2}{504}-\frac{s_3}{34560},\\
\alpha_{(2^5)}=&-\frac{754 s_1^7}{67375}-\frac{13 s_2 s_1^4}{1320}-\frac{13 s_3 s_1^2}{52800}-\frac{13 s_2^2 s_1}{22176}-\frac{13 s_4}{10644480}.
\end{align*}
From equations~\eqref{eq:correspondence1}--\eqref{eq:correspondence4} we see that $\alpha_{(2^g)}=\alpha^\DR_{(2^g)}$ for $g=2,3,4,5$. Since Conjecture~\ref{conjecture:DLYZ16} is true at the approximation up to~$\eps^{10}$, we obtain that the standard form of the Dubrovin-Zhang hierarchy coincides with the double ramification hierarchy up to genus~$5$. Note that $\tu(u_*;\eps)=\frac{\delta\og_0}{\delta u}=u$. We have 
$$
\left.(F^\red-F)\right|_{t_0\mapsto t_0+x}=\cP(w^\top,w^\top_x,\ldots;\eps).
$$
From the string equations for $F^\red$ and $F$ it follows that $\frac{\d\cP}{\d w^1}=-\eps^2\<\tau_0\>_1$. Then it is easy to see that the Miura transformation $u\mapsto w(u_*;\eps)$ has the form
$$
w(u_*;\eps)=u+\d_x^2\cQ,
$$
where $\frac{\d\cQ}{\d u^1}=\eps^2\<\tau_0\>_1$. Therefore, the sufficient condition from Proposition~\ref{proposition:sufficient condition for strong} is satisfied and we conclude that the strong DR/DZ equivalence conjecture is true at the approximation up to genus~$5$.

\end{document}